%% file: main.tex
\documentclass{JFM-FLM_Au}

\usepackage{amsmath,amssymb,bm,mathtools}
\usepackage{booktabs,graphicx,array,capt-of}
\usepackage{pdflscape}
\usepackage[table]{xcolor}
\usepackage{framed,caption}

\definecolor{shadecolor}{rgb}{1,1,1}

\DeclareCaptionFormat{revisioncaption}{\colorbox{white}{\parbox{\dimexpr\linewidth-2\fboxsep\relax}{#1#2#3}}}

\hypersetup{
  pdftitle={Competing contributions of stress and fourth-order moments to anti-Fourier heat flow in a rarefied cavity},
  pdfauthor={Ehsan Roohi},
  colorlinks=true,
  linkcolor=blue,
  citecolor=blue,
  urlcolor=blue
}
\newcommand{\doilink}[1]{\href{https://doi.org/#1}{\nolinkurl{doi:#1}}}

\newcommand{\cl}{\mathrm{cl}}
\newcommand{\Kn}{\mathit{Kn}}
\newcommand{\dd}{\mathrm{d}}

\newcommand{\Aten}{\mathsfbi{A}}

\newcommand{\qv}{\boldsymbol{q}}
\newcommand{\cv}{\boldsymbol{c}}
\newcommand{\uv}{\boldsymbol{u}}
\newcommand{\vvec}{\boldsymbol{v}}
\newcommand{\xv}{\boldsymbol{x}}

\newcommand{\Grad}{\boldsymbol{\nabla}}

\newcommand{\tr}{\mathrm{tr}}

\makeatletter
\def\oddabsfooterflag{\hbox to \textwidth{\hfill{\cppagefont\thepage}}}
\def\evenabsfooterflag{\hbox to \textwidth{{\cppagefont\thepage}\hfill}}
\def\pagelimitfooter{\hbox to \textwidth{\hfill{\cppagefont\thepage}}}
\def\ps@headings{\let\@mkboth\markboth
  \def\@oddhead{\hfill{\itshape\@righttitle}\hfill}%
  \def\@evenhead{\hfill\itshape\@lefttitle\hfill}%
  \def\@oddfoot{\hbox to \textwidth{\hfill{\cppagefont\thepage}}}%
  \def\@evenfoot{\hbox to \textwidth{{\cppagefont\thepage}\hfill}}%
  \def\sectionmark##1{\markboth{##1}{}}%
  \def\subsectionmark##1{\markright{##1}}}
\makeatother

\newtheorem{proposition}{Proposition}

\newenvironment{landscapefigurepage}
  {\clearpage\begin{landscape}\thispagestyle{headings}\centering\vspace*{\fill}}
  {\vspace*{\fill}\end{landscape}\clearpage}

\lefttitle{E. Roohi}
\righttitle{Competing contributions of stress and fourth-order moments}

\title{Competing contributions of stress and fourth-order moments to anti-Fourier heat flow in a rarefied cavity}

\author{Ehsan Roohi\aff{1}}

\affiliation{\aff{1}Department of Mechanical and Industrial Engineering, University of Massachusetts Amherst,\\160 Governors Drive, Amherst, MA 01003, USA}

\corresau{Ehsan Roohi, eroohi@umass.edu}

\makeatletter
\renewcommand{\maketitle}{%
  \thispagestyle{headings}%
  \par\begingroup\centering
  \vspace*{10pt}%
  {\Large\bfseries\@title\par}\medskip
  {\bfseries\@author\par}\smallskip
  {\small\@affiliation\par}%
  \endgroup\medskip}
\renewenvironment{abstract}{\par\noindent\textbf{Abstract}\par\small}{\par\medskip}
\makeatother

\begin{document}
\maketitle
\pagestyle{headings}
\raggedbottom

\begin{abstract}

\begin{shaded}
In a rarefied gas the heat flux can point towards warmer gas even when all walls have the same temperature. We examine which molecular velocity moments account for this behaviour in a lid-driven argon cavity. We compare the regularised 13- and 26-moment equations (R13 and R26) with direct simulation Monte Carlo (DSMC) at Knudsen numbers 0.05 and 0.2, with lid speed \(U_w=100\,\mathrm{m\,s^{-1}}\) (lid Mach number \(0.31\)). Over the central square at Knudsen number 0.2, the normalised stress projection along the DSMC heat flux is about 2.3. The tensorial fourth-order contribution opposes it and offsets about half. R26 predicts both projections within 15\%. Projected on its own heat flux, R13 gives about half of the opposing fourth-order contribution ($-0.63$ against $-1.18$), and its Fourier term is too large and partly compensates for that error. Projected on the DSMC heat flux, its total is close ($-1.10$ against $-1.18$), but it is too small in the interior and too large near the side walls. The fourth-order tensor enters the planar heat-flux balance through its divergence. At $\Kn_{\rm Gu}=0.20$, 94--95\% of the squared in-plane difference between the R26 and DSMC tensors lies in the null space of the discrete divergence. The higher-moment equations and wall conditions can still couple these components to the flow. We also evaluate the R26 closures directly on moments through fifth order measured in seven independent DSMC runs at each Knudsen number. With central differences on the sampling grid, the discrepancies are 1.7 to 5 times the estimated sampling-noise floor. They also depend on the derivative approximation and molecular collision model. Measuring the individual moments reveals compensating errors in R13 and tests components of the R26 fourth moment that its planar divergence leaves undetermined.
\end{shaded}

\end{abstract}

\section{Introduction}
\label{sec:intro}

In a rarefied gas, heat can flow locally from a colder region towards a warmer one even when the walls are isothermal. Examples include lid-driven cavities, triangular cavities and backward-facing micro-steps. Kinetic simulations have documented this behaviour in rarefied cavities~\citep{JohnGuEmerson2010,JohnGuEmerson2011,MohammadzadehRoohi2012}. Nonlinear thermal stresses and stress gradients help explain it. Sone's weakly nonlinear asymptotic theory \citep{Sone2007} has also been used to interpret these flows \citep{mahdavi2015investigation,balaj2017regulation,MahdaviRoohi2022}. Direct simulation Monte Carlo (DSMC) calculations show that cold-to-hot transfer depends on lid speed and rarefaction and remains consistent with the second law when entropy is evaluated from the molecular distribution \citep{MohammadzadehRoohi2012}. In a backward-facing micro-step, for example, heat lines bend towards the warmer inlet over a wide range of rarefaction. In Sone's heat-flow expansion this deflection results from competition between Fourier-like transport and higher-order terms involving the curvature of the velocity field \citep{MahdaviRoohi2022}.

\begin{shaded}
\citet{VenugopalPraturiGirimaji2019} simulated argon in lid-driven cavities of several shapes with the unified gas-kinetic scheme and DSMC over a range of Knudsen and Mach numbers. They found heat flux with a component along the temperature gradient, and beyond moderate Knudsen and Mach numbers the heat transport satisfied the entropy statement based on the Boltzmann $H$-function but not the one based on the Sackur--Tetrode entropy. \citet{AristovZabelokFrolova2025} review such non-classical transport in rarefied gas flows, including thermal slip, a heat flux along the walls of Couette flow without a temperature gradient along them, and heat flux against the direction expected from Fourier's law.
\end{shaded}

If a model predicts the direction of the anti-Fourier heat flux, which of its higher moments does that observation test? In Grad-type and regularized moment hierarchies, knowing the density, velocity, temperature, stress and heat flux leaves part of the heat-flux balance undetermined. That balance also contains fourth-order moments of the molecular velocity distribution. The regularized 26-moment (R26) system has evolution equations for the third-order tensor \(m_{ijk}\), the fourth-order tensor \(R_{ij}\) and the scalar fourth-order excess \(\Delta\) \citep{Grad1949,Struchtrup2005,StruchtrupTorrilhon2003,StruchtrupTorrilhon2007,GuEmerson2009}. The way these moments enter the heat-flux balance limits what we can learn about them from a measured heat flux.

\begin{shaded}
In a monatomic normal shock, the fourth-order contribution depends on a single combination, \(S=R^{\mathrm{cl}}_{xx}+\Delta/3\) \citep{RoohiShock2026}. An independent measurement of the scalar excess then suffices to separate the tensor and scalar contributions. In a two-dimensional cavity, the in-plane heat-flux equations involve the composite tensor
\begin{equation}
A_{ij}=R^{\mathrm{cl}}_{ij}+\frac{1}{3}\Delta\delta_{ij},
\qquad i,j\in\{x,y\}.
\label{eq:Adef_intro}
\end{equation}
Here \(R^{\mathrm{cl}}_{ij}\) is the contracted traceless fourth-order contribution,
\(\Delta\) is the scalar fourth-order excess, and \(\delta_{ij}\) is the Kronecker delta. Equation~\eqref{eq:RDelta_def} gives the moment definitions of \(R^{\cl}_{ij}\) and \(\Delta\). Equation~\eqref{eq:Adef_intro} groups terms already present in the heat-flux balance and requires no closure assumption. Only the divergence \(\partial A_{ij}/\partial x_j\) enters the in-plane equations. We can therefore add any symmetric, divergence-free tensor field to \(A_{ij}\) without altering this term. This familiar property leaves an entire space of tensor fields undetermined by the planar heat-flux balance. In two-dimensional flow with three velocity components and \(\partial/\partial z=0\), \(A_{zz}\) can also change the division between tensor and scalar contributions without appearing in that balance. Other equations of the complete R26 system can constrain these changes.
\end{shaded}

The regularized 13-moment (R13) system gives smooth shock profiles in cases where the unregularized Grad system develops subshocks. Tests of R13 span shock structure, two-dimensional bulk flow and wall-bounded microflows \citep{TorrilhonStruchtrup2004,Torrilhon2006,GuEmerson2007,TaheriTorrilhonStruchtrup2009,Torrilhon2016}. Regularization recovers non-equilibrium effects absent from Navier--Stokes--Fourier theory, including Knudsen layers, tangential heat flux and non-monotone temperature profiles. Its accuracy depends on the quantity compared, the degree of rarefaction, the wall model and the molecular interaction law \citep{GuBarberEmerson2007,CaiWang2020}.

R26 resolves several Knudsen-layer modes absent from R13. These modes improve predictions in Kramers flow and thermal-jump problems \citep{GuEmersonTang2010,GuEmerson2014}. For low-speed cylindrical Couette flow, \citet{GuEmerson2014Cyl} used the linearised R13 and R26 equations between rotating cylinders. For flow past a circular cylinder, \citet{GuBarberJohnEmerson2019} examined rarefaction and compressibility over a range of free-stream Mach numbers extending to 0.865. In the cavity the relevant imposed speed is that of the moving wall. The R13 heat-flux benchmark of \citet{RanaTorrilhonStruchtrup2013} used a \(50\,\mathrm{m\,s^{-1}}\) lid. At the same lid speed, \citet{MohammadzadehRanaStruchtrup2016} compared R13 and DSMC with three adiabatic stationary walls. The R26 velocity-profile benchmark of \citet{YangTangYang2019} and the cavity calculations of \citet{GuGreenbankMoulinecEmerson2024} used \(10\,\mathrm{m\,s^{-1}}\).

An earlier R26 calculation by \citet[figure~7]{GuJohnTangEmerson2009} reached \(100\,\mathrm{m\,s^{-1}}\) at their \(\Kn=0.2\). Both R26 and DSMC predicted cold-to-hot transfer, yet R26 gave a recirculating heat-flux pattern in the lower half of the cavity where DSMC showed a stagnation point near the lower left corner. What accounts for agreement in the direction of heat transfer when the flux fields differ? We address this question at \(U_w=100\,\mathrm{m\,s^{-1}}\) and two degrees of rarefaction by resolving the terms that carry heat and measuring the higher moments that enter them.

\begin{shaded}
One comparison tests the constitutive relations without first solving the R26 cavity equations. We insert the lower moments measured by DSMC into the published R26 relations, then compare their predictions for $\Phi$, $\psi$ and $\Omega$ with the fourth- and fifth-order moments measured in the same gas. At $\Kn_{\rm Gu}=0.20$ the relative discrepancies are 2.29, 2.04 and 4.81 with central differences on the sampling grid, and 2.07, 1.76 and 4.22 with differences over two cell spacings (\S\ref{sec:updated_evidence}). We estimate the effect of sampling by applying each stencil to an interpolated R26 field for which the closure holds by construction, with noise taken from the DSMC realisations. Sampling alone produces 31--58\% of the measured discrepancies with the narrower stencil and 17--33\% with the wider one. At $\Kn_{\rm Gu}=0.05$ the discrepancies are about half as large and still exceed the estimated noise. Moment truncation, differentiation, the distinct Maxwell and VSS collision laws and DSMC discretisation can all contribute to the remaining difference. The moment fields used in this comparison accompany the paper.

A complementary comparison starts from the computed R26 tensor and asks how much of its difference from DSMC enters the heat-flux equation. At $\Kn_{\rm Gu}=0.20$, \mbox{94--95\%} of the squared in-plane difference lies in the null space of the discrete divergence. This fraction ranges from 0.941 to 0.954 for the R26 fields considered, including coarse-grid solutions of the complete discrete equations (\S\ref{sec:kn020_projection}). For the $28\times28$ R26 field, the in-plane fraction ranges from 0.931 to 0.971 as the moment correction, wall exclusion, smoothing and divergence approximation vary (\S\ref{sec:kn020_projection}).

The remaining equations can distinguish fields that give the same fourth-order contribution to the heat flux. We perturb the interior R26 moments using, in turn, the out-of-plane and Airy fields of \S\ref{sec:nullspace}. The first changes the planar heat-flux equations only at round-off level. The response to the second is 0.55\% of that to a perturbation with nonzero divergence. Both perturbations change the third- and fourth-order moment equations (appendix~\ref{app:complete_response}). These prescribed perturbations probe the equations independently of the R26--DSMC difference. Figure~\ref{fig:kn005_primary} compares the cavity fields at $\Kn_{\rm Gu}=0.05$, and \S\ref{sec:updated_evidence} evaluates the closures directly on DSMC moments. Appendices~\ref{app:same_state}--\ref{app:numerical_modes} describe the sampling windows, moment estimators, quadratic entropy and numerical tests.
\end{shaded}

The R13 result at $\Kn_{\rm Gu}=0.20$ illustrates what a comparison of total heat flux can miss. Its Fourier contribution is too strong and its opposing tensorial contribution too weak. The two errors partly cancel in the total heat flux. For a moment model or a data-driven surrogate, direct measurements of the higher moments and checks on their evolution equations and wall conditions can reveal errors that this cancellation conceals.

\section{What can be tested with heat-flux measurements?}
\label{sec:scope}

Wall conditions form part of a moment closure. The R13 analysis of \citet{TorrilhonStruchtrup2008} illustrates why a comparison in the bulk alone cannot establish that a model describes a wall-bounded flow.

We distinguish four questions about heat transport in the cavity (table~\ref{tab:validation_hierarchy}). Where does the gas conduct heat towards a higher temperature? Do the models place the counter-gradient region in the same location as DSMC and reproduce its shape and connectivity? How well do they predict the magnitude and direction of \(\qv\)? Finally, do they predict the fourth moments responsible for the heat flux? The last comparison needs both the divergence \(\Grad\cdot\Aten\) and measurements of the individual moments. In interpreting all four comparisons we account for numerical resolution, the wall conditions and differences between the molecular collision models.

\begin{table}
\centering
\small
\setlength{\tabcolsep}{5pt}
\renewcommand{\arraystretch}{1.12}
\begin{tabular}{p{0.16\textwidth}p{0.27\textwidth}p{0.46\textwidth}}
\toprule
quantity tested & measurement & what the comparison establishes \\
\midrule
occurrence & region with resolved \(\qv\cdot\Grad T>0\) & identifies counter-gradient transport for the stated signal and minimum region size \\
\specialrule{0.25pt}{1.8pt}{1.8pt}
topology & Jaccard/Dice overlap, centroid, connectivity and boundary distance on a common domain & compares the location, shape and extent of the counter-gradient region \\
\specialrule{0.25pt}{1.8pt}{1.8pt}
heat-flux vector & relative vector root-mean-square (RMS) error, vector correlation, magnitude ratio and angular error & compares heat-flux direction and magnitude over the same cells \\
\specialrule{0.25pt}{1.8pt}{1.8pt}
fourth-order closure & \(\Grad\cdot\Aten\), its tensor and scalar contributions, and sampled \(R^{\cl}_{ij}\), \(\Delta\) and \(m_{ijk}\) & compares individual moments and their contribution to the heat-flux balance. Testing the entire R26 field also requires the remaining balances and wall conditions, numerical convergence, sampling uncertainties and consistent collision models \\
\bottomrule
\end{tabular}
\captionsetup{format=revisioncaption}
\caption{Tests of heat transport and moment closure. The first three use the heat flux. The fourth also requires higher moments because the heat-flux balance contains only the divergence of the composite tensor. The higher-moment evolution equations and wall conditions supply additional constraints.}
\label{tab:validation_hierarchy}
\end{table}

\begin{shaded}
The order of the tests in table~\ref{tab:validation_hierarchy} reflects the information each requires. Their sensitivity to error can follow a different order. Near a zero of \(\qv\cdot\Grad T\), for example, a small vector error can shift the boundary of the counter-gradient region. Below we perturb individual moments to determine which changes leave the heat-flux balance invariant. Whether the perturbed moments can occur in a physical flow also depends on the other equations and the walls. For the nonlinear R26 calculations, appendices~\ref{app:solver_reconciliation} and \ref{app:complete_continuation} give the residuals on each grid.
\end{shaded}

\section{Fourth-order moments in the heat-flux balance}
\label{sec:theory}

\subsection{Cold-to-hot transfer in a sheared gas}

For a slightly rarefied gas, Sone expands the heat-flow vector \(\boldsymbol Q\) in powers of the Knudsen number. We follow the notation of our micro-step analysis, with \(\bar{\Kn}\) the small asymptotic parameter and \(\boldsymbol Q_n\) the coefficient of order \(n\). The first three terms are
\begin{equation}
    \boldsymbol{Q}=\bar{\Kn}\boldsymbol{Q}_1+\bar{\Kn}^2\boldsymbol{Q}_2+\bar{\Kn}^3\boldsymbol{Q}_3,
\label{eq:sone_expansion}
\end{equation}
where \(\boldsymbol{Q}_1=0\) and \(\tau\) is the nondimensional temperature perturbation in Sone's scaling. In \eqref{eq:q2}--\eqref{eq:q3}, \(\tau\) and \(u_i\) denote the temperature perturbation and velocity at the order appropriate to each term. We omit their order subscripts and refer to \citet{Sone2007} and \citet{MahdaviRoohi2022} for the fully indexed expressions. At second order the Fourier-like contribution is
\begin{equation}
Q_{2i}=-\frac{5}{4}\gamma_2\frac{\partial \tau}{\partial x_i},
\label{eq:q2}
\end{equation}
whereas the third-order heat flux also depends on the product of the temperature perturbation and its gradient and on the curvature of the velocity field,
\begin{equation}
Q_{3i}= -\frac{5}{4}\gamma_2\frac{\partial \tau}{\partial x_i}
    -\frac{5}{4}\gamma_5\tau\frac{\partial \tau}{\partial x_i}
    +\frac{1}{2}\gamma_3\frac{\partial^2 u_i}{\partial x_j^2}+\cdots .
\label{eq:q3}
\end{equation}
Sone's transport coefficients \(\gamma_2\), \(\gamma_3\) and \(\gamma_5\) depend on the molecular model and are tabulated in \citet{Sone2007}. In a step expansion or a cavity shear layer, the velocity-curvature term can compete with the Fourier-like term and turn the heat lines towards warmer gas \citep{MahdaviRoohi2022}. This mechanism explains how cold-to-hot transfer can occur. We ask which higher moments can be inferred from the resulting heat flux.

\subsection{The planar heat-flux equations}
\label{sec:q_observable}

Let \(f=f(\xv,\vvec,t)\) be the mass-normalised molecular distribution, with \(\rho=\int f\,\dd\vvec\), \(\rho u_i=\int v_i f\,\dd\vvec\), and \(3\rho T=\int|\cv|^2f\,\dd\vvec\).  Here \(\xv\), \(\vvec\) and \(t\) are position, molecular velocity and time, \(\uv\) is the bulk velocity, \(T\) is the thermal-temperature variable in specific-energy units, and \(\cv=\vvec-\uv\) is the peculiar velocity.  The deviatoric stress is \(\sigma_{ij}=\int c_{\langle i}c_{j\rangle}f\,\dd\vvec\), where angular brackets denote the symmetric trace-free part; repeated Cartesian indices are summed.  The heat flux is the third-order moment
\begin{equation}
q_i=\frac{1}{2}\int |\cv|^2 c_i f\,\dd\vvec.
\label{eq:qdef}
\end{equation}
The transport of \(q_i\) involves the contracted fourth moment
\begin{equation}
B_{ij}=\int |\cv|^2 c_i c_j f\,\dd\vvec .
\label{eq:Bdef}
\end{equation}
This moment arises directly from the Boltzmann equation, before any closure is imposed. Let \(M_{ijk}=\int c_i c_j c_k f\,\dd\vvec\) be the raw third central moment, \(p=\rho T\) the hydrostatic pressure, \(D_t=\partial_t+u_j\partial_j\) the material derivative, and
\begin{equation}
\mathcal{C}_i=\frac12\int |\cv|^2c_i\,\mathcal{Q}(f,f)\,\dd\vvec
\label{eq:collision_q}
\end{equation}
be the collisional production of heat flux, where \(\mathcal Q(f,f)\) is the Boltzmann collision operator. We write the distribution in terms of \(\cv=\vvec-\uv(\xv,t)\). Derivatives at fixed \(\vvec\) then act on \(\uv\) through \(\cv\), and \(v_k=u_k+c_k\) splits the streaming term. We multiply the equation by an arbitrary weight \(\varphi(\cv)\), integrate over velocity and integrate the \(\partial/\partial c_k\) terms by parts. The resulting central-moment identity is
\begin{equation}
\partial_t\langle\varphi\rangle
+\partial_k\!\bigl(u_k\langle\varphi\rangle+\langle c_k\varphi\rangle\bigr)
+\bigl(D_tu_k\bigr)\!\left\langle\frac{\partial\varphi}{\partial c_k}\right\rangle
+\bigl(\partial_ju_k\bigr)\!\left\langle c_j\frac{\partial\varphi}{\partial c_k}\right\rangle
=\int\varphi\,\mathcal{Q}\,\dd\vvec,
\label{eq:central_transport}
\end{equation}
where \(\langle\psi\rangle=\int\psi f\,\dd\vvec\). The dependence of \(\cv\) on \(\uv\) gives the acceleration term and a term proportional to \(\partial_ju_k\). Integration by parts also gives a contribution proportional to \(\delta_{jk}\langle\varphi\rangle\), which puts the advection in conservative form, \(\partial_k(u_k\langle\varphi\rangle)\). The choices \(\varphi=1\), \(c_i\) and \(\frac12|\cv|^2\) in \eqref{eq:central_transport} recover the continuity, momentum and energy balances, respectively, and provide a check on the signs.

For heat flux, \(\varphi=\frac12|\cv|^2c_i\), so \(\langle\varphi\rangle=q_i\) and \(\langle c_j\varphi\rangle=\frac12B_{ij}\). Since \(\partial\varphi/\partial c_k=c_ic_k+\frac12|\cv|^2\delta_{ik}\), the definitions of \(p\), \(\sigma_{ij}\), \(M_{ijk}\) and \(q_i\) give the two moments in \eqref{eq:central_transport} as
\begin{equation}
\left\langle\frac{\partial\varphi}{\partial c_k}\right\rangle
=\sigma_{ik}+\frac52\,p\,\delta_{ik},
\qquad
\left\langle c_j\frac{\partial\varphi}{\partial c_k}\right\rangle
=M_{ijk}+q_j\delta_{ik}.
\label{eq:phi_moments}
\end{equation}
Substituting \eqref{eq:phi_moments} into \eqref{eq:central_transport} gives
\begin{equation}
\begin{split}
\partial_t q_i+\partial_j\!\left(u_jq_i+\frac12B_{ij}\right)
&+\left(\sigma_{ik}+\frac52p\delta_{ik}\right)D_tu_k\\
&+\left(M_{ijk}+\delta_{ik}q_j\right)\partial_ju_k
=\mathcal{C}_i .
\end{split}
\label{eq:exact_q_balance}
\end{equation}
Equations~\eqref{eq:central_transport} and \eqref{eq:phi_moments} give \eqref{eq:exact_q_balance} without any closure assumption. The acceleration, velocity-gradient, advection and collision terms involve third- or lower-order couplings. The only fourth-order term is \(\frac12\partial_jB_{ij}\), where the factor \(1/2\) comes from the definition of heat flux.

To express \(B_{ij}\) in R26 variables, we split the dyad into its symmetric trace-free and isotropic parts, \(c_ic_j=c_{\langle i}c_{j\rangle}+\frac13|\cv|^2\delta_{ij}\), where the angular brackets denote the trace-free part. This gives
\begin{equation}
B_{ij}=\int|\cv|^2c_{\langle i}c_{j\rangle}f\,\dd\vvec
+\frac13\,\delta_{ij}\int|\cv|^4f\,\dd\vvec .
\label{eq:Bsplit}
\end{equation}
The regularized moment hierarchy \citep{Struchtrup2005,GuEmerson2009} defines two fourth-order closure variables by subtracting the near-equilibrium part of each integral,
\begin{equation}
R^{\cl}_{ij}\equiv\int|\cv|^2c_{\langle i}c_{j\rangle}f\,\dd\vvec-7T\sigma_{ij},
\qquad
\Delta\equiv\int|\cv|^4f\,\dd\vvec-15\rho T^2 .
\label{eq:RDelta_def}
\end{equation}
The superscript ``\(\cl\)'' distinguishes this contracted deviatoric moment from the raw \(B_{ij}\).  For a local Maxwellian, \(\int|\cv|^4f\,\dd\vvec=15\rho T^2\) and the trace-free integral vanishes, so \(R^{\cl}_{ij}=\Delta=0\); for a Grad 13-moment distribution the trace-free integral equals \(7T\sigma_{ij}\), so \(R^{\cl}_{ij}\) and \(\Delta\) measure precisely the fourth-order content beyond the 13-moment description.  Substituting \eqref{eq:RDelta_def} into \eqref{eq:Bsplit} gives the exact decomposition
\begin{equation}
B_{ij}=5\rho T^2\delta_{ij}+7T\sigma_{ij}+R^{\cl}_{ij}+\frac{1}{3}\Delta\delta_{ij}.
\label{eq:Bdecomp}
\end{equation}
Apart from \(R^{\cl}_{ij}\) and \(\Delta\), all the terms in \eqref{eq:Bdecomp} depend only on \((\rho,T,\sigma_{ij})\), each of which has its own balance. The remaining fourth-order moments enter \eqref{eq:exact_q_balance} through the divergence of the tensor defined in \eqref{eq:Adef_intro}, \(A_{ij}=R^{\cl}_{ij}+\Delta\delta_{ij}/3\). Up to the common factor \(1/2\) in \eqref{eq:exact_q_balance}, the in-plane equations therefore determine
\begin{equation}
\mathcal{B}_i=(\Grad\cdot\Aten)_i=\partial_j A_{ij}.
\label{eq:Bobs}
\end{equation}
In a two-dimensional flow with \(\partial_z=0\) and \(u_z=0\), the observed components are
\begin{align}
\mathcal{B}_x &= \partial_x A_{xx}+\partial_y A_{xy},\label{eq:Bx}\\
\mathcal{B}_y &= \partial_x A_{xy}+\partial_y A_{yy}.\label{eq:By}
\end{align}
The component \(A_{zz}\) appears nowhere in \eqref{eq:Bx}--\eqref{eq:By}.  Yet \(A_{zz}\) contributes to the trace
\begin{equation}
\Delta = \tr \Aten = A_{xx}+A_{yy}+A_{zz},
\label{eq:traceA}
\end{equation}
and therefore changes the internal split
\begin{equation}
R^{\cl}_{ij}=A_{ij}-\frac{1}{3}\Delta\delta_{ij}.
\label{eq:RfromA}
\end{equation}
The dependence on \(A_{zz}\) is one way in which the cavity differs from a one-dimensional shock. Section~\ref{sec:discussion} shows how R13 supplies these fourth moments through constitutive relations and how R26 transports them.

\begin{proposition}[Planar heat-flux null space]
Consider a monatomic flow in two spatial dimensions, with three velocity components, on a domain \(\Omega\). The fourth-order contribution to its planar heat-flux equations depends on \(\Grad\cdot\Aten\) and is unchanged by any tensor perturbation \(\delta\Aten\) satisfying
\begin{equation}
\partial_j\delta A_{ij}=0,\qquad i=x,y.
\label{eq:divfree}
\end{equation}
This null space includes (i) in-plane tensors generated by an Airy potential,
\begin{equation}
\delta A_{xx}=\partial_{yy}\Phi,\qquad
\delta A_{xy}=-\partial_{xy}\Phi,
\qquad
\delta A_{yy}=\partial_{xx}\Phi,
\label{eq:airy}
\end{equation}
where \(\Phi\) is any smooth scalar function, and (ii) transverse perturbations \(\delta A_{zz}=\eta(x,y)\) with \(\delta A_{xx}=\delta A_{xy}=\delta A_{yy}=0\), where \(\eta\) is smooth and otherwise arbitrary. The latter change \(\Delta\) and \(R^{\cl}_{zz}\) while leaving the planar divergence unchanged.
\end{proposition}

\noindent\textit{Proof.}
Substitution of \eqref{eq:airy} gives
\begin{align}
\partial_x\delta A_{xx}+\partial_y\delta A_{xy}
 &=\partial_{xyy}\Phi-\partial_{yxy}\Phi=0,\\
\partial_x\delta A_{xy}+\partial_y\delta A_{yy}
 &=-\partial_{xxy}\Phi+\partial_{yxx}\Phi=0,
\end{align}
because the mixed derivatives commute. The out-of-plane perturbation drops out of
\eqref{eq:Bx}--\eqref{eq:By} because \(\partial_z=0\). If \(\Omega\) has
nonempty interior, its space of smooth, compactly supported functions,
\(C_c^\infty(\Omega)\), is infinite-dimensional. These potentials give an
infinite-dimensional family of in-plane perturbations with the same observed divergence.
The proof concerns the heat-flux equations alone. Uniqueness of a solution to
the complete R26 boundary-value problem also depends on the remaining balances
and the wall conditions. \(\square\)

\begin{table}
\begin{center}
\small
\setlength{\tabcolsep}{5pt}
\renewcommand{\arraystretch}{1.12}
\begin{tabular}{p{0.25\textwidth}p{0.30\textwidth}p{0.34\textwidth}}
\toprule
 & One-dimensional shock & Two-dimensional cavity \\
\midrule
Quantity determined by the heat-flux balance & Scalar combination \(S=A_{xx}\) & Two components \((\Grad\cdot A)_x,(\Grad\cdot A)_y\) \\
\specialrule{0.25pt}{1.8pt}{1.8pt}
Quantity left undetermined & Division of \(S=R^{\cl}_{xx}+\Delta/3\) into its two terms & Components of the tensor field \(A_{ij}\) \\
\specialrule{0.25pt}{1.8pt}{1.8pt}
Null space & Line \((\delta R^{\cl}_{xx},\delta\Delta)=\eta(1,-3)\) & In-plane Airy tensors and arbitrary changes in \(A_{zz}\) \\
\specialrule{0.25pt}{1.8pt}{1.8pt}
Additional information & One independent scalar, e.g. \(\Delta\) & Constraints on \(A\) from boundary data or additional moment equations, and information about the transverse component \\
\specialrule{0.25pt}{1.8pt}{1.8pt}
Consequence & \(S\) and \(\Delta\) determine both fourth-order contributions & A heat-flux direction alone leaves the individual fourth moments undetermined \\
\bottomrule
\end{tabular}
\caption{Fourth-order information supplied by the heat-flux balance in a normal shock and in a planar cavity. In the cavity the divergence of \(A_{ij}=R^{\cl}_{ij}+\Delta\delta_{ij}/3\) leaves both a divergence-free part of the in-plane tensor and the transverse component undetermined.}
\label{tab:schematic}
\end{center}
\end{table}

In a one-dimensional shock, measuring \(\Delta\) as well as \(S\) gives \(R^{\cl}_{xx}=S-\Delta/3\) (table~\ref{tab:schematic}). In the cavity, the in-plane balance gives the divergence of a tensor field. Measuring its trace supplies additional information but leaves the field undetermined. To recover \(A\), one also needs information from wall closure or other boundary data, additional moment equations, or out-of-plane fourth moments.

The same Airy construction applies to any two-dimensional domain with nonempty interior, since it admits a smooth potential of compact support,
\[
\Phi\in C_c^\infty(\Omega).
\]
For the unit square, an explicit family is
\begin{equation}
\begin{aligned}
\Phi(x,y)&=\Phi_0\,b_\epsilon(x)b_\epsilon(y),\\[-2pt]
b_\epsilon(s)&=
\begin{cases}
\exp\!\left[-\dfrac{1}{(s-\epsilon)(1-\epsilon-s)}\right],
& \epsilon<s<1-\epsilon,\\
0,&\text{otherwise}.
\end{cases}
\end{aligned}
\label{eq:boundary_bump}
\end{equation}
Here \(\Phi_0\) is an arbitrary amplitude and \(0<\epsilon<1/2\) is the width of the collar along each wall. For a nonzero amplitude, \eqref{eq:airy} gives a nonzero symmetric, divergence-free perturbation. It vanishes with all its derivatives in the collar. Figure~\ref{fig:airy_boundary} shows the potential followed by its three second derivatives. The vertical tripole in \(\delta A_{xx}=\partial_{yy}\Phi\), the four lobes in \(\delta A_{xy}=-\partial_{xy}\Phi\) and the horizontal tripole in \(\delta A_{yy}=\partial_{xx}\Phi\) change the local tensor. The in-plane trace, \(\delta A_{xx}+\delta A_{yy}=\nabla^2\Phi\), changes both \(\Delta\) and \(R^{\cl}_{ij}\) locally while \(\partial_jA_{ij}\) stays fixed. The last two panels show the unaltered wall collar and the numerical divergence. Its maximum normalised residual, \(7.34\times10^{-4}\) on a \(501\times501\) grid, is finite-difference error. The continuum divergence is zero. Superposing bumps of different centres, widths and amplitudes generates the infinite-dimensional family. All its members preserve the wall traces of the flux operator. Whether they satisfy the coupled R26 wall map and the remaining moment equations, or admit a nonnegative kinetic distribution, is a separate question.

In three physical dimensions all three components of \(\partial_jA_{ij}\) are observed, so \(A_{zz}\) can enter the heat-flux balance. The divergence still leaves infinitely many tensor fields undetermined. For example, for any smooth scalar \(\Phi\),
\begin{equation}
\delta A_{ij}=\delta_{ij}\nabla^2\Phi-\partial_i\partial_j\Phi
\label{eq:3d_null}
\end{equation}
is symmetric and satisfies \(\partial_j\delta A_{ij}=\partial_i\nabla^2\Phi-\partial_i\nabla^2\Phi=0\) identically. The invisible out-of-plane component is specific to the 2D3V reduction. The ambiguity in recovering a tensor from its divergence extends to 3D.

Reduced-distribution formulations of a 2D3V monatomic flow often retain the transverse molecular velocity through an internal-energy-like degree of freedom. \(A_{zz}\) is a moment of this familiar transverse motion. In a planar cavity it remains undetermined alongside the in-plane Airy modes. A comparison of in-plane heat flux alone therefore leaves both contributions untested when assessing a full R26 moment state.

\begin{center}
\begin{minipage}{\textwidth}
  \centering
  \includegraphics[width=1.0\textwidth]{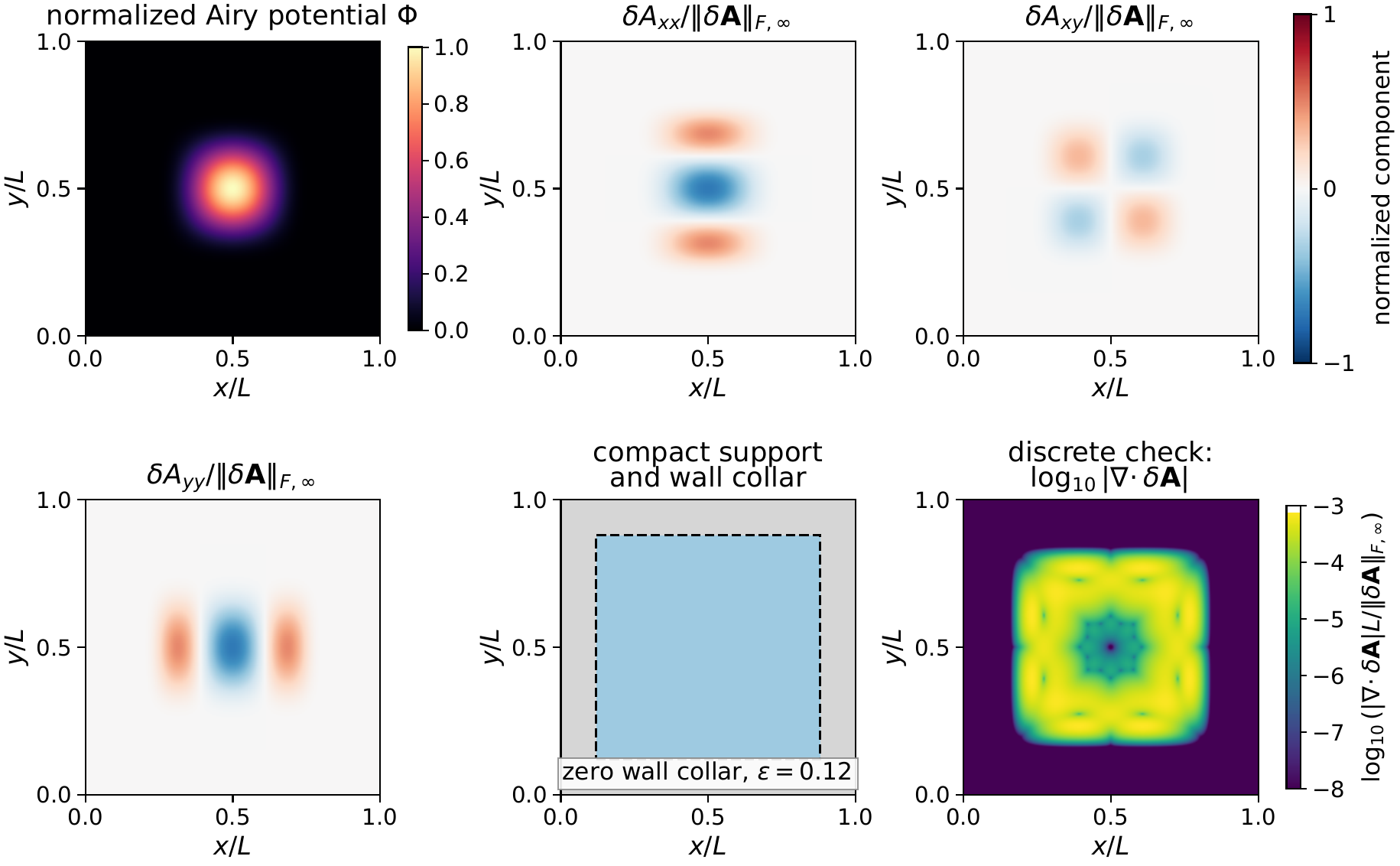}
  \captionof{figure}{Boundary-compatible Airy perturbation generated by \eqref{eq:boundary_bump}.  The potential and all three in-plane tensor components vanish in a finite collar of every wall.  The last two panels show the collar and a direct finite-difference check of \(\Grad\cdot\delta A=\boldsymbol 0\).  The small nonzero numerical residual (maximum \(7.34\times10^{-4}\) after normalisation on a \(501\times501\) grid) is differentiation error; the continuum identity follows exactly from commutation of mixed derivatives.}
  \label{fig:airy_boundary}
\end{minipage}
\end{center}

\section{DSMC calculations and heat-flux measurements}
\label{sec:dsmc}

We simulate monatomic argon in a square lid-driven cavity of side \(L\) by DSMC, with two spatial dimensions and three velocity components. For the principal comparisons with the moment equations, we use Gu's viscosity-based definition of the Knudsen number \citep{Bird1994,GuEmerson2009},
\begin{equation}
 \Kn_{\rm Gu}=\frac{\lambda_{\rm Gu}}{L},\qquad
 \lambda_{\rm Gu}=\frac{15}{2\sqrt{2}(5-2\omega)(7-2\omega)d_{\rm ref}^{2}n}
 \left(\frac{T}{T_{\rm ref}}\right)^{\omega-1/2}.
 \label{eq:kn_gu}
\end{equation}
In the exact moment balances above, the thermal variable is the energy-scaled temperature \(k_BT_{\rm abs}/m\); in the dimensional DSMC specification below, \(T\) denotes the absolute temperature \(T_{\rm abs}\), and its nondimensional form is scaled by \(T_w\).
Here \(\lambda_{\rm Gu}\) is the Gu mean free path, \(n\) is number density, \(d_{\rm ref}\) is the reference molecular diameter, \(T\) and \(T_{\rm ref}\) are the local and reference absolute temperatures, and \(\omega\) is the viscosity index.  All walls are diffuse and isothermal at \(T_w=300\,\mathrm{K}\); the top wall moves with speed \(U_w=100\,\mathrm{m\,s^{-1}}\).  The base case has \(\Kn_{\rm Gu}=0.05\), a \(160\times160\) sampling mesh, fixed collision subcells, an upstream-equivalent particle loading of 128 particles per cell, and eight independent stochastic realisations.  The argon variable-hard-sphere (VHS) parameters are \(m=6.6335\times10^{-26}\,\mathrm{kg}\), \(d_{\rm ref}=4.17\times10^{-10}\,\mathrm{m}\) at \(T_{\rm ref}=273\,\mathrm{K}\), and \(\omega=0.81\) \citep{Bird1994}.  Collisions are sampled with Bird's no-time-counter procedure.  The time step is kept below the cell-crossing and mean-collision-time scales; after transient removal, samples are accumulated at a fixed stride, giving 8501 samples per realisation and an average collision count of \(9.32\times10^4\) per cell in the ensemble used below.

Coordinates are nondimensionalised by \(L\), velocity by \(c_0=(k_BT_w/m)^{1/2}\), density by the equilibrium reference mass density \(\rho_0\), heat flux by \(\rho_0c_0^3\), the fourth-order tensor by \(\rho_0c_0^4\), and its divergence by \(\rho_0c_0^4/L\), where \(k_B\) is Boltzmann's constant. For the transition case, \(c_0=249.8797\,\mathrm{m\,s^{-1}}\) and \(\rho_0c_0^3=8.9367\times10^6\,\mathrm{W\,m^{-2}}\). The directional indicators and ratios reported below are dimensionless. We use the VHS law with \(\omega=0.81\) for the grid and particle-loading study at the lower Knudsen number. The model comparisons at \(\Kn_{\rm Gu}=0.05\) and \(0.20\) use the Maxwell-class variable-soft-sphere (VSS) transport law specified in \S\S\ref{sec:kn005_matched}--\ref{sec:kn020_matched}.

The sampled fields are \(\rho\), \(\uv\), \(T\), \(\qv\), \(\sigma_{ij}\), \(R^{\cl}_{ij}\), \(\Delta\), the composite tensor \(A_{ij}\), and all ten symmetric components of the direct third central moment \(M_{ijk}\).  When used below, \(P_{ij}=\int c_i c_j f\,\dd\vvec=\rho T\delta_{ij}+\sigma_{ij}\) denotes the central momentum-flux or pressure tensor.  The traceless R26 third-order moment is formed directly as
\begin{equation}
m_{ijk}=M_{ijk}-\frac{2}{5}\left(q_i\delta_{jk}+q_j\delta_{ik}+q_k\delta_{ij}\right),
\qquad M_{ill}=2q_i .
\label{eq:m_direct}
\end{equation}
Equation~\eqref{eq:m_direct} is the standard trace-free projection of a symmetric third-rank tensor: subtracting \(\tfrac15\left(M_{ill}\delta_{jk}+M_{jll}\delta_{ik}+M_{kll}\delta_{ij}\right)\) with \(M_{ill}=2q_i\) removes every trace.
We sample all the symmetric components needed for \(m_{ijk}\). The fourth central moments give \(R^{\cl}_{ij}\) and \(\Delta\) directly. A comparison of the complete R26 moment field would additionally require uncertainties on each component, a grid-converged \(m_{ijk}\) field, and agreement with the coupled balances and wall conditions for the same collision law.

We identify anti-Fourier regions using the cosine indicator
\begin{equation}
I_{AF}=\frac{\qv\cdot\Grad T}{|\qv|\,|\Grad T|}.
\label{eq:IAF}
\end{equation}
Where $I_{AF}$ is positive, the heat flux has a component towards warmer gas. To examine the fourth-order terms along that direction, we project the tensorial and scalar parts of $(\nabla\cdot A)$ onto the local unit vector $(\boldsymbol q/|\boldsymbol q|)$,
\begin{equation}
P_R=\frac{\boldsymbol q\cdot(\nabla\cdot R^{\mathrm{cl}})}{|\boldsymbol q|},
\label{eq:PR}
\end{equation}
\begin{equation}
P_\Delta=\frac{\boldsymbol q\cdot(\nabla\Delta/3)}{|\boldsymbol q|}.
\label{eq:PD}
\end{equation}
The sign of each projection distinguishes a fourth-order contribution along the local heat flux from one opposing it. Over the anti-Fourier cells we use two measures of their relative size,
\[
\frac{P_\Delta}{P_R}\equiv
\frac{\langle P_\Delta^2\rangle_{AF}^{1/2}}
     {\langle P_R^2\rangle_{AF}^{1/2}},\qquad
\chi_\Delta\equiv\frac{|P_\Delta|}{|P_R|+|P_\Delta|},
\]
where \(\langle\cdot\rangle_{AF}\) denotes an equally weighted mean over active anti-Fourier cells. The ratio \(P_\Delta/P_R\) compares the two RMS projections. To obtain the mean local scalar fraction \(\langle\chi_\Delta\rangle_{AF}\), we first evaluate \(\chi_\Delta\) in each cell where its denominator is nonzero.

Unless otherwise stated, fields are smoothed with a seven-point centred window before differentiating.  Let \(\Omega_0\) denote the stated analysis domain, \(M_{act}\) the set satisfying simultaneous lower bounds on \(|\qv|\) and \(|\Grad T|\), and \(M_{AF}=M_{act}\cap\{I_{AF}>0\}\).  The default lower bound is 0.05 relative to each domain maximum.  The notation \(|M|\) denotes physical area, evaluated by cell count on the uniform comparison grid.  We report three distinct support measures,
\begin{equation}
 f_{act|\Omega}=\frac{|M_{act}|}{{|\Omega_0|}},\qquad
 f_{AF|\Omega}=\frac{|M_{AF}|}{|\Omega_0|},\qquad
 f_{AF|act}=\frac{|M_{AF}|}{|M_{act}|}.
\label{eq:support_metrics}
\end{equation}
The fraction of the physical domain occupied by anti-Fourier cells is \(f_{AF|\Omega}\). The conditional fraction \(f_{AF|act}\) uses only active cells in its denominator. In comparing R13 and R26, a large conditional fraction can therefore coexist with a small anti-Fourier area.

\subsection{Ensemble, grid and particle-loading sensitivity}
\label{sec:numerical_scope}

We ran five eight-realisation VHS ensembles at \(\Kn_{\rm Gu}=0.05\), \(U_w=100\,\mathrm{m\,s^{-1}}\) and \(T_w=300\,\mathrm{K}\). To test spatial resolution we varied \(N=120\), 160 and 200 at 128 particles per cell. To test particle loading we used 64, 128 and 256 particles per cell on the \(160^2\) grid. All realisations have 8501 accumulated samples per cell, the same final physical time and the same collision law with viscosity index \(\omega=0.81\). Before averaging the forty fields, we checked positivity, the trace identities for stress, \(R^{\cl}\) and \(m_{ijk}\), the relation \(A_{ij}=R^{\cl}_{ij}+\Delta\delta_{ij}/3\), and the scaling between physical and number-density variables.

We map the fields onto a common \(160^2\) grid and use the \(160^2\), 128-particle ensemble mean as reference. For each field \(X\) in table~\ref{tab:dsmc_grid_ppc}, the relative RMS difference is \(E_X=\|X_{\rm design}-X_{\rm ref}\|_2/\|X_{\rm ref}\|_2\). Mapping the \(120^2\) grid requires linear extrapolation at 636 of the 25,600 target cell centres near the boundary. Thus its differences include discretisation, interpolation and Monte Carlo scatter, which precludes interpreting them as Richardson-error estimates. The spread over the five designs gives a descriptive range, with no assigned confidence level. The area fractions depend on the seven-point smoothing window and 0.05 activity threshold. Section~\ref{sec:kn020_projection} examines how smoothing affects the sign of \(\qv\cdot\Grad T\) in the transition case. We have not repeated the test for sensitivity to the threshold.

\begin{table}
\begin{center}
\scriptsize
\setlength{\tabcolsep}{3.5pt}
\renewcommand{\arraystretch}{1.14}
\begin{tabular}{lrrrrrrrr}
\toprule
case & \(E_T\) (\%) & \(E_u\) (\%) & \(E_q\) (\%) & \(E_\sigma\) (\%) & \(E_R\) (\%) & \(E_\Delta\) (\%) & \(f_{AF|\Omega}\) & \(P_\Delta/P_R\)\\
\midrule
\(160^2\), 128 & -- & -- & -- & -- & -- & -- & 0.04766 & 0.04192\\
\specialrule{0.25pt}{1.6pt}{1.6pt}
\(120^2\), 128 & 0.044 & 0.892 & 12.54 & 2.97 & 16.96 & 76.30 & 0.04730 & 0.04855\\
\specialrule{0.25pt}{1.6pt}{1.6pt}
\(200^2\), 128 & 0.042 & 0.765 & 12.09 & 2.70 & 16.28 & 72.88 & 0.04348 & 0.03811\\
\specialrule{0.25pt}{1.6pt}{1.6pt}
\(160^2\), 64  & 0.059 & 1.096 & 17.45 & 3.84 & 23.31 & 104.94 & 0.05375 & 0.04801\\
\specialrule{0.25pt}{1.6pt}{1.6pt}
\(160^2\), 256 & 0.042 & 0.762 & 12.27 & 2.73 & 16.54 & 74.51 & 0.04004 & 0.04181\\
\bottomrule
\end{tabular}
\captionsetup{format=revisioncaption}\caption{Effect of grid size and particles per cell on the eight-realisation DSMC means. We compare the fields on a common \(160^2\) grid, using the \(160^2\), 128-particle mean as reference. The measured differences include scatter between independent realisations. The ratio of scalar to tensorial projections \(P_\Delta/P_R\) lies between \(0.038\) and \(0.049\) in all five calculations.}

\label{tab:dsmc_grid_ppc}
\end{center}
\end{table}

\begin{shaded}
Across the five designs the relative differences in temperature, velocity and stress are 0.042--0.059\%, 0.762--1.096\% and 2.70--3.84\%, respectively. They reach 12.1--17.5\% for heat flux, 16.3--23.3\% for \(R^{\cl}\) and 72.9--104.9\% for the cell-wise \(\Delta\) field. To estimate how much of this spread arises from finite sampling, we use the cell-wise variance across the eight realisations to predict the difference between two independent eight-realisation means in the absence of a design effect. For \(\qv\) and \(R^{\cl}\) the observed differences exceed this prediction by less than 1\%. The remaining contribution, \((E^2-E_{\rm noise}^2)^{1/2}\), is at most 1.3\% and 3.3\% of the respective field norms. For \(\Delta\) it ranges from 2--13\%. The predicted and measured differences agree to about 1\%, comparable to the uncertainty in the eight-realisation variance estimate. We therefore cannot distinguish the remainder for \(\Delta\) from zero.

Averaging over \(4\times4\)-cell blocks, the resolution of table~\ref{tab:q_parts}, reduces the differences to 4.3--6.1\% for \(\qv\), 5.8--8.1\% for \(R^{\cl}\) and 36--49\% for \(\Delta\). These values agree with the sampling prediction within 0.1--0.5 percentage points. Thus the differences between ensemble means in table~\ref{tab:dsmc_grid_ppc} reveal no detectable grid or particle-loading effect above the sampling noise. Sampling noise also dominates the local \(\Delta\) field within any one design.

The tensorial projection remains dominant across the five designs. The ranges are \(f_{AF|\Omega}=0.0400\)--0.0538, \(\langle I_{AF}\rangle_{AF}=0.308\)--0.342, and \(P_\Delta/P_R=0.0381\)--0.0485. The approximate jackknife standard errors across the eight realisations are 0.0022--0.0034 for \(f_{AF|\Omega}\) and 0.0015--0.0027 for \(P_\Delta/P_R\). Expressed in standard errors of the difference from the reference, the largest departures are 2.2 for \(f_{AF|\Omega}\) at 256 particles per cell and 2.1 for \(P_\Delta/P_R\) on the \(120^2\) grid. These estimates apply to averages. The sign in any particular cell and the topology of the anti-Fourier region remain uncertain.
\end{shaded}

Relative to the \(160^2\), 128-particle reference, the maximum absolute departures are \(0.00762\) for \(f_{AF|\Omega}\), \(0.01946\) for \(\langle I_{AF}\rangle_{AF}\), and \(0.00663\) for \(P_\Delta/P_R\), or 16.0\%, 5.95\% and 15.8\% of the respective reference values. These ranges describe the combined effects of grid size, particle loading and scatter between realisations. The variance estimates attribute most of the spread to sampling. All five designs use the same window of 8501 samples per cell. At the higher Knudsen number we use 20,000 correlated samples per cell. We examine the effect of smoothing and differentiation on those data separately below.

\section{\texorpdfstring{\colorbox{white}{\parbox{0.75\linewidth}{Scalar and tensorial contributions in counter-gradient regions}}}{Scalar and tensorial contributions in counter-gradient regions}}

\label{sec:driver}

In the \(160^2\) reference calculation with 128 particles per cell, active cells cover \(f_{act|\Omega}=0.27430\) of the cavity and active anti-Fourier cells cover \(f_{AF|\Omega}=0.047656\). Anti-Fourier alignment thus occurs in \(f_{AF|act}=0.17374\) of the active cells, with mean alignment \(\langle I_{AF}\rangle_{AF}=0.32706\). Over these cells the RMS scalar projection is a small fraction of the RMS tensorial projection, \(P_\Delta/P_R=0.041921\). The mean local scalar fraction is \(\langle\chi_\Delta\rangle_{AF}=0.076884\).

\begin{shaded}
The tensorial projection dominates at every grid size and particle loading tested (table~\ref{tab:dsmc_grid_ppc}). The anti-Fourier cells occupy between 0.0400 and 0.0538 of the cavity, and \(P_\Delta/P_R\) never exceeds 0.04855.
\end{shaded}

Section~\ref{sec:kn020_matched} compares the macroscopic fields of DSMC, R13 and R26 at the higher rarefaction, \(\Kn_{\rm Gu}=0.20\). The fourth-moment statistics above refer to the lower Knudsen number.

\section{Admissibility of the null-space perturbations}
\label{sec:nullspace}

Consider the perturbed field \(\Aten^*=\Aten+\alpha\,\delta\Aten\), where \(\delta\Aten\) is the compactly supported Airy tensor in \eqref{eq:airy} or an out-of-plane \(A_{zz}\) mode. The amplitude \(\alpha\) sets the RMS size of the perturbation relative to that of the composite tensor on the active set. It serves only to normalise the perturbation. The continuum divergence \(\Grad\cdot\Aten\) is invariant for any such normalisation. For each perturbed field we can reconstruct
\begin{equation}
\Delta^*=\tr \Aten^*,\qquad R^{\cl *}_{ij}=A^*_{ij}-\frac{1}{3}\Delta^*\delta_{ij},
\label{eq:shift_recover}
\end{equation}
and evaluate the component along the local heat flux,
\begin{equation}
P_A=\frac{\qv\cdot(\Grad\cdot\Aten)}{|\qv|}
\label{eq:PA}
\end{equation}
The Airy perturbation leaves \(P_A\) unchanged in the continuum for any amplitude, and \(A_{zz}\) is absent from the planar balance. Kinetic realisability places further restrictions on the perturbed moments.

Suppose a nonnegative, number-normalised distribution \(f^*\) has the starred moments. Write \(n=\int f^*\,\dd\vvec\) for its number density and define its thermal temperature \(\theta\) by \(3n\theta=\int|\cv|^2f^*\,\dd\vvec\). Two necessary conditions follow. First, the Cauchy--Schwarz inequality gives
\begin{equation}
    \left(\int |\cv|^2 f\,\dd\vvec\right)^2
    \leq \left(\int f\,\dd\vvec\right)
          \left(\int |\cv|^4 f\,\dd\vvec\right),
\end{equation}
which, with \(\int|\cv|^2f^*\,\dd\vvec=3n\theta\) and \(\int|\cv|^4f^*\,\dd\vvec=15n\theta^2+\Delta^*\), gives the scalar lower bound \(\Delta^*\geq -6n\theta^2\).  The second check is tensorial: for any vector \(\boldsymbol a\),
\begin{equation}
    a_i a_j B^*_{ij}=\int |\cv|^2(\boldsymbol a\cdot\cv)^2 f^*\,\dd\vvec \geq 0,
\end{equation}
so the contracted fourth-order matrix
\begin{equation}
    B^*_{ij}=5n\theta^2\delta_{ij}+7\theta\sigma_{ij}+A^*_{ij}
    \label{eq:Bpsd}
\end{equation}
must be positive semidefinite. Here \(n\) and \(\theta\) are the number density and thermal temperature used in DSMC sampling. Up to the constant molecular-mass factor, this expression is \(5\rho T^2\delta_{ij}+7T\sigma_{ij}+A_{ij}\) in the mass-normalised notation of \eqref{eq:Bdecomp}. If the unperturbed field satisfies both inequalities strictly, continuity ensures that sufficiently small perturbations preserve them. Satisfying these necessary conditions still leaves open whether a nonnegative distribution exists and whether the perturbed field solves the Boltzmann equation or the R26 equations with their coupled wall conditions. We therefore use finite-amplitude perturbations only to test what the heat-flux balance can detect.

\begin{shaded}
Within the anti-Fourier region, where the tensorial contribution dominates, \(\Grad\cdot A\) still leaves \(A\) undetermined. Measuring the full three-dimensional tensor, in contrast, fixes both \(\Delta=\tr A\) and \(R^{\cl}=A-(\tr A)I/3\). For an earlier uniform-grid R26 field, a positive distribution matched all 28 target moments at 347 of the 352 non-corner wall nodes. The other five nodes remain unclassified (appendix~\ref{app:same_state}). Pointwise realisability alone places no condition on the spatial variation required by the full R26 equations. In the discrete system the third- and fourth-order moment equations respond to both the Airy and the transverse perturbations (appendix~\ref{app:complete_response}).
\end{shaded}

\section{Implications for R13, R26 and kinetic-model validation}
\label{sec:discussion}

The heat-flux balance \eqref{eq:exact_q_balance} of \S\ref{sec:q_observable} is common to R13 and R26. Its fourth-order contribution beyond the lower moments is
\[
 \tfrac12\partial_j A_{ij}
 =\tfrac12\partial_jR_{ij}+\tfrac16\partial_i\Delta,\qquad i=x,y .
\]
Here \(R_{ij}=R^{\cl}_{ij}\) is the contracted fourth-order deviator and \(\Delta\) the scalar fourth-order excess, both defined in \eqref{eq:RDelta_def}. Both models transport the deviatoric stress \(\sigma_{ij}\) and heat flux \(q_i\). They differ in how they supply \(R_{ij}\), \(\Delta\) and the trace-free third moment \(m_{ijk}=\int c_{\langle i}c_jc_{k\rangle}f\,\dd\vvec\). The peculiar velocity \(\cv\) has three components. Repeated indices are summed, \(\partial_j\) denotes differentiation with respect to \(x_j\), and angular brackets take the symmetric trace-free part in three dimensions, even though \(\partial_z=0\).

Write \(\rho\) for density, \(\uv\) for gas velocity and \(\theta=k_BT_{\rm gas}/m_{\rm Ar}\) for temperature in specific-energy units (the \(T\) of \S\ref{sec:q_observable}). Here \(T_{\rm gas}\) is in kelvin, \(k_B\) is Boltzmann's constant and \(m_{\rm Ar}\) is the argon atomic mass. The hydrostatic pressure is \(p=\rho\theta\) and the total energy per unit volume is \(E=\rho(3\theta+u_i u_i)/2\). Both models obey the same steady conservation laws,
\[
 \begin{aligned}
 \partial_j(\rho u_j)=0,\qquad
 \partial_j(\rho u_i u_j+p\delta_{ij}+\sigma_{ij})&=0,\\
 \partial_j\big[(E+p)u_j+u_i\sigma_{ij}+q_j\big]&=0.
 \end{aligned}
\]

In R13 let \(U\) collect the planar components of \((\rho,\uv,\theta,\sigma,q,m,R,\Delta)\). The steady nondimensional bulk equations have the form \citep{RanaTorrilhonStruchtrup2013}
\[
 A_x(U)\,\partial_xU+A_y(U)\,\partial_yU
       +(\rho/\Kn_R)\,\mathcal C(U)=0,\qquad
 \Kn_R=\sqrt{2/\pi}\,\Kn_{\rm Gu}.
\]
The matrices \(A_x,A_y\) multiply spatial gradients and \(\mathcal C\) is the collision production in the sign convention above. Nine planar fields have transport equations. The other eight entries give \(m_{ijk}\), \(R_{ij}\) and \(\Delta\) through constitutive equations. Their leading dimensional gradient terms for Maxwell molecules are, respectively,
\[
 -\frac{2\mu}{\rho}\partial_{\langle k}\sigma_{ij\rangle},\qquad
 -\frac{24\mu}{5\rho}\partial_{\langle i}q_{j\rangle},\qquad
 -\frac{12\mu}{\rho}\partial_kq_k ,
\]
where \(\mu\) is the dynamic viscosity. The calculations also retain the nonlinear gradient and production terms, including stress--heat-flux products in the \(m\) equations and stress and heat-flux squares in the fourth-order equations (appendix~\ref{app:r13_production_explicit}).

R26 gives these eight components their own transport equations \citep{GuEmerson2009}. With \(\mathcal D a=\partial_l(u_l a)\), the dimensional steady balances read
\[
\begin{aligned}
 \mathcal D m_{ijk}+\partial_l\Phi_{ijkl}
   +3p\,\partial_{\langle k}(\sigma_{ij\rangle}/\rho)
   +\tfrac{3p}{2\mu}m_{ijk}&=\mathcal S^{(m)}_{ijk},\\
 \mathcal D R_{ij}+\partial_k\psi_{ijk}
   +\tfrac{28p}{5}\partial_{\langle j}(q_{i\rangle}/\rho)
   +\tfrac{7p}{6\mu}R_{ij}&=\mathcal S^{(R)}_{ij},\\
 \mathcal D\Delta+\partial_i\Omega_i
   +8p\,\partial_i(q_i/\rho)+\tfrac{2p}{3\mu}\Delta&=\mathcal S^{(\Delta)} .
\end{aligned}
\]
The fourth-rank deviator \(\Phi_{ijkl}\), third-rank deviator \(\psi_{ijk}\) and vector \(\Omega_i\) are the regularising fluxes. They contain gradients of \(m/\rho\), \(R/\rho\) and \(\Delta/\rho\) and their nonlinear couplings. The sources \(\mathcal S^{(m)}_{ijk},\mathcal S^{(R)}_{ij},\mathcal S^{(\Delta)}\) collect the remaining couplings, including the gradient of \(R_{ij}\) in the \(m_{ijk}\) equation. Appendix~\ref{app:actual_algorithm} gives these sources (there denoted by \(M_{ijk},S_{ij},N\)) and the three fluxes in full. Thus two fields can agree in \(\partial_jA_{ij}\) while imposing different demands on the three R26 balances.

All four walls reflect diffusely with full accommodation at \(T_w=300\,\mathrm K\). The lid moves at \(100\,\mathrm{m\,s^{-1}}\) along \(+x\); the other three walls are stationary. Let \(n\) point into the gas and \(t\) along the wall. The normal gas velocity is zero, and each model's diffuse-wall moment relations determine the tangential slip and gas temperature jump. The two models use the respective effective wall pressures
\[
 \begin{aligned}P_{13}&=p+\tfrac12\sigma_{tt}-\frac{R_{tt}}{28\theta}
                      -\frac{\Delta}{120\theta},\\
 P_{26}&=p+\tfrac12\sigma_{nn}-\frac{R_{nn}}{28\theta}
                      -\frac{\Delta}{120\theta}-\frac{\Phi_{nnnn}}{24\theta}.\end{aligned}
\]
With these pressures, the tangential slip \(V_t=u_t-u_{w,t}\) relative to the local wall velocity \(u_{w,t}\) satisfies
\[
 \begin{gathered}
 P_s V_t=-\sqrt{\pi\theta/2}\,\sigma_{nt}-q_t/5-m_{nnt}/2+B_s,\\
 B_{13}=0,\qquad B_{26}=(9\Omega_t+70\psi_{nnt})/(2520\theta),
 \end{gathered}
\]
where \(s=13\) or \(26\) labels the model and subscripts \(n,t\) denote projections on the local wall axes. The tangential-stress convention in \(P_{13}\) is that used in our R13 implementation. To complete the eleven R26 relations at a smooth-wall node we extrapolate \(p\), \(\sigma_{nt}\), \(q_n\), \(m_{nnn}\), \(m_{ntt}\) and \(R_{nt}\) from the interior. The temperature jump is given in appendix~\ref{app:finite_wall_coupling}.

At \(T_w=300\,\mathrm{K}\), the equilibrium sound speed in argon is \(a_w=\sqrt{5k_BT_w/(3m_{\rm Ar})}\simeq323\,\mathrm{m\,s^{-1}}\). Our lid Mach number \(U_w/a_w=0.31\) is subsonic. It compares the imposed wall speed with the sound speed, whereas \(\Kn_{\rm Gu}\) measures rarefaction. At a sharp lid the prescribed wall velocity changes by the full \(100\,\mathrm{m\,s^{-1}}\) at each upper corner. In the fine-grid R26 calculations the largest residuals are in the wall extrapolation near these corners. Tapering the lid velocity to zero there, and repeating DSMC with the same profile, provides a separate test of the wall treatment. At \(\Kn_{\rm Gu}=0.05\) the present R26 refinement sequence reached a discrete root on \(38^2\) with the sharp lid and on \(90^2\) with the tapered lid. For the tapered lid, the four discrepancies from matched DSMC changed by at most \(3\%\) across the \(70^2\), \(80^2\) and \(90^2\) roots (tables~\ref{tab:root_states} and \ref{tab:smooth_lid}).

\begin{shaded}
The nominal Knudsen number does not describe the local departure from equilibrium in the cavity. \citet{MohammadzadehRoohi2012} noted that the limits of the slip regime were established for long channels with one-dimensional gradients and should be reconsidered for bounded two-dimensional domains. From the DSMC fields, smoothed over two cells, we evaluated the gradient-length Knudsen numbers of density and temperature, $\Kn_\phi=\lambda|\nabla\phi|/\phi$, and a shear-based number $\Kn_S=\lambda\|\nabla\uv\|/U_w$, with $\lambda$ the local mean free path \citep{BoydChenCandler1995,RoohiReview2026}. In plane Couette flow between walls a distance $L$ apart $\Kn_S$ is at most the nominal value, and at the present lid Mach number $\Kn_\rho$ is at least about fifty times smaller; at the wall of plane Poiseuille flow $\Kn_S$, based on the centreline speed, is four times the nominal value. In the cavity (figure~\ref{fig:kngl}) the median of $\Kn_S$ in the layer $y/L>0.9$ below the lid, away from the corners, is 0.24 at $\Kn_{\rm Gu}=0.05$ and 0.70 at $\Kn_{\rm Gu}=0.20$, that is 4.7 and 3.5 times the nominal value, and in the two lid corners $\Kn_\rho$ is comparable with the nominal value. Over the whole cavity the median of $\Kn_S$ is 0.032 and 0.096, and it exceeds the nominal value in 30\% and 24\% of the area. The cavity thus combines a lid layer whose shear-based number is several times the nominal value, a core at about half that value, and corner density gradients that are absent from low-speed channel flow. A single nominal Knudsen number does not characterise it. These values depend on the velocity scale: with the local sound speed in place of $U_w$ they are lower by the factor 0.31.
\end{shaded}

\begin{figure}[tbp]
\centering
\includegraphics[width=\linewidth]{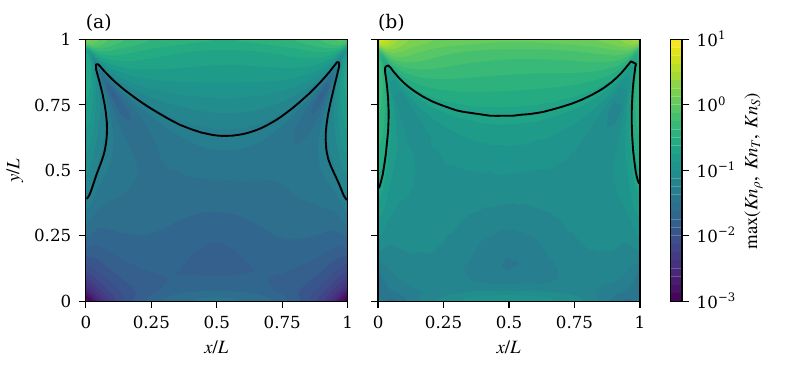}
\captionsetup{format=revisioncaption}
\caption{Local Knudsen number $\max(\Kn_\rho,\Kn_T,\Kn_S)$ from the DSMC fields at (a) $\Kn_{\rm Gu}=0.05$ and (b) $\Kn_{\rm Gu}=0.20$. On the black curve it equals the nominal value.}
\label{fig:kngl}
\end{figure}

The wall slip and the bulk heat flux sample different combinations of moments. The R26 slip involves \(\psi_{nnt}\) and \(\Omega_t\), whereas the bulk heat-flux equation samples the fourth moments through \(\partial_jA_{ij}\). This distinction motivates the comparison of the computed bulk moments with their DSMC values.

\subsection{R13 and R26 cavity fields}
\label{sec:model_diagnostic}

We compare fields computed from R13 and R26 with DSMC at \(\Kn_{\rm Gu}=0.05\) and \(0.20\), using \eqref{eq:kn_gu} to define rarefaction. The R26 calculations use the coefficients of \citet{GuEmerson2009} and the DSMC velocity normalisation. Since these coefficients depend on the molecular collision law, the comparison also reflects differences between the collision operators.

The R13 calculations use the 17-field cavity formulation of \citet{RanaTorrilhonStruchtrup2013}.  The coefficient matrices and saddle-point structure follow that reference, while the spatial discretisation enforces a conservative shared-face finite-volume continuity balance, a compatible mass constraint, defect-Newton linearisation, the two-point wall extrapolation and the tangential-stress effective-pressure convention.  We use the published coefficients.  For earlier conservative finite-volume algorithms and improved moment wall conditions, see \citep{GuEmerson2007}.  At \(\Kn_{\rm Gu}=0.05\), independent \(60^2\) and \(80^2\) solutions satisfy the residual, mass, continuity, positivity and wall-pressure checks.  At \(\Kn_{\rm Gu}=0.20\), the \(60^2\) solution has a core relative residual of \(1.13\times10^{-14}\), positive density and temperature, and a maximum local continuity residual of \(4.77\times10^{-14}\).  Our spatial discretisation differs in these respects from the 2013 calculation.

For R26 we discretise the nonlinear bulk equations, closure relations and smooth-wall conditions of \citet{GuEmerson2009}.  The solver reconstructs the complete three-dimensional symmetric trace-free tensors before applying the planar 17-variable reduction, which retains \(q_i\), \(R^{\cl}_{ij}\), \(m_{ijk}\) and \(\Delta\).  In the notation of Gu \& Emerson, \(A_{\psi1}=1.698\) in the \(\psi\)-type closure. We retain the nonlinear source and all the \(\phi\), \(\psi\) and \(\Omega\) regularisation terms, use quotient gradients for density-normalised moments and keep the geometric frame fixed on each wall face. The regularising terms and coefficients are those of R26 throughout. The divergence enters with its published coefficient, and DSMC serves as an independent comparison.

For the field comparisons we use R13 on \(60\times60\) grids and R26 on a stretched \(60\times60\) grid at \(\Kn=0.05\) and a uniform \(90\times90\) grid at \(\Kn=0.2\). The R13 fields satisfy the discrete equations to the residuals given below. Iteration reduced the R26 residual on both grids, but a finite residual remains, particularly near the lid corners. Appendix~\ref{app:actual_algorithm} gives the variables, residual definitions and stopping criteria.

We first pool the DSMC moments in the laboratory frame, weighting each run by its number of samples, and then transform to central moments. At $\Kn=0.05$, one run covers steps $10^5$--$5\times10^6$ and six cover $10^5$--$1.98\times10^6$. Excluding empty sampling blocks gives 1,617,300 sample times per cell in total. At $\Kn=0.2$, seven runs start at $2\times10^6$ and end between $4.5\times10^6$ and $8\times10^6$, contributing 288,000 sample times per cell.

\subsubsection{Discrete residuals and measures of agreement}

\begin{shaded}
The supplementary figures show local RMS residuals for the 17 equations at each Knudsen number. Each panel gives the unfiltered interior RMS. For the colour maps we calculate $[G(Mr^2)/G(M)]^{1/2}$, where $G$ denotes Gaussian averaging over one native grid cell and $M$ selects the bulk equations. We assess convergence using the residual before averaging. For R13 we evaluate the original equations without the extra mass constraint. For R26 the residuals at $\Kn=0.05$ include the equation scaling used in the nonlinear solve, whereas those at $\Kn=0.2$ are unscaled. Their magnitudes therefore use different normalisations. The auxiliary R13 equations impose constitutive relations.

\end{shaded}

\begin{shaded}
We obtained the steady R13 fields by defect-Newton and Jacobian-free Newton--Krylov iteration. For R26 at $\Kn_{\rm Gu}=0.20$, we used projected, regularised linear corrections on the uniform grid (appendix~\ref{app:recorded_behaviour}) and continued the calculation with trust-region Gauss--Newton steps. Levenberg--Marquardt iteration on successively finer grids gave a second $90\times90$ field (appendix~\ref{app:complete_continuation}) and the discrete roots in tables~\ref{tab:root_states} and \ref{tab:smooth_lid}. We also integrated the discrete equations in time from gas at rest (appendix~\ref{app:transient}). The final R13 fields differ from the steady fields by at most $1.6\times10^{-8}$ at $\Kn_{\rm Gu}=0.20$ and $8.6\times10^{-11}$ at $0.05$, relative to the largest steady value. For R26 with the smooth lid of table~\ref{tab:smooth_lid}, the corresponding difference on the $70\times70$ grid at $\Kn_{\rm Gu}=0.05$ is $1.6\times10^{-6}$. Appendix~\ref{app:transient} also describes start-up with the sharp lid and with the smooth lid at $\Kn_{\rm Gu}=0.20$.
\end{shaded}

Re-evaluating the residual checks how closely the fields satisfy the equations on their own grid. The near-wall oscillations in R13 also require a refinement study before they can be interpreted as a feature of the continuum solution. For R26 we test the wall conditions, mass and pointwise residuals separately from the least-squares objective.

Within each comparison we evaluate DSMC and both models on the same physical domain.  In the notation below, subscripts \(M\) and \(D\) denote a moment model and DSMC, respectively; \(M_D\) is the DSMC-defined common mask, and the norm and inner product are equally weighted over its cells.  We use Jaccard and Dice overlap for topology and
\begin{equation}
\begin{aligned}
 E_X&=\frac{\|X_M-X_D\|_{M_D}}{\|X_D\|_{M_D}},\\
 C_q&=\frac{\langle\qv_M,\qv_D\rangle_{M_D}}{\|\qv_M\|_{M_D}\|\qv_D\|_{M_D}},
 &G_q&=\frac{\|\qv_M\|_{M_D}}{\|\qv_D\|_{M_D}}.
\end{aligned}
\label{eq:model_errors}
\end{equation}
for relative field error, vector correlation and RMS gain; \(X_M\) and \(X_D\) are the model and DSMC versions of the field \(X\).  The same definitions with \(\qv\) replaced by \(\uv\) give \(E_u\) and \(C_u\).  The mean angular difference \(\bar\vartheta_q\) is weighted by \(|\qv_M||\qv_D|\).  We use these definitions in the field comparisons and specify the smoothing, activity threshold and corner exclusion for the binary masks.

\enlargethispage{3\baselineskip}
\subsubsection{Maxwell-molecule comparison at \(\Kn_{\rm Gu}=0.05\)}
\label{sec:kn005_matched}

\begin{shaded}
At $\Kn=0.05$ we use the VSS transport representation of inverse-power-law Maxwell molecules, with $\omega=1$, VSS angular-scattering parameter $\alpha_{\rm VSS}=2.140$, diffuse walls and the reference scales given above (flow and collision parameters in \S\ref{sec:kn020_matched}). This representation reproduces the Maxwell-molecule transport properties but has a different angular scattering law. Figure~\ref{fig:kn005_primary} compares density, temperature, speed and heat-flux magnitude in the pooled DSMC reference and the R13 and R26 calculations. The velocity fields agree more closely than the heat-flux fields. On the common comparison mask, the relative velocity differences from DSMC are 0.0364 for R13 and 0.0149 for R26, whereas the corresponding heat-flux differences are 0.6415 and 0.1316 (both moment grids $60^2$, table~\ref{tab:current_fields}). These errors use the complete vectors. The centreline cuts in figure~\ref{fig:kn005_centerlines} retain the oscillations on the native R13 grid, and figure~\ref{fig:kn005_antifourier} compares the direction of heat flow with the temperature gradient. Figures~\ref{fig:kn005_primary} and~\ref{fig:kn005_centerlines} use the sample-weighted DSMC mean. Error, gain and overlap estimates calculated with equal sampling windows use a different mean.
\end{shaded}

For the orientation map we evaluate $I_{AF}=\widehat{\boldsymbol q}\cdot\widehat{\nabla T}$ after applying the same Gaussian filter to both vectors. Positive values indicate heat transfer towards warmer gas. Where either vector is weak, its direction and hence the location of the boundary between positive and negative values are uncertain. Figures~\ref{fig:kn005_antifourier} and~\ref{fig:kn020_antifourier} use a larger DSMC sample from independent runs of the same cases ($7.9\times10^{6}$ time samples per cell from 89 runs at $\Kn=0.05$ and $2.0\times10^{6}$ from 26 runs at $\Kn=0.2$).

The models supply the composite tensor \(A_{ij}\) in equation~\eqref{eq:exact_q_balance} in different ways.  R13 does not evolve \(R^{\cl}_{ij}\) and \(\Delta\); they are slaved to gradients of the transported moments through the regularisation closure, which for Maxwell molecules reads, to leading (linear) order \citep{StruchtrupTorrilhon2003,Struchtrup2005},
\begin{equation}
A^{\rm R13}_{ij}=-\frac{\mu}{\rho}\left[\frac{24}{5}\,\partial_{\langle i}q_{j\rangle}+4\,(\partial_kq_k)\,\delta_{ij}\right],
\qquad
m^{\rm R13}_{ijk}=-2\,\frac{\mu}{\rho}\,\partial_{\langle k}\sigma_{ij\rangle},
\label{eq:r13_slaving}
\end{equation}
where \(\mu\) is the shear viscosity, with nonlinear corrections that remain algebraic in the lower moments.  Equation~\eqref{eq:r13_slaving} is the leading term of a gradient expansion and is therefore accurate where the fields vary on the geometric scale \(L\).  In the wall and corner layers of a cavity the relevant scale is \(\lambda\); the slaved terms are then evaluated at gradients of order \(1/\lambda\), formally outside the expansion that produced them, and any closure error there would enter the heat-flux balance through its only fourth-order term, \(\tfrac12\partial_jA_{ij}\) in \eqref{eq:exact_q_balance}.  Because \(\partial_jA^{\rm R13}_{ij}\) consists of second derivatives of \(\qv\), this feedback acts as a local diffusion of heat flux and cannot represent the additional exponential Knudsen-layer modes of the kinetic solution. R26 instead promotes \(m_{ijk}\), \(R^{\cl}_{ij}\) and \(\Delta\) to balance-equation variables with their own wall conditions. These additional variables support Knudsen-layer modes in Kramers, velocity-defect and temperature-jump problems \citep{GuEmersonTang2010,GuEmerson2014,Torrilhon2016}.

\begin{shaded}
Higher moments couple to the velocity through third-moment transport in the stress balance. Near the walls, discretisation and wall conditions also affect the R13--R26 difference. The ratio \(P_\Delta/P_R\approx0.04\) compares the two fourth-order contributions along the heat flux.
\end{shaded}

\begin{landscapefigurepage}
  \enlargethispage{20pt}\noindent\begin{minipage}{0.98\linewidth}
  \centering
  \includegraphics[width=\linewidth,height=0.88\textheight,keepaspectratio]{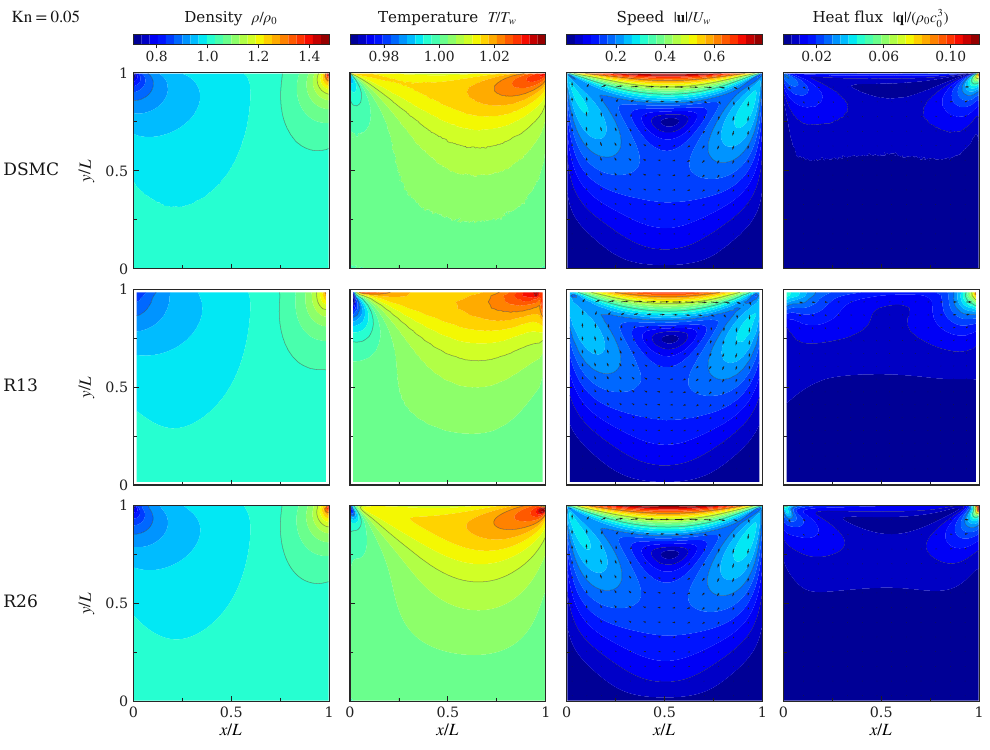}\par
  \captionof{figure}{Primary fields at $\Kn=0.05$. Rows show the sample-weighted DSMC reference, R13 and R26. Columns show density, temperature, speed and heat-flux magnitude, each on a common colour scale for the three methods. The R13 fields are shown over their computed domain. Table~\ref{tab:current_fields} gives the corresponding field errors.}
  \label{fig:kn005_primary}
\end{minipage}
\end{landscapefigurepage}

\begin{landscapefigurepage}
  \enlargethispage{20pt}\noindent\begin{minipage}{0.98\linewidth}
  \centering
  \makebox[\linewidth][c]{\includegraphics[width=575bp]{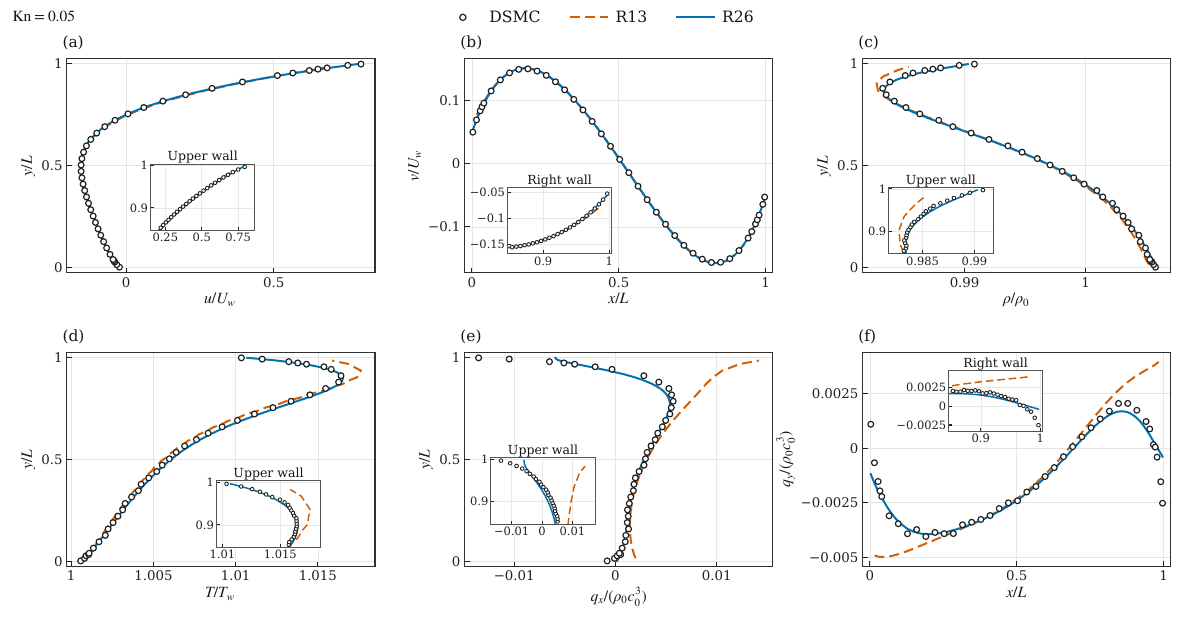}}\par
  \captionof{figure}{Profiles near the cavity centre at $\Kn=0.05$. Symbols show selected DSMC samples and curves show the R13 and R26 profiles on their respective grids. On the $160^2$ comparison grid the sampled row or column lies at $79.5/160=0.496875$. All profiles use a single row or column without spatial averaging or curve fitting. Insets enlarge the region near the upper or right wall.}
  \label{fig:kn005_centerlines}
\end{minipage}
\end{landscapefigurepage}

\begin{landscapefigurepage}
  \enlargethispage{20pt}\noindent\begin{minipage}{0.98\linewidth}
  \centering
  \includegraphics[width=495bp]{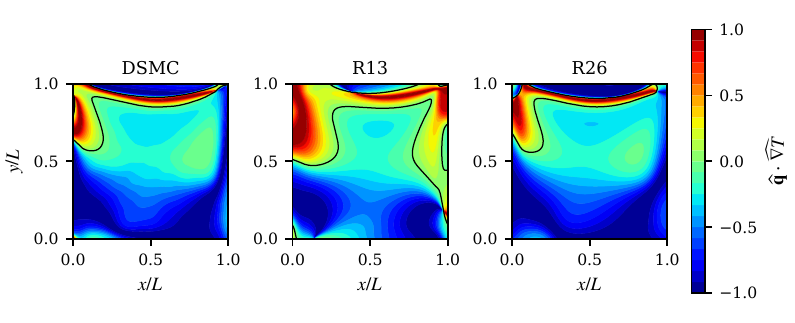}\par
  \captionof{figure}{Alignment of heat flux with the temperature gradient, $\widehat{\boldsymbol q}\cdot\widehat{\nabla T}$, at $\Kn=0.05$ (DSMC, R13, R26), after Gaussian smoothing of width $0.03L$. The black curve marks zero alignment.}
  \label{fig:kn005_antifourier}
\end{minipage}
\end{landscapefigurepage}

\subsubsection{Transition-regime comparison at \(\Kn_{\rm Gu}=0.20\)}
\label{sec:kn020_matched}

\begin{shaded}
The $\Kn=0.2$ case uses the same cavity, diffuse walls, lid speed and Maxwell--VSS transport class ($\omega=1$, $\alpha_{\rm VSS}=2.140$), with the same viscosity calibration as the moment comparison. Both here and at $\Kn_{\rm Gu}=0.05$ (\S\ref{sec:kn005_matched}), the cavity has side $L=10^{-6}$\,m, the walls are at $T_w=300$\,K and the lid moves at $U_w=100$\,m\,s$^{-1}$. The VSS collision diameter is $d_{\rm VSS}=4.661\times10^{-10}$\,m and the viscosity-equivalent diameter in equation~\eqref{eq:kn_gu} is $d_{\rm ref}=4.633\times10^{-10}$\,m, both at $T_{\rm ref}=273$\,K. Both cases use a $160^2$ grid with 256 simulator particles per cell and the time step $\Delta t=\Delta x/[4(4c_{\rm mp}+U_w)]$ with $c_{\rm mp}=(2k_BT_w/m)^{1/2}$. At $\Kn_{\rm Gu}=0.05$ the number density is $n_0=3.45384\times10^{25}\,\mathrm{m^{-3}}$, $\Delta x/\lambda_{\rm Gu}=0.125$ and $\Delta t/\tau_{\rm coll}=1.16\times10^{-2}$, where $\tau_{\rm coll}=\lambda_c/c_{\rm mp}$ and $\lambda_c=(T_w/T_{\rm ref})^{1/2}/(\sqrt2\pi d_{\rm VSS}^2n_0)$. SPARTA generated the random numbers. A restarted run used a new seed and thus followed a different realisation from an uninterrupted run. The seven $\Kn_{\rm Gu}=0.20$ realisations use $n_0=8.63460\times10^{24}\,\mathrm{m^{-3}}$, which gives $\lambda_{\rm Gu}/L=0.200000$ in equation~\eqref{eq:kn_gu}, a $160^2$ grid with $\Delta x/\lambda_{\rm Gu}=0.03125$, 256 simulator particles per cell (6,553,600 in total), $\Delta t/\tau_{\rm coll}=2.90\times10^{-3}$, full diffuse accommodation (accommodation coefficient 1), and a separate random seed for each realisation (appendix~\ref{app:same_state}). For $\omega=1$ both the VSS collision rate and its scattering-angle distribution are independent of the relative speed, so the Prandtl number is $2/3$ and the matched viscosity also fixes the Maxwell-molecule thermal conductivity; the model differs from inverse-power-law Maxwell molecules only in its angular scattering law, which affects higher-order collision productions. The seven runs contribute in proportion to their sample counts. Standard errors estimated for seven equal sampling windows apply only to the equal-weight comparison.
\end{shaded}

For R13, $\Kn_{\rm Rana}=0.1595769122$ corresponds to $\Kn_{\rm Gu}=\sqrt{\pi/2}\Kn_{\rm Rana}=0.2$. We express the R26 velocities and moments in DSMC units by converting from $\sqrt{208T_w}$ to $c_0=\sqrt{k_BT_w/m}$, raising the conversion factor to the velocity degree of each moment. The figures use the $90\times90$ R26 continuation estimate. Below we also compare a $28\times28$ R26 solution from a different discretisation. The R13 fields end at the edge of their native grid.

\begin{landscapefigurepage}
  \enlargethispage{20pt}\noindent\begin{minipage}{0.98\linewidth}
  \centering
  \includegraphics[width=\linewidth,height=0.88\textheight,keepaspectratio]{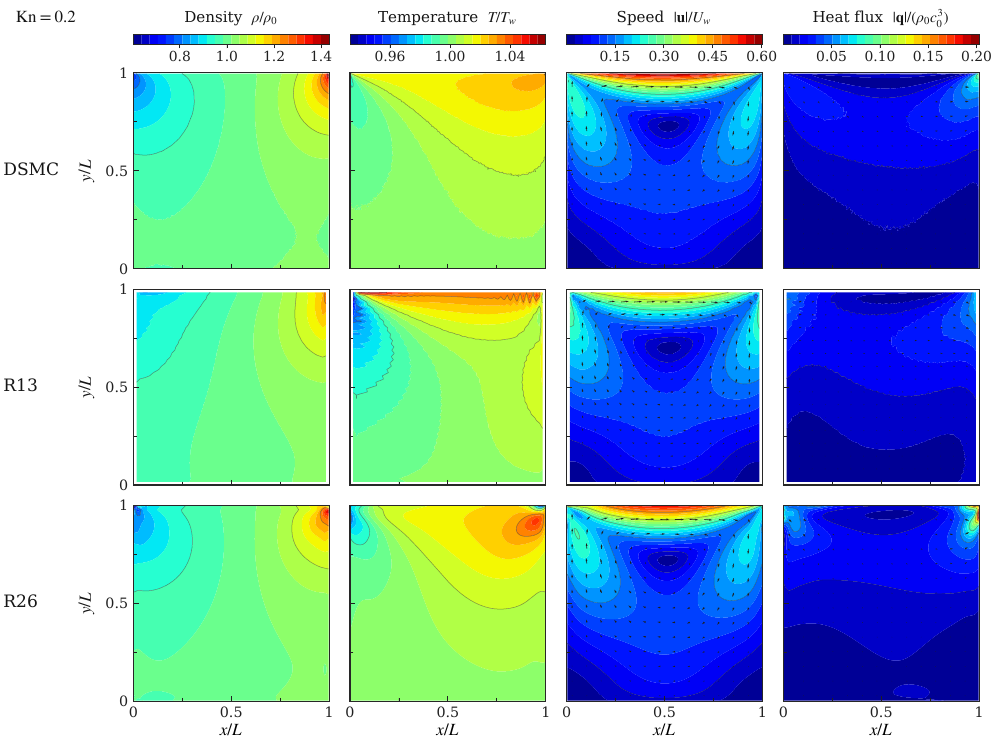}\par
  \captionof{figure}{Primary fields at $\Kn=0.2$. Rows show the sample-weighted DSMC reference, R13 and R26. Columns show density, temperature, speed and heat-flux magnitude, each on a common colour scale for the three methods. The R13 fields are shown over their computed domain.}
  \label{fig:kn020_primary}
\end{minipage}
\end{landscapefigurepage}

\begin{landscapefigurepage}
  \enlargethispage{20pt}\noindent\begin{minipage}{0.98\linewidth}
  \centering
  \makebox[\linewidth][c]{\includegraphics[width=575bp]{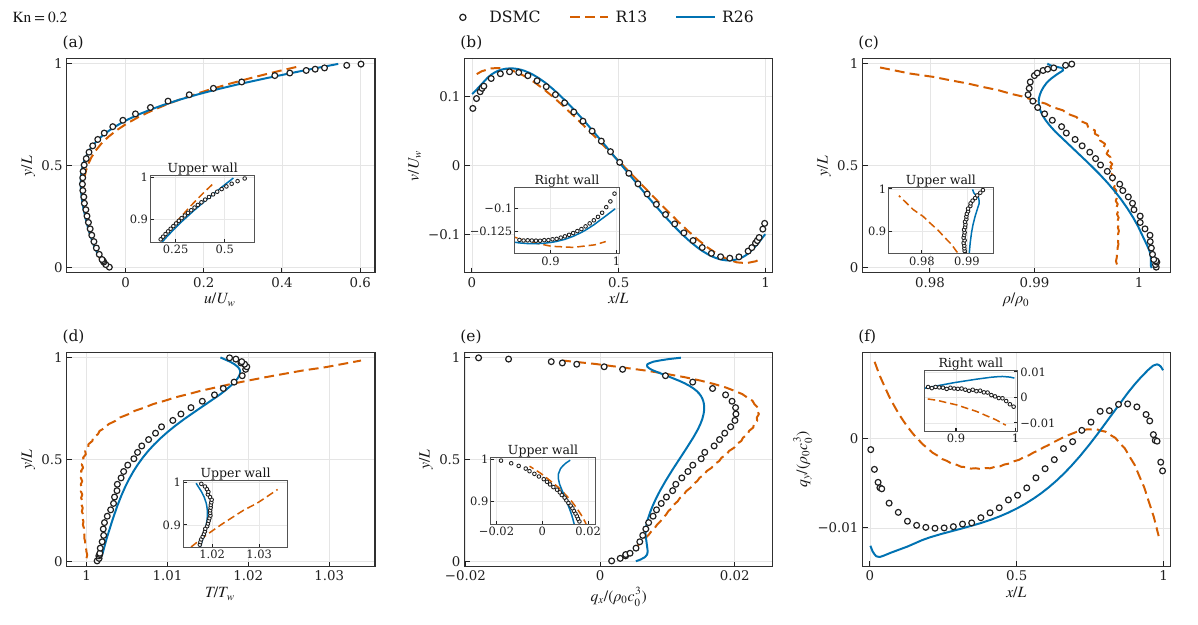}}\par
  \captionof{figure}{Profiles near the cavity centre at $\Kn=0.2$. Symbols show selected DSMC samples and curves show the R13 and R26 profiles on their respective grids. On the $160^2$ comparison grid the sampled row or column lies at $79.5/160=0.496875$. All profiles use a single row or column without spatial averaging or curve fitting. Insets enlarge the region near the upper or right wall.}
  \label{fig:kn020_centerlines}
\end{minipage}
\end{landscapefigurepage}

\input{updated_field_metrics.tex}

Figure~\ref{fig:kn020_antifourier} compares where heat flows towards warmer or cooler gas. Each panel shows the cosine of the angle between the local heat flux and the temperature gradient in that field, after Gaussian smoothing. Resolving near-wall oscillations in the underlying fields requires a grid comparison. Quantitative errors and overlaps also depend on the R26 field and DSMC averaging window being compared.

\begin{landscapefigurepage}
  \enlargethispage{20pt}\noindent\begin{minipage}{0.98\linewidth}
  \centering
  \includegraphics[width=495bp]{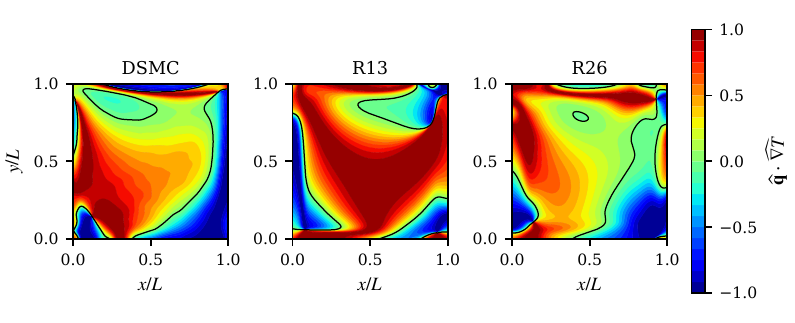}\par
  \captionof{figure}{Alignment of heat flux with the temperature gradient, $\widehat{\boldsymbol q}\cdot\widehat{\nabla T}$, at $\Kn=0.2$ (DSMC, R13, R26), after Gaussian smoothing of width $0.03L$. The black curve marks zero alignment.}
  \label{fig:kn020_antifourier}
\end{minipage}
\end{landscapefigurepage}

\subsubsection{The fourth-order discrepancy and its divergence-free part}
\label{sec:kn020_projection}

\begin{shaded}
Using central differences on the full fields before excluding eight cells at each wall gives a relative divergence discrepancy of 0.617453 between the $28\times28$ R26 solution of the earlier discretisation and the equal-window DSMC mean. We define $E_{{\rm div}A}=\|\Grad_h\cdot(\Aten_{\rm R26}-\Aten_{\rm DSMC})\|/\|\Grad_h\cdot\Aten_{\rm DSMC}\|$, where $\Grad_h\cdot$ is the discrete in-plane divergence, $\Aten_{\rm DSMC}$ is the seven-realisation mean and the norms cover both components on the interior grid. Omitting each realisation in turn gives 0.638--0.639. A bootstrap of 2000 sets of whole realisations gives a 95\% percentile range of 0.652--0.773. Omissions and repeated realisations increase the sampling noise in the DSMC divergence, so these ranges describe sensitivity to sampling. For the $90\times90$ R26 field used to start the grid refinements and its $180\times180$ refinement (appendix~\ref{app:complete_continuation}), the same measure gives 0.600014 and 0.599268. A correction based on independent realisations subtracts an estimated 70.84\% of the observed squared discrepancy and gives an energy ratio of 0.390286. Its bias is unknown, and the amount subtracted depends on the assumptions of the correction. A projection onto 36 smooth test functions gives 0.332596 before and 0.332275 after the correction. These values describe only the chosen functions. The divergence-free projection below uses a different operator applied after the wall cells are removed.
\end{shaded}

\begin{shaded}
We resampled whole DSMC realisations, using the same weights for heat flux and temperature gradient. After excluding eight cells at each wall, pointwise bootstrap percentile intervals include zero for \(\qv\cdot\Grad T\) in 80.18\% of the unsmoothed cells. With three- and five-cell smoothing the respective fractions are 42.73\% and 19.16\%. Smoothing changes the field whose sign is being tested. These pointwise intervals and the overlap of the chosen regions therefore leave the shape of the full counter-gradient region uncertain. A larger sample resolves most of it (\S\ref{sec:updated_evidence}, figure~\ref{fig:canonical_sign}).
\end{shaded}

\begin{shaded}

In the transition-regime DSMC calculation we also sampled the mass-normalised momentum-flux
(pressure-tensor) components \(P_{xx},P_{xy},P_{yy}\) and \(P_{zz}\), where
\(P_{ij}=\int c_i c_j f\,\dd\vvec=\rho\Sigma_{ij}\) defines the
peculiar-velocity covariance \(\Sigma_{ij}\).  The sampled central Sonine set
is \(B_{xx},B_{xy},B_{yy}\) and \(B_{zz}\), formed after subtracting the
instantaneous cell centre of mass and therefore carrying a finite-particle
central-moment bias.  We apply a fixed-count bias correction, with the limitations stated below, before evaluating
equation~\eqref{eq:Bdecomp}.  For a cell containing \(N\) independent
simulator particles, the biased sample fourth central moment \(m^{(4)}_{ij}\) obeys
\begin{equation}
\begin{aligned}
 \mathbb{E}m^{(4)}_{ij}
 &=a_N B_{ij}+\rho b_N\left(\Sigma_{ij}\tr\Sigma
             +2\Sigma_{ik}\Sigma_{kj}\right),\\
 a_N&=\frac{(N-1)(N^2-3N+3)}{N^3},
 & b_N&=\frac{(N-1)(2N-3)}{N^3},
\end{aligned}
\label{eq:finite_N_fourth}
\end{equation}
Here \(m^{(4)}_{ij}\) and \(B_{ij}\) are density-weighted moments. We multiply the sampled per-particle Sonine moment by \(\rho\) before this correction. Pressure is scaled by \(\rho_0c_0^2\), fourth moments by \(\rho_0c_0^4\), and \(\Sigma=P/\rho\). The expectation identity is conditional on fixed density and a fixed number of independent particles; \(\mathbb E\) denotes expectation over repeated samples.  For a single velocity
component, \eqref{eq:finite_N_fourth} reduces to the classical finite-sample
identity for the biased fourth central moment,
\(\mathbb{E}m^{(4)}=[(N-1)(N^2-3N+3)\mu_4+3(N-1)(2N-3)\sigma^4]/N^3\), which
uses the population fourth central moment \(\mu_4\) and variance \(\sigma^2\), fixes \(a_N\) and \(b_N\), and provides an independent check of the tensor
combination.  The primary
correction uses the density-conditioned mean count \(N=256\rho\); repeating
the calculation with fixed \(N=256\) gives the sensitivity to this choice. An average count only approximates the changing particle population. Correlations between colliding particles can leave a residual bias in the time-averaged moments.
The average of the ten independently accumulated block reconstructions differs
from the final corrected tensor by less than \(1.9\times10^{-5}\) in component
root-mean-square (RMS) norm, well below the resolved model discrepancy.
\end{shaded}

Let \(\delta\Aten=\Aten_{\rm R26}-\Aten_{\rm DSMC}\), and pack its in-plane
components with the Frobenius weight as
\(\delta\boldsymbol a=(\delta A_{xx},\sqrt{2}\delta A_{xy},\delta A_{yy})\),
so that the Euclidean norm of \(\delta\boldsymbol a\) equals the Frobenius
norm of the symmetric in-plane tensor, off-diagonal multiplicity included.
On the interior rectangle obtained by removing eight cells from each wall, the
discrete divergence \(D_h\) defines the orthogonal split
\begin{equation}
 \delta\boldsymbol a_H
 =\left[I-D_h^{\mathsf T}(D_hD_h^{\mathsf T})^{\dagger}D_h\right]
   \delta\boldsymbol a,
 \qquad
 \delta\boldsymbol a_V=\delta\boldsymbol a-\delta\boldsymbol a_H .
\label{eq:discrete_projection}
\end{equation}
Here \(I\) is the identity on the packed in-plane grid space, \(\dagger\) denotes the Moore--Penrose pseudoinverse, and the subscripts \(H\) and \(V\) denote hidden (divergence-null) and visible (divergence-generating) parts.  Thus \(D_h\delta\boldsymbol a_H=0\), while
\(\delta\boldsymbol a_V\) is the minimum-norm part capable of changing the
discrete fourth-order divergence.  The \(A_{zz}\) discrepancy is added to the
hidden energy because it is absent exactly from
equations~\eqref{eq:Bx}--\eqref{eq:By}.  The calculation is an integration-like
least-squares projection; no spatial derivative of a DSMC field is used to
construct the hidden tensor.

Table~\ref{tab:kn020_fourth_components} shows that the individual composite
components differ substantially even though the circulation and heat-flux
direction agree much more closely.  The \(A_{xx}\) and \(A_{yy}\) relative RMS
errors are 0.578, and the \(A_{xy}\) error is 0.942.  DSMC \(A_{zz}\) is small,
so its relative error is not a useful scale; moreover, its RMS magnitude
(0.0031) is comparable to the \(1.7\times10^{-3}\) systematic difference
between the two bias corrections, so only an upper bound is claimed for the
DSMC transverse component itself.  Its absolute RMS discrepancy from R26 is
0.0171, equal to 17.4 ensemble standard errors and ten times that correction
systematic.  It contributes 18.2\% of the total fourth-order error energy on
the reference interior domain.

\begin{shaded}
In table~\ref{tab:kn020_fourth_components}, the RMS of the DSMC $A_{zz}$ is at most 13--17\% of that of each in-plane component, and $A_{zz}^2$ is at most about 0.5\% of the DSMC fourth-order energy $A_{xx}^2+2A_{xy}^2+A_{yy}^2+A_{zz}^2$ evaluated with the RMS values; the corresponding share for R26 is 10\%. At $\Kn_{\rm Gu}=0.20$ the transverse component is therefore a small part of the sampled fourth-order state but a tenth of the R26 one.
\end{shaded}

\begin{center}
\begin{minipage}{\textwidth}
\centering
\small
\setlength{\tabcolsep}{7pt}
\renewcommand{\arraystretch}{1.14}
\rowcolors{1}{yellow!40}{yellow!40}
\begin{tabular}{lrrrrr}
\toprule
component & DSMC RMS & R26 RMS & relative RMS error & uncentered \(C\) & error/SE \\
\midrule
\(A_{xx}\) & 0.0208 & 0.0151 & 0.578 & 0.822 & 12.3 \\
\specialrule{0.25pt}{1.8pt}{1.8pt}
\(A_{xy}\) & 0.0244 & 0.0361 & 0.942 & 0.777 & 35.7 \\
\specialrule{0.25pt}{1.8pt}{1.8pt}
\(A_{yy}\) & 0.0181 & 0.0159 & 0.578 & 0.818 & 10.8 \\
\specialrule{0.25pt}{1.8pt}{1.8pt}
\(A_{zz}\) & 0.00314 & 0.0189 & 5.44 & 0.637 & 17.4 \\
\bottomrule
\end{tabular}
\captionsetup{format=revisioncaption}
\captionof{table}{Finite-particle-corrected composite fourth-order comparison at
\(\Kn_{\rm Gu}=0.20\).  SE denotes the ensemble standard error of the seven
DSMC realisations.  The large relative \(A_{zz}\) value reflects its small
DSMC RMS. The absolute discrepancy of 0.0171 is 17.4 ensemble standard errors. The uncentred correlation and the errors use the same interior cells, eight cells from each wall. With \(G=\|A_{ij}^{\rm R26}\|/\|A_{ij}^{\rm DSMC}\|\), the unrounded values satisfy \(E^2=1+G^2-2GC\) to \(3.6\times10^{-15}\).}
\label{tab:kn020_fourth_components}
\end{minipage}
\end{center}

The reference projection places 95.1\% of the in-plane error energy in
\(\ker D_h\); adding \(A_{zz}\) raises the hidden fraction to 96.0\%.
Across 48 combinations of correction, wall crop, smoothing width and
divergence operator, the fraction with \(A_{zz}\) remains 95.0--97.3\% (in-plane 93.1--97.1\%).  Separate seed
projections give \(0.9540\pm0.00015\) with \(A_{zz}\) (mean \(\pm\) standard error).

\begin{shaded}
These fractions must be judged against the null-space content of
structureless fields, because the discrete kernel is itself a substantial
subspace.  The in-plane kernel of the reference operator occupies 0.352 of
the packed degrees of freedom, and five isotropic Gaussian tensor fields
projected through it give hidden in-plane fractions of \(0.354\pm0.001\).
All twenty-one pairwise seed-difference DSMC noise fields,
\((\Aten_i-\Aten_j)/\sqrt{2}\), where \(\Aten_i\) is the reconstructed tensor from realisation \(i\), projected through the identical central and
forward operators, give in-plane hidden fractions of 0.352--0.377 and total
fractions of 0.521--0.541. Only 4.9\% of the in-plane R26--DSMC discrepancy
changes the discrete divergence, compared with \(\approx\!63\%\) of the sampling noise.
The visible share of the model discrepancy is thus thirteen times smaller,
and its hidden share is about 2.6 times that of the noise.
\end{shaded}

\begin{shaded}
The discrete divergence on the interior rectangle admits more perturbations than the Airy fields in \eqref{eq:boundary_bump}. Its null space also contains uniform tensors and fields with nonzero values on the boundary of the cropped region. For the equal-window DSMC estimate the mean discrepancy is $(\delta A_{xx},\delta A_{xy},\delta A_{yy},\delta A_{zz})=(0.0054,0.0160,0.0049,-0.0067)$ in units of $\rho_0c_0^4$. This uniform tensor accounts for 45.4\% of the hidden in-plane energy (41--52\% over the 48 choices of correction, domain, smoothing and differentiation). After subtracting the mean, 91.3\% of the spatially varying discrepancy remains hidden in the plane and 93.5\% when $A_{zz}$ is included. The visible in-plane share of 8.7\% is about seven times smaller than the $\approx\!63\%$ for sampling noise. These fractions describe this discrete operator and these fields. Their continuum limit and the physical cause of the concentration in the null space remain open.
Figure~\ref{fig:kn020_fourth_order} shows why the divergence-free part matters in this cavity. In the lower row it retains nearly all the structure of the in-plane discrepancy. The remainder changes the divergence and has much smaller amplitude.
The upper row shows a different effect. The DSMC \(A_{zz}\) is small and nearly featureless, with the transverse fourth moment close to the value implied by the scalar excess. The R26 \(A_{zz}\) has six times the DSMC RMS amplitude and varies across the cavity. This component enters the coupled moment equations even though it drops out of the explicit planar divergence. The difference includes numerical, sampling, boundary and collision-model effects as well as any closure error.
\end{shaded}

The additional moments in R26 improve predictions for canonical Couette and Poiseuille flows \citep{GuEmerson2009}. Their benefit in a cavity depends on the quantity measured. For example, we compared our R26 centreline velocities with those of \citet{YangTangYang2019} at their \(\Kn=0.1\) (appendix~\ref{app:yang_validation}). Relative errors are 3.9\% for \(u_x\) and 6.2\% for \(u_y\), with correlations of 0.99962 and 0.99986. Further grid refinement changed the velocities by less than 0.7\%.

\begin{shaded}
The published R13 and R26 velocity profiles for plane Poiseuille flow are reproduced within 0.2\%, and the cavity velocities at their $\Kn=0.01$ within 2\% (appendix~\ref{app:yang_validation}). The heat flux next to the walls differs. At $\Kn=0.01$ the R26 equations give about half the kinetic value at the walls, as does the exact solution for plane Couette flow, and at $\Kn=0.1$ much less at the lid ($-0.47$ against $-6.40$ in the units of table~\ref{tab:pub_cavity}). The published R26 values equal the kinetic ones at the lid and side walls at $\Kn=0.01$, though not at the floor ($-0.041$ against $-0.152$), and are close to them at the lid at $\Kn=0.1$ (table~\ref{tab:pub_cavity}). A finite-volume treatment that evaluates the wall relations in the cell next to the wall brings the lid value close to the published one, but it changes the boundary-value problem and does not improve the agreement with DSMC away from the walls (appendix~\ref{app:pub_walltreat}). The published profiles therefore support the velocity calculation and the implementation of the equations, and the heat-flux results here rely on the DSMC comparisons.

Transport of additional moments in R26 and its gradient closures could account for some of the differences in the wall layer. The comparison also depends on collision coefficients, discretisation, sampling and the wall conditions, so it cannot isolate regularisation error. In \S\ref{sec:updated_evidence} we test the closures with the retained moments taken directly from DSMC.
\end{shaded}

Heat flux is a third-order velocity moment whose balance contains divergences of fourth-order moments, so it is especially sensitive to errors in the non-equilibrium tails and to the wall values of the added moment fields.  In a lid-driven cavity, the discontinuous change from a stationary side wall to the moving lid creates interacting wall, corner and Knudsen layers precisely where kinetic-distribution errors are largest \citep{GuEmerson2014}.  R26 transports additional moments whose near-wall values depend on the collision coefficients and wall closure \citep{RanaGuptaSprittlesTorrilhon2021,CaiWang2020}.  Conversely, an R13 closure can exhibit partial error cancellation for one heat-flux profile without recovering the underlying higher moments.  At both Knudsen numbers in table~\ref{tab:current_fields}, the R26 heat-flux field is closer to DSMC than the R13 field. At $\Kn_{\rm Gu}=0.05$ grid resolution accounts for much of the R13 discrepancy. On finer grids it falls from 0.6415 to 0.3016 at $180^2$, still more than twice the R26 value on its $60^2$ grid (appendix~\ref{app:complete_continuation}).

\begin{landscapefigurepage}
  \enlargethispage{20pt}\noindent\begin{minipage}{0.98\linewidth}
  \centering
  \includegraphics[width=505bp]{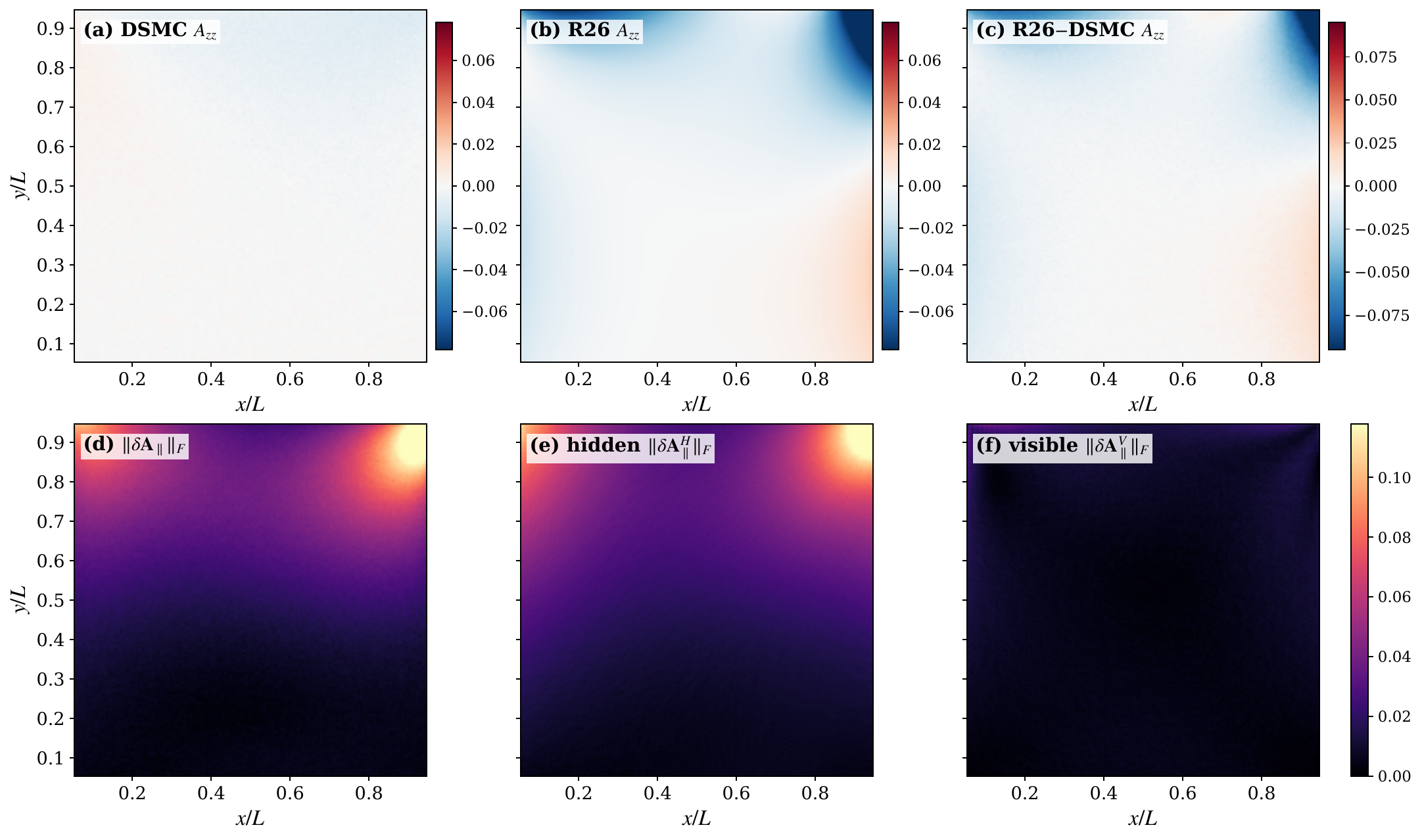}\par
  \captionof{figure}{Fourth-order moments for the equal-window DSMC ensemble and the $28\times28$ solution of the earlier R26 discretisation. The DSMC moments include the finite-particle correction. The
  null-space projection uses \(\Kn_{\rm Gu}=0.20\). The upper row shows
  the transverse component \(A_{zz}\) from DSMC, R26 and their difference.
  The lower row shows the Frobenius magnitude of the in-plane R26--DSMC discrepancy,
  its divergence-free projection and the minimum-norm remainder that produces its divergence.
  The fields cover the reference interior domain without smoothing.
  Panels (d)--(f) show that the divergence-free component accounts for most of the squared in-plane tensor difference. The divergence error is evaluated separately.}
  \label{fig:kn020_fourth_order}
\end{minipage}
\end{landscapefigurepage}

\begin{shaded}
A divergence-free perturbation makes no instantaneous contribution to the fourth-order term in the heat-flux balance. The higher-moment equations and wall conditions can still couple it to \(m_{ijk}\), stress and heat flux. A tangent-linear test of the complete discrete residual for the \(90\times90\) continuation field at \(\Kn_{\rm Gu}=0.20\) (appendix~\ref{app:complete_response}) shows that the third- and fourth-order equations respond to the out-of-plane and Airy perturbations of \S\ref{sec:nullspace}. Heat-flux measurements alone miss these directions. To quantify their effect as the flow evolves would require a pair of converged R26 evolutions.
\end{shaded}

In one dimension, measuring the scalar excess separates the two fourth-order contributions. In the cavity the same measurement leaves the Airy freedom unresolved.  The in-plane Airy freedom requires information that constrains the tensor field itself, including higher-moment balances and their wall traces.  Maxwell-type and thermodynamically admissible R13 wall conditions, linear regularized 26-moment (LR26) \(H\)-theorem boundary conditions, and more recent Onsager formulations all treat the boundary map as part of the physical closure \citep{GuEmerson2007,RanaStruchtrup2016,RanaGuptaSprittlesTorrilhon2021,CaiTorrilhonYang2024,LinWangYangCai2025}.  Fundamental-solution studies likewise show that regularization introduces explicit Knudsen-layer modes whose amplitudes are fixed through the wall conditions \citep{ClaydonShresthaRanaSprittlesLockerby2017}.  The compactly supported Airy construction above leaves the flux-side wall trace unchanged, but it is not a proof of compatibility with any one of these complete coupled wall maps.  A well-posed R13 or R26 boundary-value problem is consistent with the DSMC construction. The heat-flux vector alone cannot supply the missing wall and higher-moment information.

\subsection{Heat-flux agreement and fourth-order error}\label{sec:interpretation}

Once the lower-order terms are known, the heat-flux balance determines \(\partial_jA_{ij}\). The five eight-realisation VHS comparisons at \(\Kn_{\rm Gu}=0.05\) show that the tensorial fourth-order projection is larger than the scalar projection across the grids and particle counts studied. We sampled all ten symmetric components of \(m_{ijk}\). Cell-wise estimates of \(\Delta\) remain particularly sensitive to sampling. Comparison with the independent cavity calculation of Yang and co-workers likewise distinguishes velocity, which our R26 calculation reproduces closely, from their near-wall heat flux, where the profiles differ (appendix~\ref{app:yang_validation}).

At \(\Kn_{\rm Gu}=0.20\) we compare the R26 solution on the finest grid in the \(25^2\)--\(28^2\) sequence with the seven-realisation Maxwell--VSS mean and its finite-particle-corrected fourth moment. R26 agrees more closely with DSMC than R13 in velocity and heat-flux orientation. For R26 the relative heat-flux error \(E_q=0.407\), measured with equal sampling windows, still exceeds the DSMC sampling uncertainty. For the same R26 solution the error is 0.3021 with the sample weighting and domain in table~\ref{tab:current_fields} (\S\ref{sec:updated_evidence}). The tensor projection measures the part of the discrepancy to which the planar heat-flux balance is insensitive.

\begin{shaded}
This refinement sequence ends at \(28^2\). Section~\ref{sec:updated_evidence} and appendix~\ref{app:complete_continuation} also consider the \(90\times90\) continuation estimate from the other R26 discretisation and its \(120\times120\) and \(180\times180\) refinements.
\end{shaded}

\subsection{\texorpdfstring{\colorbox{white}{\parbox{0.75\linewidth}{Constraints from the wall conditions and moment equations}}}{Constraints from the wall conditions and moment equations}}\label{sec:boundary_scope}

\begin{shaded}
The compactly supported Airy perturbations vanish near every wall. They therefore preserve the wall values entering the fourth-order heat-flux term. A physical cavity solution must satisfy the other moment equations, Knudsen-layer matching and the coupled wall conditions, and its moments must admit a nonnegative distribution. These requirements can restrict the perturbations further. The R26 calculations here use the wall relations of Gu \& Emerson.
\end{shaded}

\begin{shaded}
To test the fourth-order contribution to heat flux, the other terms in the balance must first be constrained independently. With the lower moments known, one can compare $\tfrac12\partial_jA_{ij}$ directly and in weak form (\S\S\ref{sec:kn020_projection} and \ref{sec:updated_evidence}). Higher-moment samples allow a separate test of the closure relations on the same retained state (table~\ref{tab:new_closure}). At the wall the relevant data are the fourth-order inputs to the wall relations (appendix~\ref{app:finite_wall_coupling}). Agreement in stress direction offers a further check, although it cannot determine the tensor from its divergence. Recovering the full tensor requires additional moment information and boundary constraints, with an assessment of the uncertainty in each. Once the tensor is known, its trace and trace-free parts are unique. The numerical comparisons here apply to the cavity configurations studied.
\end{shaded}

\section{\texorpdfstring{\colorbox{white}{\parbox{0.75\linewidth}{R26 closures evaluated on measured DSMC moments}}}{R26 closures evaluated on measured DSMC moments}}
\label{sec:updated_evidence}

\begin{shaded}
We now evaluate the closures on the lower moments of the same gas sample used to measure the higher moments. Each of the seven DSMC realisations contributes 20,000 samples between steps two million and four million. We sum the raw moments in the laboratory frame before transforming to a common local mean velocity. This gives all 56 moments through degree five on the native $160\times160$ grid, including the sampled components odd in $z$.

The three DSMC averages used at $\Kn_{\rm Gu}=0.20$ differ in the order of averaging and in the sampling window. For the equal-weight mean in \S\S\ref{sec:kn020_matched}--\ref{sec:kn020_projection} and table~\ref{tab:kn020_fourth_components}, each realisation was sampled every 100 steps for $2\times10^6$ steps after $2\times10^5$ warm-up steps. Its fourth moments were formed about the instantaneous cell mean velocity and corrected for finite particle count (table~\ref{tab:kn020_fourth_components}) before the seven fields were averaged. The macroscopic fields in table~\ref{tab:current_fields} use the unequal windows of \S\ref{sec:model_diagnostic}, weighted by the number of samples. The closure calculations below use the common late window described above. Averaging bias and differences in the collision laws remain possible.
\end{shaded}

\begin{shaded}
We use the published Gu--Maxwell coefficients to calculate $\Phi$, $\psi$ and $\Omega$ from the retained DSMC moments. These three quantities have tensor ranks four, three and one and velocity degrees four, five and five, respectively. We take central differences over one cell spacing and exclude the outer two rows on each side, leaving 24,336 cells. For each family the relative discrepancy is $E_1=\|Z-C\|/\|Z\|$, with $Z$ the directly measured higher moment and $C$ the closure prediction. Each norm includes the full Cartesian multiplicities in three dimensions. Since both $Z$ and $C$ are evaluated on the same DSMC field, this test removes differences due to solving a separate R26 flow. Sampling, differentiation and differences in the collision law can still affect it.
\end{shaded}

\begin{center}
\begin{minipage}{\linewidth}
\centering
\rowcolors{1}{yellow!40}{yellow!40}
\begin{tabular}{lrrrrrr}
\toprule
 & \multicolumn{4}{c}{$\Kn_{\rm Gu}=0.20$} & \multicolumn{2}{c}{$\Kn_{\rm Gu}=0.05$} \\
Family & $E_1$ & floor & $V/M$ & $\|B\|^2/M$ & $E_1$ & floor \\
\midrule
$\Phi$ & 2.291119 & 0.71 & 0.364071 & $1.32706\times10^{-7}$ & 1.027 & 0.20 \\
$\psi$ & 2.040090 & 1.06 & 0.616889 & $3.94784\times10^{-7}$ & 0.923 & 0.34 \\
$\Omega$ & 4.810683 & 2.77 & 0.665947 & $4.83695\times10^{-7}$ & 2.194 & 0.74 \\
\bottomrule
\end{tabular}
\captionsetup{format=revisioncaption}
\captionof{table}{R26 closure predictions compared with the higher moments measured in the same pooled DSMC sample. Derivatives use one grid spacing. The floor is $E_1$ for an exact closure perturbed by sampling noise, using the same derivatives, cells and weights as the DSMC comparison (\S\ref{sec:updated_evidence}). At $\Kn_{\rm Gu}=0.05$ the noise is estimated from the seven realisations at that Knudsen number. The last two columns for $\Kn_{\rm Gu}=0.20$ give the scatter among realisations and the effect of pooling, each normalised by the mean squared discrepancy. Appendix~\ref{app:same_state} gives the definitions and statistical limitations.}
\label{tab:new_closure}
\end{minipage}
\end{center}

\begin{shaded}
The differences between measured and predicted moments also appear in the contraction of the fourth moment with $\nabla\uv$ and in the gradient of $\Omega$ in the bulk equations. Appendix~\ref{app:paired_bulk} evaluates these terms on the same fields with two derivative stencils. Widening the centred difference from one to two cell spacings reduces the dimensionless forcing norm from 18.5 to 9.0 in the $R$ equation and from 24.5 to 9.2 in the $\Delta$ equation (table~\ref{tab:paired_bulk_forcing}). For $\Phi$, $\psi$ and $\Omega$ the discrepancies with the wider stencil are 2.07, 1.76 and 4.22 (appendix~\ref{app:same_state}), respectively 9.6\%, 13.9\% and 12.3\% below those in table~\ref{tab:new_closure}. Their ranking and order of magnitude persist across the two stencils. The change measures sensitivity to differentiation at fixed resolution. These terms are only part of the bulk balances, so a bound on the residual or the continuum error would require further information.
\end{shaded}

\begin{shaded}
To estimate how much of the measured discrepancy sampling alone might produce, we interpolated the R26 state at the same Knudsen number to the DSMC cells and treated its own closure values as the measured higher moments. To every field we added the sampling noise of the pool, taken from the measured difference between each realisation and the pool and scaled by $[n_s/(N-n_s)]^{1/2}$, where $n_s$ and $N$ are the samples per cell of realisation $s$ and of the pool, so that its variance equals that of the pool. With the same derivative, mask and weights, and normalised by the DSMC moments, the exact closure gives $E_1$ = 0.71, 1.06 and 2.77 (column `floor' of table~\ref{tab:new_closure}), between 31\% and 58\% of the DSMC values. With differences over two cell spacings the floor falls to 0.36, 0.53 and 1.39, between 17\% and 33\% of the DSMC values, whereas the DSMC values decrease only by the 10--14\% quoted above. Sampling noise alone would give a decrease comparable to that of the floor. Most of the discrepancy persists and reflects a mismatch between the published closure and the simulated gas, including the difference between the Maxwell and VSS collision models.

The same test with the seven realisations at $\Kn_{\rm Gu}=0.05$ gives $E_1$ = 1.03, 0.92 and 2.19, 2.7 to 5.1 times its floor (with the wider stencil, 0.85, 0.81 and 2.20 against 0.12, 0.19 and 0.43). The mismatch is about half as large at the lower Knudsen number and remains well above the noise. The change with rarefaction may reflect moment truncation, the difference between VSS and Maxwell angular scattering in the higher-order collision terms, and the DSMC discretisation. The two cases alone cannot apportion these effects. The floor also depends on the use of an R26 field, smoother than the sampled gas, to approximate the true state. This can change the floor by some tens of per cent, less than the factor of 1.7 to 7 separating it from the DSMC values.
\end{shaded}

\begin{shaded}
The pooled DSMC moments allow us to locate the discrepancy within the heat-flux balance. For steady flow of Maxwell molecules the collision term of \eqref{eq:exact_q_balance} is $-2pq_i/(3\mu)$; inserting \eqref{eq:Bdecomp}, removing the pressure gradient with the steady momentum balance and dividing by $2p/(3\mu)$ writes the heat flux as a sum of parts,
\begin{equation}
\begin{split}q_i={}&q^{F}_i+q^{\sigma}_i+q^{R}_i+q^{\Delta}_i+q^{N}_i,\qquad q^{F}_i=-\frac{15}{4}\mu\partial_iT,\\ &\phantom{{}={}}q^{\sigma}_i=-\frac{3\mu}{2p}T\partial_j\sigma_{ij},\quad q^{R}_i=-\frac{3\mu}{4p}\partial_jR^{\cl}_{ij},\quad q^{\Delta}_i=-\frac{\mu}{4p}\partial_i\Delta,\end{split}
\label{eq:q_parts}
\end{equation}
where $q^{N}_i$ collects the products of lower-order fields with gradients of $T$, $\rho$, $\sigma_{ij}$, $q_i$ and $u_i$, including $m_{ijk}\partial_ju_k$. We evaluated every part from the pooled DSMC moments and from the R26 and R13 states, all on the DSMC cells. The fields were averaged over blocks of $4\times4$ cells before differentiation: on single cells the sampling noise of the fourth moments dominates $\Delta$ and the sum of the parts, whereas the projections in table~\ref{tab:q_parts} change by less than 0.03 between single cells and $8\times8$ blocks. For each part, table~\ref{tab:q_parts} gives its projection $\langle q^{\alpha}_iq_i\rangle/\langle q_iq_i\rangle$ on the heat flux of the same solution over $0.1\le x/L,\,y/L\le0.9$. For an exact steady Maxwell balance the dimensionless projections sum to one. With $4\times4$-cell averaging, their measured sum differs from one by at most 0.02 in each of the three fields.

In the gas the stress part carries the heat flux and the fourth-order part $q^{R}$ opposes it. At $\Kn_{\rm Gu}=0.20$ the projection of $q^{\sigma}$ is 2.34, $q^{R}$ removes 1.18 of it, and the Fourier part, $-0.21$, is small and of opposite sign to the heat flux. R26 reproduces the two large parts within 15\% ($q^{R}$: $-1.33$ against $-1.18$). In the projection on its own heat flux, R13 gives only about half of $q^{R}$ ($-0.63$) and a Fourier part more than three times too large ($-0.70$); the two errors partly cancel in the sum, so a comparison of the heat flux alone does not reveal them. At $\Kn_{\rm Gu}=0.05$ R13 loses $q^{R}$ almost completely ($-0.04$ against $-0.22$), and the region of cold-to-hot transfer, where $q_i\partial_iT>0$, is 2.4 times too large (0.21 of the area against 0.09); R26 overestimates $q^{R}$ by about 30\% ($-0.28$) and gives a smaller region (0.05). The largest R13 error is therefore in the part supplied by its algebraic relation for $R^{\cl}_{ij}$, and it is present already at the lower Knudsen number; R26, which solves a field equation for $R^{\cl}_{ij}$, recovers the part that opposes the stress-driven heat flux. The $\Delta$ part is small in all three solutions; R26 gives it the wrong sign at $\Kn_{\rm Gu}=0.20$ ($+0.10$ against $-0.04$). The decomposition uses the Maxwell-molecule collision rate for all three solutions. The viscosity law is the same, but the different angular scattering of the VSS gas in the DSMC can change the collisional production of the higher moments and, through it, the measured fields and gradients; this effect is not separated here.
\end{shaded}
\begin{center}
\begin{tabular}{lrrrrrr}
\toprule
 & \multicolumn{3}{c}{$\Kn_{\rm Gu}=0.20$} & \multicolumn{3}{c}{$\Kn_{\rm Gu}=0.05$} \\
Part & DSMC & R26 & R13 & DSMC & R26 & R13 \\
\midrule
Fourier, $q^{F}$ & $-0.21$ & $-0.36$ & $-0.70$ & $+0.13$ & $+0.17$ & $-0.01$ \\
stress, $q^{\sigma}$ & $+2.34$ & $+2.52$ & $+2.30$ & $+1.07$ & $+1.10$ & $+1.03$ \\
$R^{\cl}_{ij}$, $q^{R}$ & $-1.18$ & $-1.33$ & $-0.63$ & $-0.22$ & $-0.28$ & $-0.04$ \\
$\Delta$, $q^{\Delta}$ & $-0.04$ & $+0.10$ & $-0.10$ & $-0.02$ & $-0.03$ & $-0.01$ \\
nonlinear, $q^{N}$ & $+0.09$ & $+0.08$ & $+0.14$ & $+0.04$ & $+0.04$ & $+0.04$ \\
\midrule
cold-to-hot area & 0.85 & 0.77 & 0.85 & 0.09 & 0.05 & 0.21 \\
\bottomrule
\end{tabular}
\captionof{table}{Parts of the heat flux, equation~\eqref{eq:q_parts}: projection $\langle q^{\alpha}_iq_i\rangle/\langle q_iq_i\rangle$ of each part on the heat flux of the same solution over $0.1\le x/L,\,y/L\le0.9$, after averaging over blocks of $4\times4$ DSMC cells. The last row is the fraction of that area where $q_i\partial_iT>0$.}\label{tab:q_parts}
\end{center}

\begin{shaded}
Figures~\ref{fig:q_parts_profiles}--\ref{fig:q_parts_maps005} show where each part acts. They use the larger DSMC sample of figures~\ref{fig:kn005_antifourier} and~\ref{fig:kn020_antifourier}, which is independent of the pool in table~\ref{tab:q_parts} and gives the same projections to within 0.002; the cold-to-hot area at $\Kn_{\rm Gu}=0.05$ is 0.07 against 0.09. On the horizontal centreline (figure~\ref{fig:q_parts_profiles}) the sum of the five parts falls on the measured heat flux. At the centre of the cavity at $\Kn_{\rm Gu}=0.20$ the stress part of the DSMC field is 2.0 in units of the r.m.s. heat flux $q_{\rm ref}$ of that field, the Fourier part is $-0.7$ and $q^{R}$ is $-0.4$; their sum, 0.8, is the heat flux. R26 gives 2.3, $-0.8$ and $-0.75$ for the same parts. In R13 $q^{R}$ is small away from the walls ($-0.15$) and the Fourier part is more negative ($-0.9$), which is the compensation noted above. At $\Kn_{\rm Gu}=0.05$ the fourth-order parts are confined to within about $0.15L$ of the side walls, and in the interior the heat flux is the sum of the stress and Fourier parts.

Figures~\ref{fig:q_parts_maps020} and~\ref{fig:q_parts_maps005} map the component of each part along the temperature gradient, $q^{\alpha}_i\partial_iT/|\nabla T|$, for the DSMC field. A positive value means that the part carries heat towards hotter gas; the Fourier part is negative by definition. The stress part is positive over 99\% of the central square at $\Kn_{\rm Gu}=0.20$ and 90\% at $\Kn_{\rm Gu}=0.05$, and is largest under the lid. At $\Kn_{\rm Gu}=0.20$ the fourth-order part is negative over 95\% of the square. Heat flows from cold to hot where the stress part exceeds the other two, which is most of the cavity at $\Kn_{\rm Gu}=0.20$ and a region next to the upstream wall at $\Kn_{\rm Gu}=0.05$.

The stress part is the counterpart in the moment equations of the velocity-curvature term of \eqref{eq:q3}: for a Newtonian stress, incompressible flow and uniform viscosity, $\partial_j\sigma_{ij}=-\mu\nabla^2u_i$ and $q^{\sigma}_i=(3\mu^2/2\rho)\nabla^2u_i$. Molecules arriving at a point from opposite sides carry the velocity of the layers they left; when the velocity profile is curved the two contributions do not cancel, and a heat flux remains without a temperature gradient \citep{Sone2007,MahdaviRoohi2022}. The moment decomposition adds two results. The third-order moment $m_{ijk}$ enters the heat-flux balance directly only through the nonlinear term $m_{ijk}\partial_ju_k$, whose projection is 0.04 at $\Kn_{\rm Gu}=0.20$. The fourth-order tensor $R^{\cl}_{ij}$ opposes the stress part: in the projection of table~\ref{tab:q_parts} it removes half of it at $\Kn_{\rm Gu}=0.20$ (a quarter at the centre of the cavity) and a fifth at $\Kn_{\rm Gu}=0.05$. At $\Kn_{\rm Gu}=0.20$ the heat flux is thus the remainder of larger opposing parts. \citet{AkhlaghiRoohiStefanov2018} reached a related result by separating ballistic from collisional molecules: the two contributions to the heat flux have similar magnitudes and different directions, and cold-to-hot transfer results from their interplay.

Figure~\ref{fig:apriori_closure} separates the closure from the solution. We evaluated the R13 relations for $R^{\cl}_{ij}$ and $\Delta$ on the DSMC fields themselves and formed the fourth-order part $q^{R}+q^{\Delta}$ from them; values are in units of $q_{\rm ref}$ of the DSMC field. At $\Kn_{\rm Gu}=0.20$ the result follows the R13 solution (correlation 0.93 over the central square) and not the part measured in DSMC (0.65). At the centre of the cavity it is $-0.15$ against the measured $-0.43$ (R13 solution $-0.18$), and at $x/L=0.95$ it is 1.5 where the measured part is 0.2. The R13 error there lies in the algebraic relation. At $\Kn_{\rm Gu}=0.05$ the relation misses the layer at the upstream wall (measured 0.22 at $x/L=0.05$, R26 0.20), but at the downstream wall it follows the measurement and the R13 solution does not; that part of the error comes from the R13 fields. Projected on the DSMC heat flux, $q^{R}$ of R13 at $\Kn_{\rm Gu}=0.20$ is $-1.10$ against $-1.18$: R13 errs in the distribution of this part, too small in the interior and too large at the side walls, more than in its total. R26 reproduces the measured distribution in the interior (correlation 0.95), with a magnitude about 1.6 times too large at the centre ($-0.67$) and too small at the side walls.
\end{shaded}

\begin{figure}[tbp]
\centering
\includegraphics[width=\linewidth]{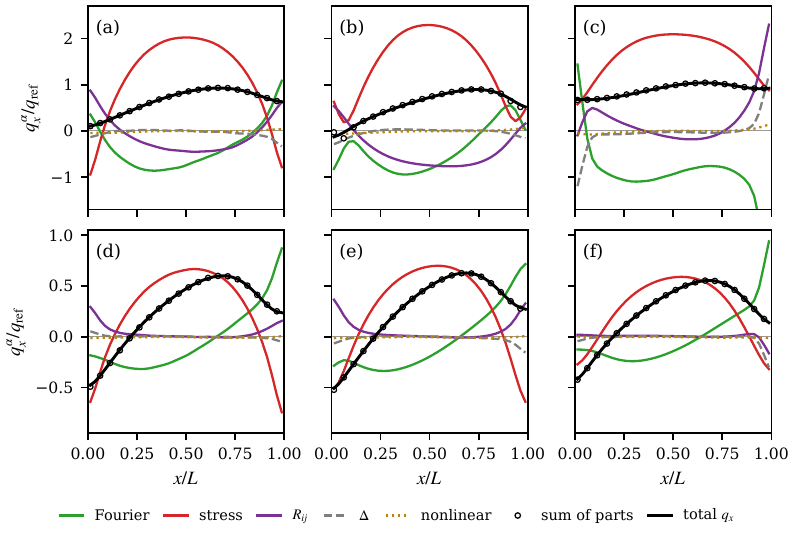}
\captionsetup{format=revisioncaption}
\caption{Parts of $q_x$, equation~\eqref{eq:q_parts}, on the horizontal centreline $y=L/2$ for (a,d) DSMC, (b,e) R26 and (c,f) R13, in units of the r.m.s. heat flux $q_{\rm ref}$ of each field over $0.1\le x/L,\,y/L\le0.9$. Upper row $\Kn_{\rm Gu}=0.20$, lower row $\Kn_{\rm Gu}=0.05$. Symbols show the sum of the five parts. Fields are averaged over $4\times4$ cells and smoothed over $0.03L$.}
\label{fig:q_parts_profiles}
\end{figure}

\begin{figure}[tbp]
\centering
\includegraphics[width=\linewidth]{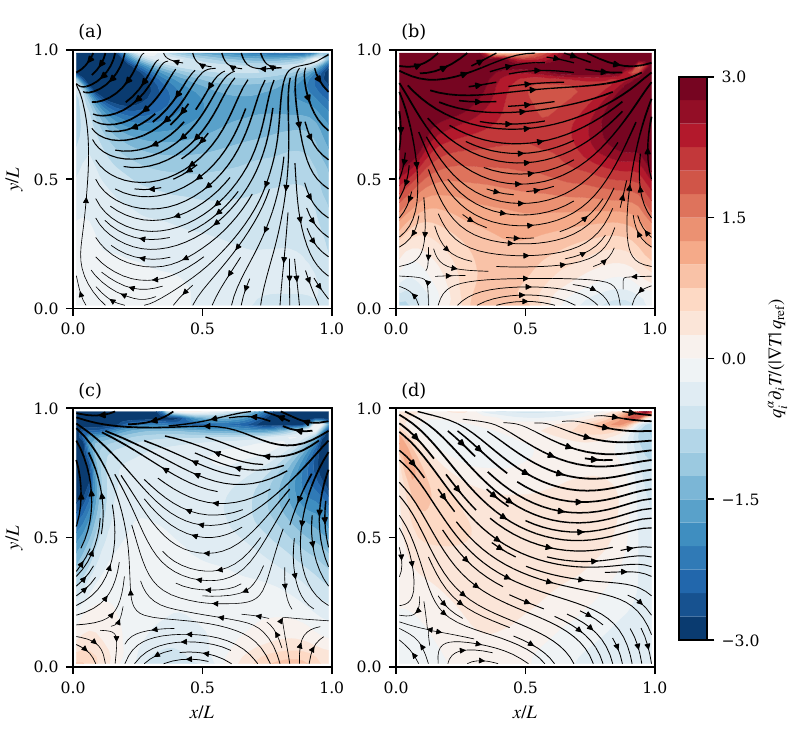}
\captionsetup{format=revisioncaption}
\caption{Component along the temperature gradient of (a) the Fourier part, (b) the stress part, (c) the fourth-order part $q^{R}+q^{\Delta}$ and (d) the heat flux, for the DSMC field at $\Kn_{\rm Gu}=0.20$. Positive values carry heat towards hotter gas. Lines follow each vector field. Averaging and smoothing as in figure~\ref{fig:q_parts_profiles}.}
\label{fig:q_parts_maps020}
\end{figure}

\begin{figure}[tbp]
\centering
\includegraphics[width=\linewidth]{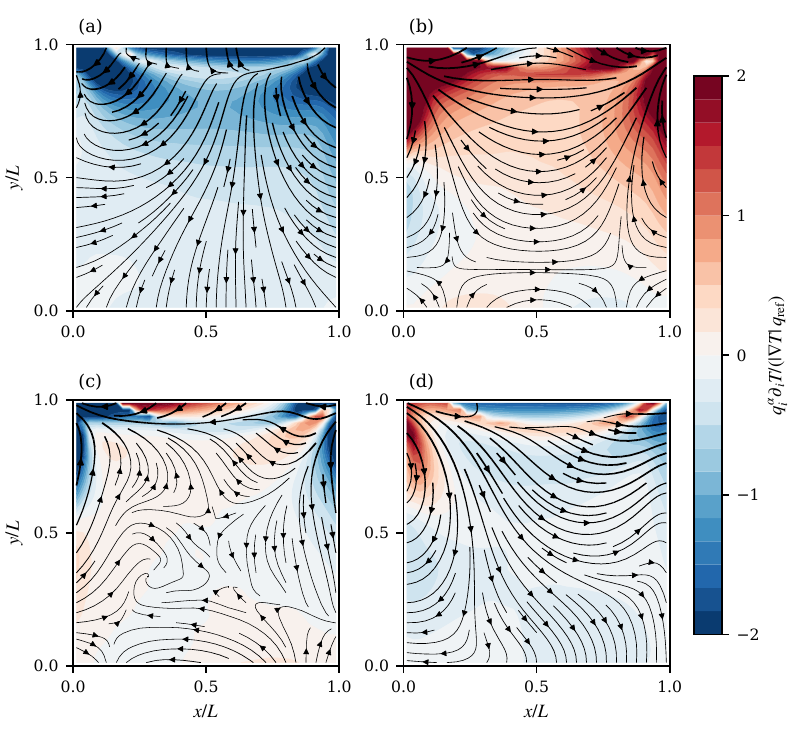}
\captionsetup{format=revisioncaption}
\caption{As figure~\ref{fig:q_parts_maps020}, at $\Kn_{\rm Gu}=0.05$.}
\label{fig:q_parts_maps005}
\end{figure}

\begin{figure}[tbp]
\centering
\includegraphics[width=\linewidth]{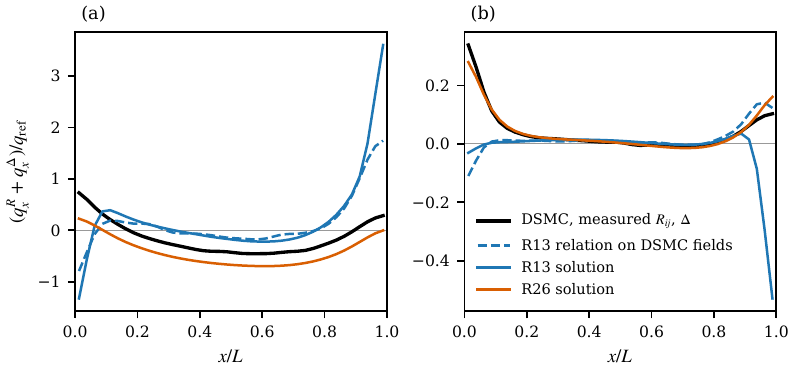}
\captionsetup{format=revisioncaption}
\caption{Fourth-order part $q^{R}_x+q^{\Delta}_x$ on the horizontal centreline at (a) $\Kn_{\rm Gu}=0.20$ and (b) $\Kn_{\rm Gu}=0.05$, in units of $q_{\rm ref}$ of the DSMC field: from the moments measured in DSMC, from the R13 relations evaluated on the DSMC fields, and from the R13 and R26 solutions.}
\label{fig:apriori_closure}
\end{figure}

\begin{shaded}
For figure~\ref{fig:canonical_A} we compare the R26 tensor on the $28\times28$ grid and a reconstruction from the $90\times90$ continuation calculation with the same pooled late-window DSMC moments. The $90\times90$ field used for this tensor reconstruction differs slightly from the field used to evaluate residuals and start the grid refinements, although both come from the same continuation stage (appendix~\ref{app:complete_continuation}). We interpolate both R26 tensors bilinearly to the DSMC cell centres, without extrapolation. After removing eight cells at each wall, 20,736 cells remain. Table~\ref{tab:canonical_table6} gives $E=\|A_m-A_d\|/\|A_d\|$, $G=\|A_m\|/\|A_d\|$ and the uncentred correlation $C=\langle A_m,A_d\rangle/(\|A_m\|\|A_d\|)$ on the same unsmoothed fields. They satisfy $E^2=1+G^2-2GC$ to floating-point accuracy before rounding. The three-dimensional norm includes the measured $xz$ and $yz$ components, which can be nonzero in a finite DSMC sample. Over $0.05\le x/L,\,y/L\le0.95$, further iteration of the $90\times90$ field used for refinement changes $\Aten$ by $0.23\%$. The $120\times120$ and $180\times180$ fields differ by $0.19\%$ over the same interior after 36 and eight nonlinear steps, respectively. For comparison, the R26--DSMC errors in table~\ref{tab:canonical_table6} are 0.62--0.98 for the in-plane components and 0.96 for the full tensor. After further iteration, six auxiliary equations near the upper right lid corner still have residuals above $3\times10^{-2}$. These differences include both the change of grid and the additional nonlinear iteration (appendix~\ref{app:complete_continuation}).

A separate Levenberg--Marquardt calculation on successively finer grids gives a $90\times90$ R26 field with a residual nine times smaller than that of the continuation estimate (appendix~\ref{app:complete_continuation}, table~\ref{tab:root_states}). The R26 discrepancies in table~\ref{tab:current_fields} change by at most 5.6\%, and the divergence-free fraction of the in-plane discrepancy changes from 0.9540 to 0.9482. On coarser grids the nonlinear calculations reach roots of the complete discrete equations. At $\Kn_{\rm Gu}=0.20$ these roots give fractions of 0.9409 on $20^2$ and 0.9452 on $24^2$. On the finest grids with such roots ($38^2$ and $24^2$ at the respective Knudsen numbers), the R26 discrepancies remain below the R13 values in table~\ref{tab:current_fields}, including those for R13 on $180^2$ at $\Kn_{\rm Gu}=0.05$. The small residuals on the coarser grids leave a separate spatial discretisation error.
\end{shaded}

\begin{center}
\begin{minipage}{\linewidth}
\centering
\rowcolors{1}{yellow!40}{yellow!40}
\begin{tabular}{llrrr}
\toprule
State & Field & $E$ & $G$ & uncentred $C$ \\
\midrule
$28\times28$ & $A_{xx}$ & 0.577590 & 0.725649 & 0.821993 \\
$28\times28$ & $A_{xy}$ & 0.942603 & 1.478213 & 0.776821 \\
$28\times28$ & $A_{yy}$ & 0.578752 & 0.876678 & 0.817638 \\
$28\times28$ & $A_{zz}$ & 5.482902 & 6.067652 & 0.638977 \\
$28\times28$ & $A$ (full 3D) & 0.904157 & 1.323269 & 0.730593 \\
$90\times90$ & $A_{xx}$ & 0.624322 & 0.753456 & 0.781677 \\
$90\times90$ & $A_{xy}$ & 0.981261 & 1.516641 & 0.770560 \\
$90\times90$ & $A_{yy}$ & 0.619749 & 0.937345 & 0.797213 \\
$90\times90$ & $A_{zz}$ & 6.282002 & 6.828535 & 0.597884 \\
$90\times90$ & $A$ (full 3D) & 0.964174 & 1.377225 & 0.714160 \\
\bottomrule
\end{tabular}
\captionsetup{format=revisioncaption}
\captionof{table}{Fourth-order moments relative to the common late-window DSMC estimate. The $28\times28$ field is the discrete root of the earlier R26 scheme. The $90\times90$ rows use the tensor reconstruction from the separate continuation calculation (appendix~\ref{app:complete_continuation}). The two calculations use different grids and discretisations. Table~\ref{tab:kn020_fourth_components} uses the equal-window DSMC estimate with its original finite-particle correction.}
\label{tab:canonical_table6}
\end{minipage}
\end{center}

\begin{shaded}
Taking finite differences on the native grid before removing the wall cells gives relative divergence errors of 0.614838 for the $28\times28$ solution and 0.597936 for the $90\times90$ tensor reconstruction. A weak comparison on 36 smooth test functions and two vector components, giving 72 coefficients, instead increases from 0.330705 to 0.338904. The two comparisons weight spatial variation differently. Neither error follows from the size of the tensor discrepancy or the energy in its divergence-free part. The sampling-energy correction above was derived for the equal-window estimate.
\end{shaded}

\begin{shaded}
For the same interior cells and common late-window DSMC mean, the relative heat-flux error increases from 0.235413 to 0.272687 between the two R26 fields, while the velocity error decreases from 0.079456 to 0.064504. Other values quoted at $\Kn_{\rm Gu}=0.20$ use different samples or domains. The $28\times28$ solution gives 0.407 against the equal-window mean in \S\ref{sec:interpretation} and 0.3021 against the sample-weighted mean restricted to the R13 domain in table~\ref{tab:current_fields}. For the latter comparison the $90\times90$ continuation field of \S\ref{sec:model_diagnostic} gives 0.3644. Thus changing the averaging and domain alone takes 0.407 to 0.3021. At fixed averaging and domain, changing the R26 field takes 0.3021 to 0.3644. The calculation closer to DSMC in velocity gives a larger heat-flux error.
\end{shaded}

\begin{shaded}
Figure~\ref{fig:canonical_sign} shows the sign of $q\cdot\nabla T$, with no activity threshold. Removing a square of side $0.05L$ at each upper corner leaves 25,472 cells. Without smoothing, the seven-realisation late-window pool gave a positive fraction of 0.582 and Jaccard overlaps with the R26 positive regions of 0.581 ($28\times28$ solution) and 0.553 ($90\times90$ tensor-reconstruction field), and the range of $q\cdot\nabla T$ over the seven realisations included zero in 95\% of cells: a sample of that size does not determine the sign in a single cell. The DSMC panel therefore uses the larger sample of figure~\ref{fig:kn020_antifourier} (26 independent runs of unequal length) with the same Gaussian smoothing; the R26 fields are not filtered. The positive fraction is 0.693 (0.85 over the central square of table~\ref{tab:q_parts}) and the overlaps are 0.765 and 0.714. Without the filter the larger sample gives 0.650, 0.701 and 0.658, so most of the change comes from the sample. At least 90\% of the runs agree on the sign in 88\% of the cells; the remaining cells lie along the boundary of the counter-gradient region.
\end{shaded}

\begin{figure}[tbp]
\centering
\includegraphics[width=\linewidth,height=0.80\textheight,keepaspectratio]{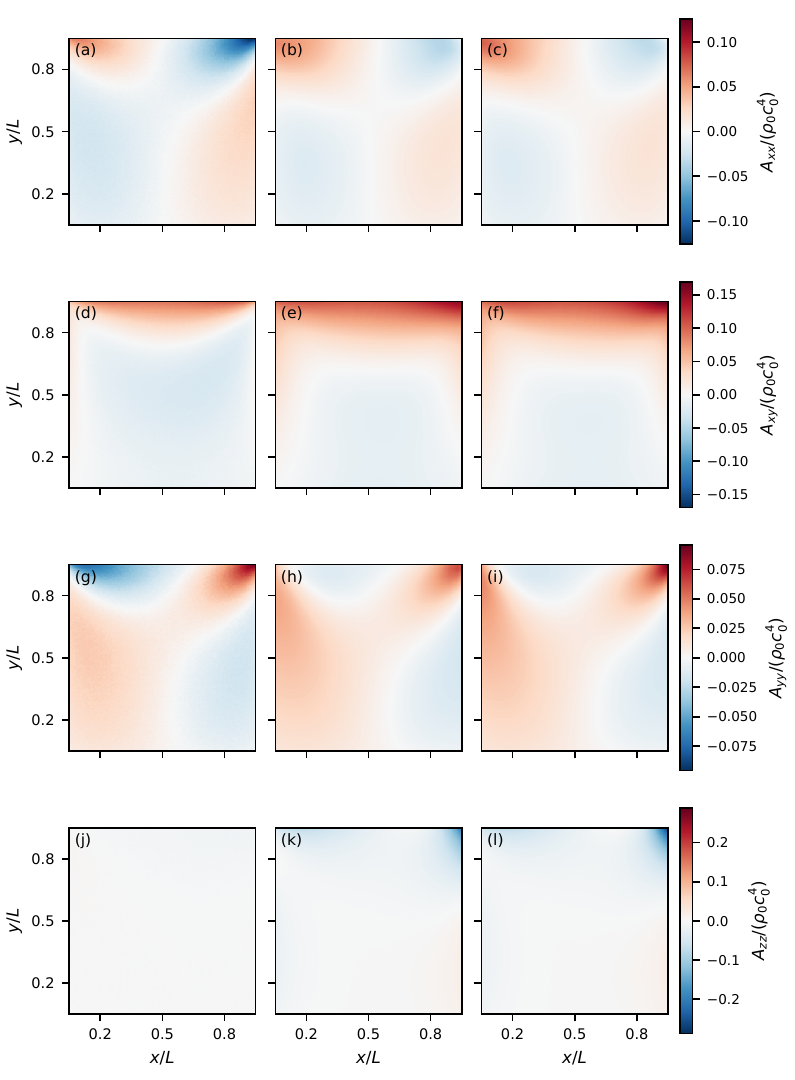}
\captionsetup{format=revisioncaption}
\caption{Composite fourth moments on the common late-window DSMC grid, excluding eight cells at each wall. From top to bottom the rows show $A_{xx}$, $A_{xy}$, $A_{yy}$ and $A_{zz}$, with a common symmetric colour scale in each row. Columns show the pooled moments of seven DSMC realisations (a,d,g,j), the $28\times28$ R26 solution (b,e,h,k), and the $90\times90$ R26 continuation field used for the tensor reconstruction (c,f,i,l). The R26 fields use different discretisations and are interpolated bilinearly to DSMC cell centres.}
\label{fig:canonical_A}
\end{figure}

\begin{figure}[tbp]
\centering
\includegraphics[width=360bp]{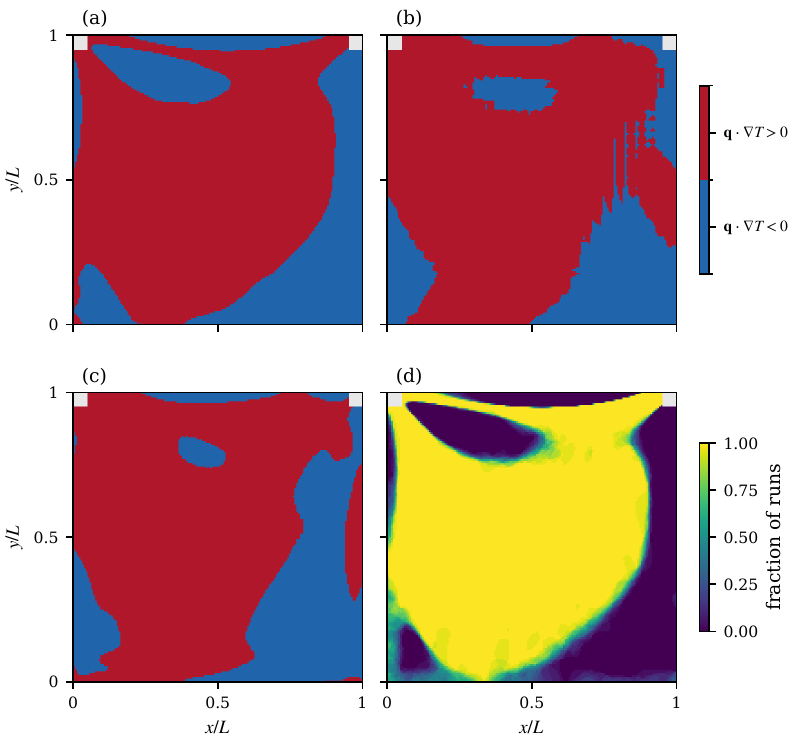}
\captionsetup{format=revisioncaption}
\caption{Sign of $q\cdot\nabla T$ for (a) DSMC, after Gaussian smoothing of width $0.03L$, (b) the $28\times28$ R26 solution and (c) the $90\times90$ R26 continuation field used for the tensor comparison. Panel (d) shows the fraction of the independent DSMC runs with positive sign. Grey squares mark the two excluded upper corners.}
\label{fig:canonical_sign}
\end{figure}

\section{\texorpdfstring{\colorbox{white}{\parbox{0.75\linewidth}{Conclusions}}}{Conclusions}}
\label{sec:conclusion}

\begin{shaded}
Over the central square of the cavity at $\Kn_{\rm Gu}=0.20$, the stress and tensorial fourth-moment terms have opposite projections along the heat flux. Both R13 and R26 capture the stress contribution. R26 predicts the opposing fourth-moment contribution more closely. R13 underestimates the opposition from the fourth moment, but its excessive Fourier contribution acts in the same direction as the missing term. Their partial cancellation helps explain why predicting the direction of heat transfer is a weak test of the closure.
\end{shaded}

\begin{shaded}
The balance itself leaves much of the fourth-order field undetermined. An in-plane tensor with zero divergence makes no fourth-order contribution to that balance, and the transverse component drops out altogether. In the R26--DSMC comparison at $\Kn_{\rm Gu}=0.20$, the part invisible to the discrete divergence accounts for 0.954 of the squared in-plane tensor difference for the $90\times90$ continuation estimate. The corresponding fractions for the coarse-grid roots and the separately grid-sequenced $90\times90$ estimate are 0.941 to 0.948 (appendix~\ref{app:complete_continuation}). Even a small remaining component can have a substantial divergence if it varies over a short distance.
\end{shaded}

\begin{shaded}
The omitted information enters elsewhere in the moment hierarchy. The third- and fourth-order equations respond to both the transverse and the Airy perturbations (appendix~\ref{app:complete_response}), and their wall relations constrain the admissible fields. Thus the ambiguity in interpreting a heat-flux measurement is compatible with a unique solution of the complete cavity problem. The pointwise kinetic construction in appendix~\ref{app:same_state} tests a different condition, the existence of a distribution with the specified moments.

For the two R26 discretisations, closer agreement with DSMC in velocity and in the directly differentiated fourth-order divergence accompanies larger discrepancies in the heat flux and the tensor itself. The weak divergence projection also weights the spatial differences differently. The finest grids on which we obtained roots of the full R26 equations gave smaller discrepancies than the R13 grids in table~\ref{tab:current_fields}. Across the $120^2$ and $180^2$ calculations, a change of $0.19\%$ in the interior tensor is small compared with its discrepancy from DSMC. These computations share a starting field and differ in both grid spacing and number of nonlinear steps.
\end{shaded}

\begin{shaded}
The most direct test here evaluates the published R26 closures using the lower moments measured in the gas and compares their predictions with the higher moments in that same sample. The remaining discrepancies exceed the estimated sampling noise at both Knudsen numbers. They also depend on spatial differentiation and the higher-order collision coefficients of Maxwell molecules and the VSS gas. Separating the scatter among DSMC realisations from changes due to pooling is essential to interpreting that comparison. For a model of heat transport, these results favour measurements of the stress and fourth-moment contributions alongside the total heat flux.
\end{shaded}

\appendix
\renewcommand{\theHsection}{appendix.\Alph{section}}
\input{numerical_algorithms.tex}
\section{Checks on the fourth-order projection}
\label{app:projection_checks}

The reduced fourth-order data contain fifteen sampled fields for each of the seven DSMC realisations. The reconstruction uses an average over ten blocks. Equation~\eqref{eq:finite_N_fourth} gives the exact correction to the sample central fourth moment for a fixed particle count \(N\). The number of particles in a cell fluctuates. We therefore used the measured local mean \(256\rho\) for \(N\), and repeated the calculation with \(N=256\). The RMS differences between the two corrections are \(1.74\times10^{-3}\), \(1.09\times10^{-4}\), \(1.72\times10^{-3}\) and \(1.72\times10^{-3}\) for \(A_{xx},A_{xy},A_{yy},A_{zz}\), respectively. We included both corrections in the projection tests and estimated the ensemble standard error separately.

We multiply the off-diagonal component by \(\sqrt{2}\) when storing the discrete tensor as a vector. Its Euclidean norm then equals the Frobenius norm of the tensor. A least-squares solution of \eqref{eq:discrete_projection} gives the projection onto \(\operatorname{range}(D_h^{\mathsf T})\). The remainder is orthogonal to this projection and has zero discrete divergence to the accuracy of the solve. In the reference calculation the normalized orthogonality error is \(2.7\times10^{-14}\), and the normalized residual \(\|D_h\delta\boldsymbol a_H\|/\|D_h\delta\boldsymbol a\|\) is \(4.9\times10^{-8}\). We repeated the projection with centred and forward differences, omitting 4, 8, 12 or 16 cells at each boundary, using smoothing widths of 1, 3 or 7 cells and applying both particle-count corrections. Including \(A_{zz}\), the hidden fraction ranges from 95.0 to 97.3\%. We quote the lower value when reporting a single estimate.

Five isotropic Gaussian tensor fields give an in-plane hidden fraction of \(0.354\pm0.001\). The twenty-one pairwise differences between the seven DSMC fields give in-plane fractions of 0.352--0.377 and total fractions of 0.521--0.541. The in-plane values agree with the fraction of in-plane degrees of freedom in the discrete kernel, \(1-2(N_g-2)^2/(3N_g^2)=0.352\) at \(N_g=144\), where \(N_g\) is the number of retained grid points in each direction. 

\begin{shaded}
All 48 tests used collocated centred or forward differences. For structureless fields the in-plane kernel fraction tends to $1/3$ with refinement, well below the measured hidden fraction. Its continuum value for the cavity fields remains unknown. The tests did not include staggered or higher-order stencils.
\end{shaded}

\section{Comparison with published R13 and R26 results}
\label{app:yang_validation}

\begin{shaded}
We compared our moment calculations with published R13 and R26 results for plane Poiseuille flow, plane Couette flow and the lid-driven cavity. Wherever the moment equations themselves determine the published result, our calculations reproduce it: the R26 and R13 velocity profiles of Poiseuille flow at $\Kn=0.1$ and $0.5$ within 0.1--0.2\%, the cavity velocities at two Knudsen numbers within 0.3--6\%, and the cavity heat flux at $\Kn=0.1$ away from the walls within 16\%. The exception is the heat flux in the layer next to the walls. At $\Kn=0.01$ the R26 equations give about half the kinetic value at the walls, as does the exact solution for plane Couette flow, and at $\Kn=0.1$ much less at the lid ($-0.47$ against $-6.40$). The published R26 values equal the kinetic ones at the lid and side walls at $\Kn=0.01$, though not at the floor ($-0.041$ against $-0.152$), and are close to them at the lid at $\Kn=0.1$. Sections~\ref{app:pub_couette}--\ref{app:pub_hybrid} examine this difference. Two treatments of the wall reproduce parts of the published wall values: reversing the sign of one term (\S\ref{app:pub_couette}) and evaluating the wall relations in the cell next to the wall (\S\ref{app:pub_walltreat}). Neither is a solution of the R26 equations with the wall relations derived from diffuse reflection, which we use.
\end{shaded}

\subsection{\texorpdfstring{\colorbox{white}{\parbox{0.75\linewidth}{Plane Poiseuille flow}}}{Plane Poiseuille flow}}\label{app:pub_poiseuille}

\begin{shaded}
For steady flow between parallel plates driven by a small uniform force, the linearised R26 equations reduce to ordinary differential equations in the wall-normal coordinate $y$ for $u_x$, $\sigma_{xy}$, $q_x$, $m_{xyy}$, $m_{xxx}$ and $R_{xy}$. We solved them, with the closures of appendix~\ref{app:actual_algorithm} and the wall relations of \S\ref{sec:discussion}, by Chebyshev collocation. Solutions on 65 and 97 points agree to five significant figures. The R13 equations were treated in the same way.

Figure~\ref{fig:pub_poiseuille}(a,b) compares these solutions with the R26 and R13 velocity profiles of \citet[their figure~1]{YangTangYang2019}, whose points we read from the vector artwork. Over the published R26 points the root-mean-square difference is 0.07\% of the centreline velocity at $\Kn=0.1$ and 0.01\% at $\Kn=0.5$. For R13 the differences are 0.04\% and 0.20\%. At $\Kn=1$ the published profiles differ from the exact solutions by 3.5\% (R26) and 9.6\% (R13). Our R13 solution also follows the maximum velocity of pressure-driven flow given by \citet[their figure~12b]{GuEmerson2009} within 0.5\% for $\Kn\le1$ (figure~\ref{fig:pub_poiseuille}c). In the linear regime the pressure gradient acts as a uniform force.

The R26 curve of that figure lies 0.7--3.7\% above our R26 solution for $0.1<\Kn<1$. It is reproduced within 0.5\% when the sign of the term $-2\theta\partial_km_{ijk}$ in the source $S_{ij}$ of the $R_{ij}$ equation (appendix~\ref{app:actual_algorithm}) is reversed. The profiles of Yang et al.\ distinguish the two signs: with the sign reversed, the difference from their R26 points at $\Kn=0.5$ rises from 0.01\% to 2.6\%. Their R26 profiles thus follow the equations as printed by \citet{GuEmerson2009} and used here. They also distinguish the wall treatment: with the wall-cell treatment of \S\ref{app:pub_walltreat} the differences become 0.46\% at $\Kn=0.1$ and 3.9\% at $\Kn=0.5$. We infer that these points are solutions of the printed boundary-value problem, possibly analytical ones, whereas the cavity values next to the walls in the same paper are not (\S\ref{app:pub_cavity}). The paper does not say how either was computed.
\end{shaded}

\begin{figure}[tbp]
  \centering
  \includegraphics[width=\textwidth]{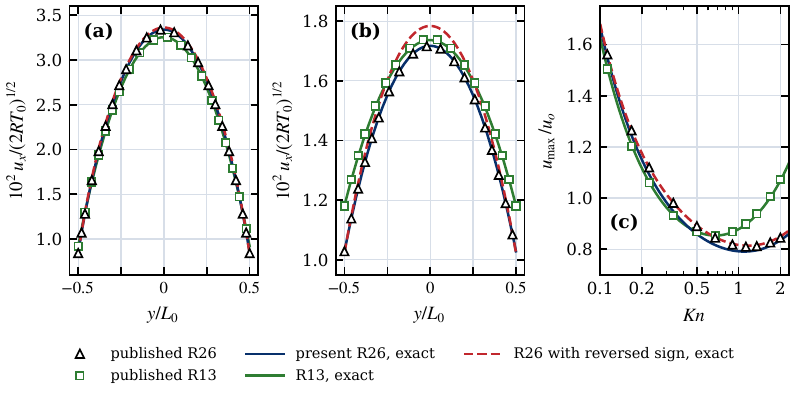}
  \captionsetup{format=revisioncaption}
  \caption{Plane Poiseuille flow. (a,b) Velocity profiles of \citet[their figure~1]{YangTangYang2019} at $\Kn=0.1$ and $0.5$: published R26 (triangles) and R13 (squares) against the exact solutions of the linearised R26 and R13 equations used here (solid lines). The dashed line is the R26 solution with the sign of $-2\theta\partial_km_{ijk}$ in $S_{ij}$ reversed. (c) Maximum velocity of pressure-driven flow against $\Kn$: published R26 and R13 curves of \citet[their figure~12b]{GuEmerson2009} (symbols at the points of their curves) and the same three solutions. Here $u_o$ is their velocity scale.}
  \label{fig:pub_poiseuille}
\end{figure}

\subsection{\texorpdfstring{\colorbox{white}{\parbox{0.75\linewidth}{Plane Couette flow}}}{Plane Couette flow}}\label{app:pub_couette}

\begin{shaded}
In plane Couette flow the heat flux along the walls, $q_x$, vanishes in the Navier--Stokes--Fourier description and is confined to the Knudsen layers; this heat flux without a temperature gradient along the wall is one of the non-classical effects reviewed by \citet{AristovZabelokFrolova2025}. In the linearised R26 equations it is carried by two exponential modes with decay lengths $1.96$ and $0.79\,\mu\sqrt\theta/p$. At an isolated wall the wall relations give $|q_x|=0.124\,|\sigma_{xy}|\sqrt\theta$.

\citet[their figure~3]{GuEmerson2009} give $q_x$ for argon between walls moving at $\pm50\,\mathrm{m\,s^{-1}}$ from R13, R26 and DSMC, at five Knudsen numbers. At this wall speed the nonlinear terms raise the R26 wall value by 10--50\%, so we solved the full nonlinear equations of appendix~\ref{app:actual_algorithm} for this flow by Chebyshev collocation. The residuals are below $10^{-10}$, and two grids give the same wall value to six significant figures. Table~\ref{tab:pub_couette} and figure~\ref{fig:pub_couette} show the result, normalised as in the published figure by the normal heat flux at the wall. The R26 equations give 60\% of the DSMC wall value at $\Kn=0.1$ and 35\% at $\Kn=1$.

The published R26 curves lie much closer to DSMC. We could reproduce them only by reversing the sign of the term $-2\theta\partial_km_{ijk}$ in $S_{ij}$. The solution with reversed sign agrees with the published curves within 0.4\% at the wall and 0.4\% in root-mean-square over the half gap at all five Knudsen numbers. It also gives the published gas velocity at the wall within 1\% and the maximum velocity of Poiseuille flow (\S\ref{app:pub_poiseuille}). We also made 88 single changes of a coefficient or sign in the linearised equations and wall relations. Reversing this sign leaves a root-mean-square difference of 5\% from the published curves, removing the term 15\%, and every other change more than 16\%; the linearised equations as printed give 26\%.

The printed sign is nevertheless the consistent one. The definitions of $R_{ij}$ and $\psi_{ijk}$ subtract $7\theta\sigma_{ij}$ and $9\theta m_{ijk}$ from the corresponding full moments, which leaves $(9-7)\theta\,\partial_km_{ijk}$ on the left-hand side of the $R_{ij}$ equation. DSMC supports the same sign. In the cavity of \S\ref{app:pub_cavity} we evaluated every term of the linearised $R_{xy}$ balance from the DSMC moments, with $\psi_{ijk}$ and $\Omega_i$ measured and not closed. Between 0.1 and 1 mean free path below the lid the residual of the balance is 2\% of its largest term with the printed sign and twice that term with the sign reversed. Later papers of the same group print the same sign \citep[equation~2.8]{GuBarberJohnEmerson2019}, and the Poiseuille profiles of Yang et al.\ conform to it (\S\ref{app:pub_poiseuille}). We therefore use the equations as printed.

The finite-volume treatment of the wall of \S\ref{app:pub_walltreat} offers a second explanation, and we tested it in the same way. In its continuum limit the treatment replaces the tangential wall relations by modified ones, given in \S\ref{app:pub_walltreat}. With these relations and the equations as printed, the nonlinear solution gives $|q_x|/|q_y|=1.219$, $1.349$, $1.200$, $1.026$ and $0.896$ at the five Knudsen numbers (table~\ref{tab:pub_couette}). This is 14\% above the published value at $\Kn=0.1$ and 6\% above at $0.25$, but 18--44\% below it at $\Kn\ge0.5$, because the ratio falls with $\Kn$ where the published one rises. For the maximum velocity of Poiseuille flow in figure~12(b) of \citet{GuEmerson2009} the same treatment differs from the published R26 curve by 0.5\% for $\Kn<0.3$ and 2.5\% for $0.3\le\Kn\le1$, against 0.04\% and 0.2\% with the sign reversed. The reversed sign therefore fits the published Couette and Poiseuille curves of \citet{GuEmerson2009} much better than the wall-cell treatment does. The coefficients of the modified relations come from the linearised analysis, and at this wall speed the nonlinear terms are not small, so the wall-cell values are approximate; the difference in the trend with $\Kn$ is the more robust result. Both explanations are inferences from the published curves. We do not know how those curves were computed.
\end{shaded}

\begin{center}
\begin{minipage}{\linewidth}
\centering
\begin{tabular}{lccccc}
\toprule
$\Kn$ & 0.1 & 0.25 & 0.5 & 0.75 & 1.0 \\
\midrule
DSMC \citep{GuEmerson2009} & 1.249 & 1.443 & 1.603 & 1.663 & 1.698 \\
published R26 & 1.070 & 1.277 & 1.461 & 1.551 & 1.601 \\
R26 equations as used here & 0.755 & 0.840 & 0.753 & 0.657 & 0.590 \\
sign of $-2\theta\partial_km_{ijk}$ reversed & 1.070 & 1.282 & 1.462 & 1.550 & 1.599 \\
wall-cell treatment (continuum limit) & 1.219 & 1.349 & 1.200 & 1.026 & 0.896 \\
\bottomrule
\end{tabular}
\captionsetup{format=revisioncaption}
\captionof{table}{Plane Couette flow of \citet[their figure~3]{GuEmerson2009}, wall speed $\pm50\,\mathrm{m\,s^{-1}}$: tangential heat flux at the wall divided by the normal heat flux there, $|q_x|/|q_y|$. The first two rows are read from the published figure. The last three are our solutions of the nonlinear R26 equations, the last with the modified wall relations of \S\ref{app:pub_walltreat}.}
\label{tab:pub_couette}
\end{minipage}
\end{center}

\begin{figure}[tbp]
  \centering
  \includegraphics[width=\textwidth]{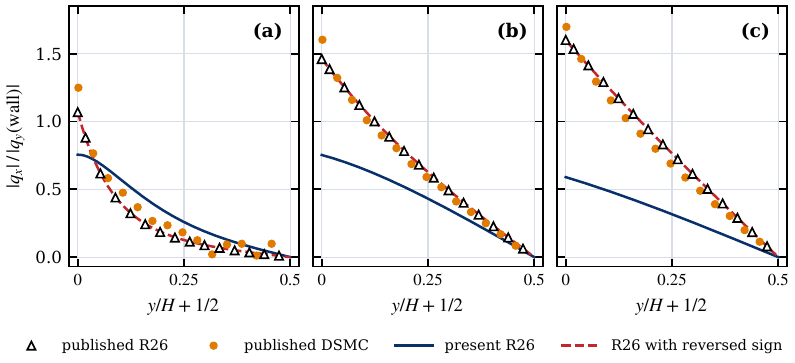}
  \captionsetup{format=revisioncaption}
  \caption{Plane Couette flow of \citet[their figure~3]{GuEmerson2009}: tangential heat flux across half of the gap of width $H$, divided by the normal heat flux at the wall, against distance from the wall, at (a) $\Kn=0.1$, (b) $0.5$ and (c) $1.0$. Triangles and circles are the published R26 curve and DSMC points. The solid line is our solution of the nonlinear R26 equations, and the dashed line the solution with the sign of $-2\theta\partial_km_{ijk}$ in $S_{ij}$ reversed.}
  \label{fig:pub_couette}
\end{figure}

\subsection{\texorpdfstring{\colorbox{white}{\parbox{0.75\linewidth}{Lid-driven cavity}}}{Lid-driven cavity}}\label{app:pub_cavity}

We compared the R26 velocity profiles with those of \citet{YangTangYang2019} for a fully diffuse lid-driven cavity with \(L_0=10^{-5}\,\mathrm{m}\), \(T_0=273\,\mathrm{K}\), \(U_0=10\,\mathrm{m\,s^{-1}}\) and their \(\Kn=0.1\). We read the coordinates of the triangular R26 markers directly from their vector figure~7, without fitting or smoothing. All three grids (\(17^2\), \(21^2\) and \(25^2\)) met our positivity, global-balance and residual tolerances. On the \(25^2\) grid the relative \(L_2\) differences from the published data are 0.0391 for \(u_x\) on the vertical centreline and 0.0618 for \(u_y\) on the horizontal centreline. The respective correlations are 0.99962 and 0.99986, and the changes from \(21^2\) to \(25^2\) are 0.33\% and 0.62\%. We then solved the same problem on \(34^2\), \(48^2\) and \(64^2\) grids by grid sequencing. The \(34^2\) and \(48^2\) solutions satisfy the discrete equations to a largest scaled residual below \(10^{-9}\); the \(64^2\) solution (profiles within 2\% of the \(48^2\) solution) has a largest scaled residual of \(1.5\times10^{-5}\). A later \(56^2\) solution is a root (largest scaled residual \(4\times10^{-12}\)) and differs from the \(64^2\) profiles by at most 0.5\% of their range. From \(48^2\) to \(64^2\) the velocity profiles change by at most 0.09\% and 0.12\% of their range, within the same criterion, and the relative \(L_2\) differences from the published data are 0.0368 and 0.0562 on \(64^2\). Both profiles satisfy the criteria set before the comparison, namely a relative \(L_2\) error below 0.15 and a change below 0.10 on the last grid refinement.

\begin{figure}[tbp]
  \centering
  \includegraphics[width=\textwidth]{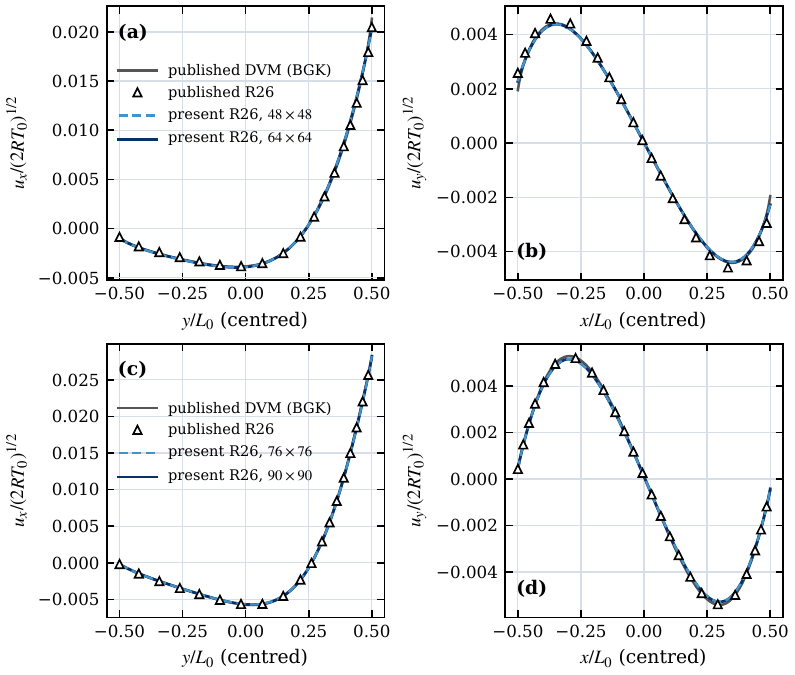}
  \captionsetup{format=revisioncaption}
\caption{Centreline velocities in the diffuse cavity of \citet{YangTangYang2019}, in units of \(\sqrt{2RT_0}\), where \(R\) is the specific gas constant and \(T_0\) their reference temperature, at (a,b) \(\Kn=0.1\) and (c,d) \(\Kn=0.01\). Triangles are the published R26 data read from the vector figures and the grey line is their DVM solution. Blue lines are the present R26 calculations on the two finest grids, which overlap at this scale.}
  \label{fig:yang_velocity_validation}
\end{figure}

\begin{shaded}
We repeated the comparison at their $\Kn=0.01$ (their figure~5) on nine grids from $25^2$ to $90^2$, each a root of the discrete equations. On $90^2$ the relative $L_2$ differences from the published R26 velocities are 0.0027 ($u_x$) and 0.020 ($u_y$), and the profiles change by less than 0.1\% of their range between the last two grids (heat flux: 1.2\%). Figure~\ref{fig:yang_velocity_validation} shows the velocities at both Knudsen numbers.

At $\Kn=0.1$ the heat flux agrees with the published profiles in the interior of the cavity and differs next to the walls (figure~\ref{fig:yang_heat_flux}). On the $64^2$ grid the relative $L_2$ differences from the published R26 profiles of $q_x$ and $q_y$ are 0.70 and 0.81 over all published points, above the tolerance of 0.30, and 0.14 and 0.15 over the points more than $0.1L_0$ from a wall. Of the squared difference, 98\% for $q_x$ and 97\% for $q_y$ lie within $0.1L_0$ of a wall. The largest change between successive grids falls from 3.4\% of the profile range ($25^2$ to $34^2$) to 1.2\% ($48^2$ to $64^2$) for $q_x$, and from 1.8\% to 0.7\% for $q_y$, so grid resolution does not account for the difference. We also ran DSMC for this case, with Maxwell molecules on $160^2$ cells as in \S\ref{sec:dsmc}. Our R26 heat flux differs from it by 0.27 and 0.28 over the whole centrelines and by 0.11 and 0.12 more than $0.1L_0$ from the walls.
\end{shaded}

\begin{figure}[tbp]
  \centering
  \includegraphics[width=\textwidth]{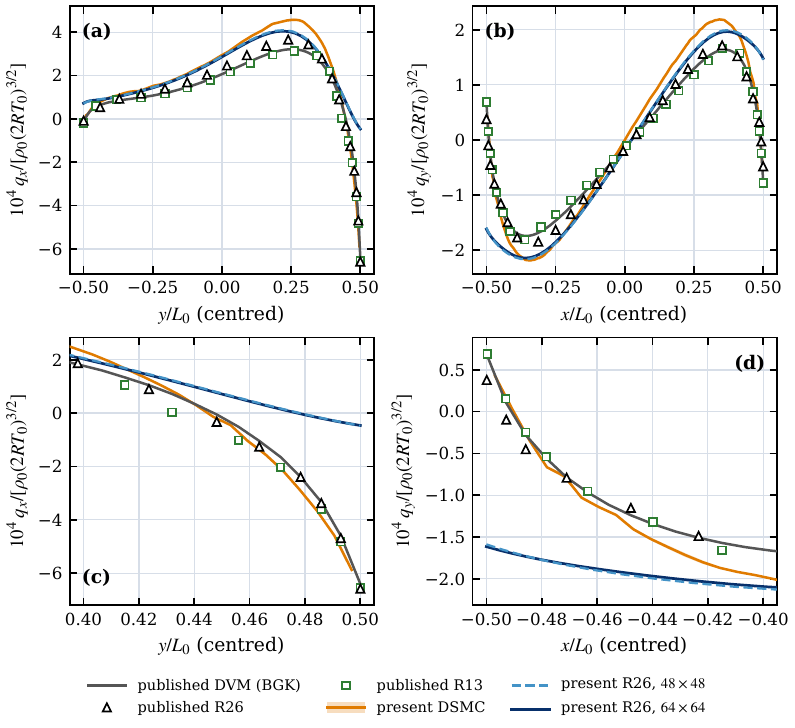}
  \captionsetup{format=revisioncaption}
  \caption{Heat flux in the cavity of figure~\ref{fig:yang_velocity_validation} at $\Kn=0.1$, in units of $\rho_0(2RT_0)^{3/2}$ as in \citet{YangTangYang2019}: (a) $q_x$ on the vertical centreline and (b) $q_y$ on the horizontal centreline; (c,d) the same within $0.1L_0$ of the lid and of the left wall. Grey line, triangles and squares: published DVM, R26 and R13. Blue lines: present R26 on the \(48^2\) grid and the \(64^2\) solution (profiles within 2\% of the \(48^2\) solution; largest scaled residual \(1.5\times10^{-5}\)). Orange line with band: present DSMC and the scatter between its two half-samples; more than \(0.05L_0\) from a wall the DSMC profile is averaged over seven cells.}
  \label{fig:yang_heat_flux}
\end{figure}

\begin{figure}[tbp]
  \centering
  \includegraphics[width=\textwidth]{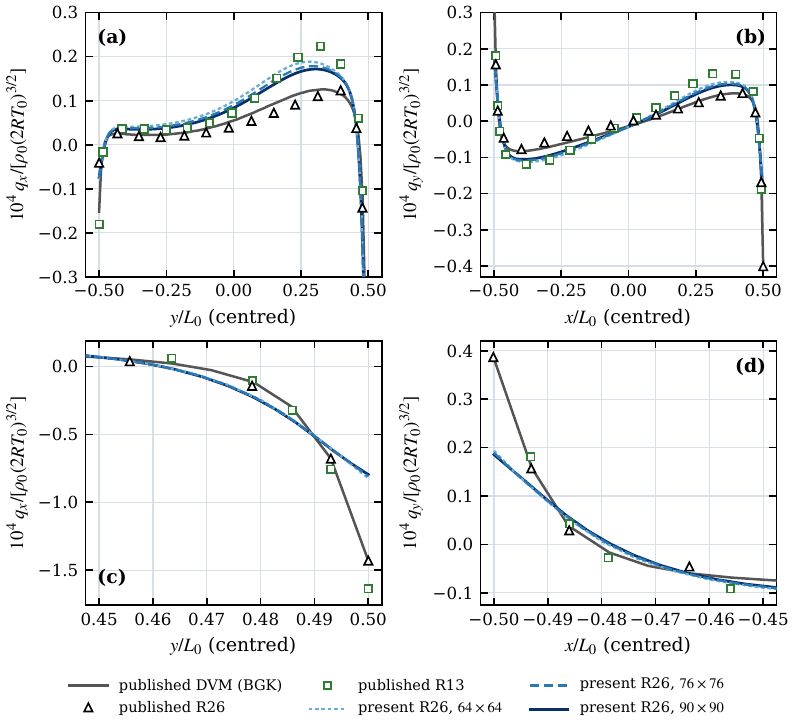}
  \captionsetup{format=revisioncaption}
  \caption{As figure~\ref{fig:yang_heat_flux} at $\Kn=0.01$, with the present R26 on the three finest grids. In (a) and (b) the ordinate is limited to show the interior of the cavity. Panels (c) and (d) show the wall layers within $0.05L_0$ of the lid and of the left wall, where the published R26 points lie on the DVM line.}
  \label{fig:yang_heat_flux_kn001}
\end{figure}

\begin{shaded}
At $\Kn=0.01$ the Knudsen layers are about $0.01L_0$ thick, and the heat flux converges slowly on our grids (figure~\ref{fig:yang_heat_flux_kn001}). Over the points more than $0.05L_0$ from a wall the relative $L_2$ difference from the published DVM profile falls from 0.57 to 0.47 and 0.40 on the $64^2$, $76^2$ and $90^2$ grids for $q_x$, and from 0.43 to 0.36 and 0.32 for $q_y$. The heat flux in this region is an order of magnitude smaller than at the lid. Next to the walls the steep variation is reproduced in form, but not in amplitude.

Table~\ref{tab:pub_cavity} lists the values at the walls, with the discrete-velocity (DVM) solution of the BGK equation that \citet{YangTangYang2019} give as their kinetic reference. At $\Kn=0.01$ the heat flux at a wall is essentially that of its Knudsen layer. Our R26 values there are 0.56, 0.48, 0.49 and 0.49 of the DVM values at the lid, the left and right walls and the floor, which is the deficit found for Couette flow in \S\ref{app:pub_couette}. At the middle of the lid our solution on the $90^2$ grid gives $|q_x|=0.122\,|\sigma_{xy}|\sqrt\theta$, against $0.124$ for the exact Couette solution; with the same shear stress the DVM value corresponds to $0.219$. The published R26 values, in contrast, equal the DVM values to within 0.1\% at the lid and at both side walls, and likewise at the grid point adjacent to each of these walls. The published R13 values exceed the DVM by 11--15\% at the same points, as expected of a moment solution: the exact R13 wall value for Couette flow, $0.270\,|\sigma_{xy}|\sqrt\theta$, also exceeds the kinetic one.

At $\Kn=0.1$ the Knudsen layers reach further into the cavity, and at each wall their heat flux opposes that of the interior. The R26 layer has about half the kinetic amplitude. It therefore leaves a small value at the lid, $-0.47$ against $-6.40$, and does not reverse the sign of the interior flux at the other three walls. Our DSMC supports the kinetic values: in the cells next to the lid and the side walls it gives $-5.88$, $0.43$ and $-0.43$, where the DVM line gives $-5.60$, $0.44$ and $-0.49$. The published R26 values at the four points nearest the lid are within 3\% of the DVM. At the side walls they have the sign of the DVM and about 60\% of its magnitude.

We used $\Kn_{\rm Gu}=0.1$. \citet{YangTangYang2019} cite the mean free path of \citet{Struchtrup2005} without giving its formula. The same group later gives the Gu convention, $\lambda=(\mu_0/p_0)(\pi RT_0/2)^{1/2}$, with $R$ the specific gas constant and $T_0=273$\,K \citep[equation~8.8]{YangMenLi2026}. This supports our choice. Other definitions differ by factors such as the $\sqrt{\pi/2}$ relating $\Kn_{\rm Rana}$ and $\Kn_{\rm Gu}$ (\S\ref{sec:kn020_matched}). Both R26 calculations draw on the wall conditions of \citet{GuEmerson2009}. Yang and colleagues specify diffuse reflection for their discrete-velocity reference, and we use full accommodation.

The equations printed by \citet{YangTangYang2019} differ from those of \citet{GuEmerson2009} in two respects. Their equation~12 relaxes $m_{ijk}$, $R_{ij}$ and $\Delta$ at the single rate $p/\mu$ of the BGK operator, where we use the Maxwell-molecule rates $3p/(2\mu)$, $7p/(6\mu)$ and $2p/(3\mu)$, and their closure for $\Omega_i$ (their equation~A1) omits the term $-4\mu\partial_j(R_{ij}/\rho)$. We solved the $\Kn=0.1$ cavity on the $25^2$ grid with each difference separately, with both, and with all relaxation rates and closure constants at their BGK values. The value of $10^4q_x$ at the lid changes from $-0.55$ to between $+0.36$ and $-1.36$, against $-6.58$ published, and $q_y$ at the side walls keeps the sign opposite to the published one. Their own Poiseuille profiles, moreover, follow the equations of \citet{GuEmerson2009}: with the rates of their equation~12 the difference from their R26 points at $\Kn=0.5$ would be 2.5\%, and with their closure for $\Omega_i$ 1.0\%, instead of 0.01\% (\S\ref{app:pub_poiseuille}). The printed differences therefore do not account for the wall layer of the published profiles.
\end{shaded}

\begin{center}
\begin{minipage}{\linewidth}
\centering
\begin{tabular}{llrrrrr}
\toprule
 & & \multicolumn{3}{c}{published} & \multicolumn{2}{c}{present} \\
$\Kn$ & wall & DVM & R26 & R13 & R26 & DSMC \\
\midrule
0.01 & lid & $-1.428$ & $-1.429$ & $-1.637$ & $-0.796$ & -- \\
 & left wall & $0.387$ & $0.387$ & -- & $0.186$ & -- \\
 & right wall & $-0.401$ & $-0.400$ & -- & $-0.196$ & -- \\
 & floor & $-0.152$ & $-0.041$ & $-0.180$ & $-0.074$ & -- \\
\addlinespace
0.1 & lid & $-6.40$ & $-6.58$ & $-6.56$ & $-0.47$ & $-5.88$ \\
 & left wall & $0.670$ & $0.380$ & $0.690$ & $-1.62$ & $0.43$ \\
 & right wall & $-0.762$ & $-0.478$ & $-0.779$ & $1.48$ & $-0.43$ \\
 & floor & $-0.210$ & $-0.086$ & $-0.186$ & $0.72$ & $-0.07$ \\
\bottomrule
\end{tabular}
\captionsetup{format=revisioncaption}
\captionof{table}{Heat flux at the walls of the cavity of \citet{YangTangYang2019}, in units of $10^{-4}\rho_0(2RT_0)^{3/2}$: $q_x$ at the middle of the lid and of the floor, and $q_y$ at mid-height of the side walls. The published values are the wall points of their figures~5 and~7; a dash means that the figure has no point at the wall. The present R26 values are those of the $90^2$ grid ($\Kn=0.01$) and of the $64^2$ grid ($\Kn=0.1$). The DSMC values are those of the cell next to the wall, centred $0.003L_0$ from it.}
\label{tab:pub_cavity}
\end{minipage}
\end{center}

\begin{figure}[tbp]
  \centering
  \includegraphics[width=\textwidth]{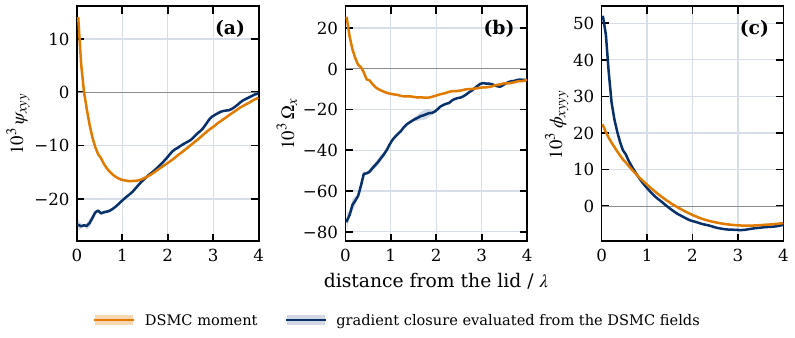}
  \captionsetup{format=revisioncaption}
  \caption{Moments $\psi_{xyy}$, $\Omega_x$ and $\Phi_{xyyy}$ below the lid of the $\Kn=0.1$ cavity, on the vertical centreline, against distance from the lid in mean free paths $\lambda$. Orange line with band: the moments measured in DSMC and the scatter between its two half-samples. Blue line: the R26 gradient closures of appendix~\ref{app:actual_algorithm} evaluated from the DSMC fields. Units: $\rho_0=1$, $RT_w=1$.}
  \label{fig:closure_wall_distance}
\end{figure}

\subsection{\texorpdfstring{\colorbox{white}{\parbox{0.75\linewidth}{Treatment of the cell next to the wall}}}{Treatment of the cell next to the wall}}\label{app:pub_walltreat}

\begin{shaded}
In the finite-volume implementation of the R26 equations by the same group, the derivatives that enter the wall relations are approximated by their values in the cell adjacent to the wall \citep[\S4.2]{GuGreenbankMoulinecEmerson2024}. We implemented this statement in our solver as follows. Every moment is given a value at the wall face. The moments that are odd in the wall-normal velocity, $\sigma_{nt}$, $q_n$, $R_{nt}$, $m_{nnn}$ and $m_{ntt}$, take the value of the adjacent cell, so that their normal derivative vanishes over the half cell next to the wall, and the closures $\phi_{ijkl}$, $\psi_{ijk}$ and $\Omega_i$ in the wall relations are evaluated in that cell. The other moments are given by the wall relations. The wall relations of \S\ref{sec:discussion}, in contrast, prescribe only the even moments and leave the odd ones to the field equations. Other readings of the published statement are possible, and we do not claim that this is the implementation of \citet{GuGreenbankMoulinecEmerson2024}.

Figure~\ref{fig:wallcell_near_wall}(a) shows linear Couette flow at $\Kn=0.1$ on 25 to 800 cells. With the wall relations as printed, $|q_x|/(|\sigma_{xy}|\sqrt\theta)$ at the wall converges at second order to the exact $0.1235$. With the wall-cell treatment it rises from 0.189 on 25 cells to 0.2016 on 400 and 0.2018 on 800. The limit is therefore not the printed problem but a different boundary-value problem, in which the vanishing normal derivative of the odd moments acts as an additional set of wall conditions. The difference between the wall value and that of the first cell falls from 0.026 to 0.005, but figure~\ref{fig:wallcell_near_wall}(b) shows that the profile does not become smooth. A sublayer whose thickness decreases with the cell size remains next to the wall, and across it the heat flux falls from the wall value to about 0.103, the value with which the interior solution meets the wall. In the continuum limit the sublayer is a jump at the wall. Eliminating it gives modified wall relations for the printed equations: the four wall relations for $u_t$, $q_t$, $m_{nnt}$ and $m_{ttt}$ are applied, without their $\psi$ and $\Omega$ terms, to $q_w=q_t+\tfrac{5}{14}\psi_{nnt}+\tfrac{1}{14}\Omega_t$ and $u_w=u_t-\tfrac25(q_w-q_t)$, where $q_t$, $u_t$, $\psi_{nnt}$ and $\Omega_t$ are the values of the interior solution at the wall. Solutions with these relations agree with the finite-volume values on 1600 cells within 0.3\% for $0.05\le\Kn\le1$; they give the values of \S\ref{app:pub_couette}.

In the cavity the treatment brings the lid value close to the published one and gives the side walls the sign of the kinetic solution (figure~\ref{fig:walltreat_cavity}, table~\ref{tab:walltreat}). At $\Kn=0.1$ the lid value of $10^4q_x$ is $-5.30$ on the $25^2$ grid and $-5.43$ on $34^2$, against $-0.47$ with the printed relations, $-6.58$ published and $-5.88$ from DSMC. The side-wall values of $q_y$ take the sign of the kinetic solution, with 2.2 to 2.6 times the published magnitude. At $\Kn=0.01$ the lid value is $-1.24$ on $48^2$, against $-1.43$ published, and the side-wall values are about 80\% of the published ones. The profile next to the wall is not that of the published curves. Figure~\ref{fig:wallcell_near_wall}(c) shows $q_x$ below the lid at $\Kn=0.1$ with the grid nodes marked. With the printed relations the solution is continuous at the wall: the difference between the wall value and that of the first node is $0.31$, $0.20$, $0.12$ and $0.08$ on the $25^2$ to $64^2$ grids. With the wall-cell treatment it is 2.63 and 2.46 on $25^2$ and $34^2$, so that the first node, $0.09\lambda$ from the lid, has $-2.97$ where the published R26 and DVM profiles have about $-4.5$. \citet{YangTangYang2019} state that they used a $101\times101$ non-uniform grid, and their profiles are shown by a few points per mean free path. On a grid of that kind a sublayer of a few hundredths of a mean free path would be spanned by several nodes, and the profile drawn through them would appear continuous. This is an observation about resolution, not about how the published figures were drawn. Away from the walls the two treatments give similar heat flux, and the wall-cell treatment does not improve the agreement with DSMC over the centrelines. The relative $L_2$ differences on $34^2$ are 0.26 for $q_x$ and 0.44 for $q_y$, 0.27 and 0.28 with the printed relations on the same grid, and on $56^2$ 0.30 and 0.46 against 0.28 and 0.28. The increase for $q_y$ comes from an overshoot within $0.1L_0$ of the side walls (figure~\ref{fig:walltreat_cavity}b). On $56^2$ the lid value is $-5.51$, within 0.4\% of the $48^2$ value, and the side-wall values are $1.00$ and $-1.13$. The differences from DSMC grow with refinement, from 0.22 to 0.26, 0.29 and 0.30 for $q_x$ on $25^2$, $34^2$, $48^2$ and $56^2$.

Two sources of difference from the published curves can be separated. The first is the wall treatment. It changes the lid value by about a factor of ten, from $-0.47$ to $-5.5$, and decides the sign at the side walls. The second is the discretisation, which matters once the wall treatment is fixed, because that treatment puts a jump between the wall value and the interior. Our grid has its first node one spacing from the wall, whereas a cell-centred finite-volume grid has its first cell centre half a cell away. On $56^2$ the wall value is $-5.51$ and the first node $-3.40$, so a value half a spacing from the wall would be about 20\% smaller than the wall value. Near-wall values and plotted wall values therefore depend on the arrangement of the unknowns at the level of tens of per cent. Segregated, pressure-based finite-volume solvers of the SIMPLE type also commonly add stabilising diffusion and under-relax the boundary values, which can change the near-wall layer further. We have not modelled these features.

A finite-volume calculation and ours may also differ in what is called a converged solution. A segregated solver with under-relaxation usually stops when the fields no longer change between iterations, whereas our Newton iteration seeks a root of the discrete equations; we report the residual of every field and, for these comparisons, the change between the last two grids. With the wall-cell treatment our solver reaches a root readily: the largest scaled residual is $2.6\times10^{-10}$, $2.7\times10^{-12}$, $2.2\times10^{-10}$ and $4.4\times10^{-11}$ on $25^2$, $34^2$, $48^2$ and $56^2$. With the printed relations it reaches roots up to $56^2$ ($4.1\times10^{-12}$), with more iterations on the finer grids, and the $64^2$ solution stops at $1.5\times10^{-5}$. This suggests that the additional wall conditions ease the solution next to the wall, which would also suit a solver for general geometries.

The two approaches thus solve different problems. The wall relations of \citet{GuEmerson2009}, used here, are derived from the diffuse-reflection model with the 26-moment distribution function, and the exact solution of the R26 equations with them keeps about half of the kinetic heat flux at the wall. The wall-cell treatment adds wall conditions that are not derived from the wall model. Their effect resembles that of the Knudsen layer: they raise the wall value towards the kinetic one, in Couette flow at $\Kn=0.1$ from 0.124 to 0.202. We regard this as a property of the treatment rather than of the moment equations, because it is concentrated in a sublayer that shrinks with the cell, and it does not improve the interior, where, more than one or two mean free paths from the wall, the closures of $\psi_{ijk}$ and $\Omega_i$ are close to the measured moments (\S\ref{app:pub_hybrid}). The finite-volume implementation of \citet{GuGreenbankMoulinecEmerson2024} was designed for general geometries, where this choice of wall gradients is a natural one, and the hybrid method of the same group addresses the wall layer directly. For the present paper the comparison has one consequence: agreement of an R26 solution with the kinetic heat flux at a wall depends on the wall treatment, so the heat-flux evidence of the main text rests on the DSMC comparisons and not on agreement with the published R26 values at the walls.
\end{shaded}

\begin{figure}[tbp]
  \centering
  \includegraphics[width=\textwidth]{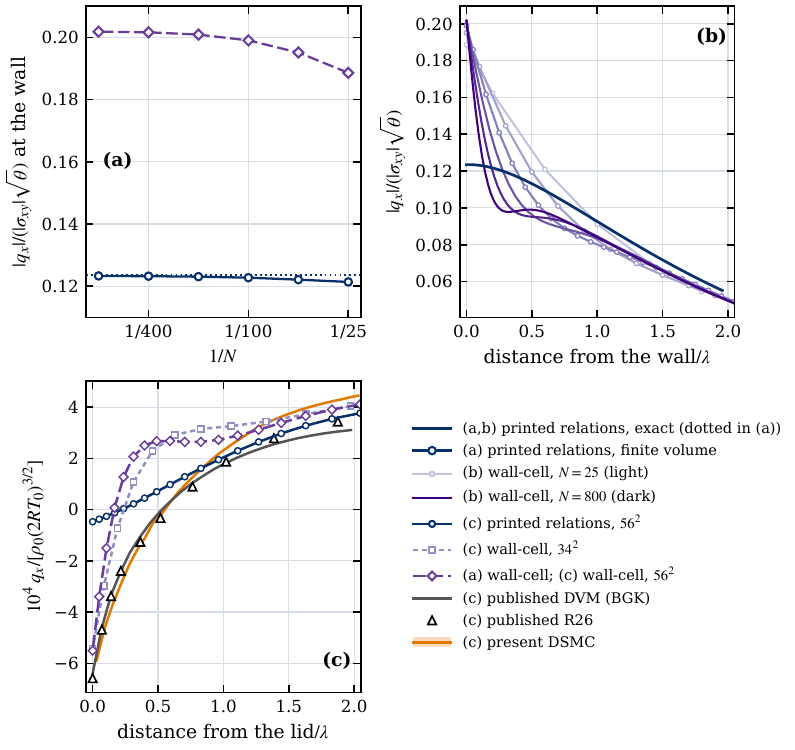}
  \captionsetup{format=revisioncaption}
  \caption{Wall-cell treatment at $\Kn=0.1$. (a) Linear plane Couette flow: $|q_x|/(|\sigma_{xy}|\sqrt\theta)$ at the wall against the inverse number of cells $N$ across the gap, with the wall relations of \S\ref{sec:discussion} (circles; dotted line: exact solution) and with the wall-cell treatment (diamonds). (b) The same quantity against distance from the wall in mean free paths $\lambda$: exact solution with the wall relations of \S\ref{sec:discussion}, and the wall-cell treatment on $N=25$, 50, 100, 200, 400 and 800 cells (light to dark; the wall-face value is joined to the cell centres; symbols on the three coarsest grids). (c) $q_x$ below the lid of the cavity of \citet{YangTangYang2019} on the vertical centreline, in units of $10^{-4}\rho_0(2RT_0)^{3/2}$, with the grid nodes marked: present R26 with the wall relations of \S\ref{sec:discussion} ($56^2$) and with the wall-cell treatment ($34^2$, $56^2$), published DVM and R26, and present DSMC with the scatter between its half-samples, against distance from the lid in mean free paths.}
  \label{fig:wallcell_near_wall}
\end{figure}

\begin{figure}[tbp]
  \centering
  \includegraphics[width=\textwidth]{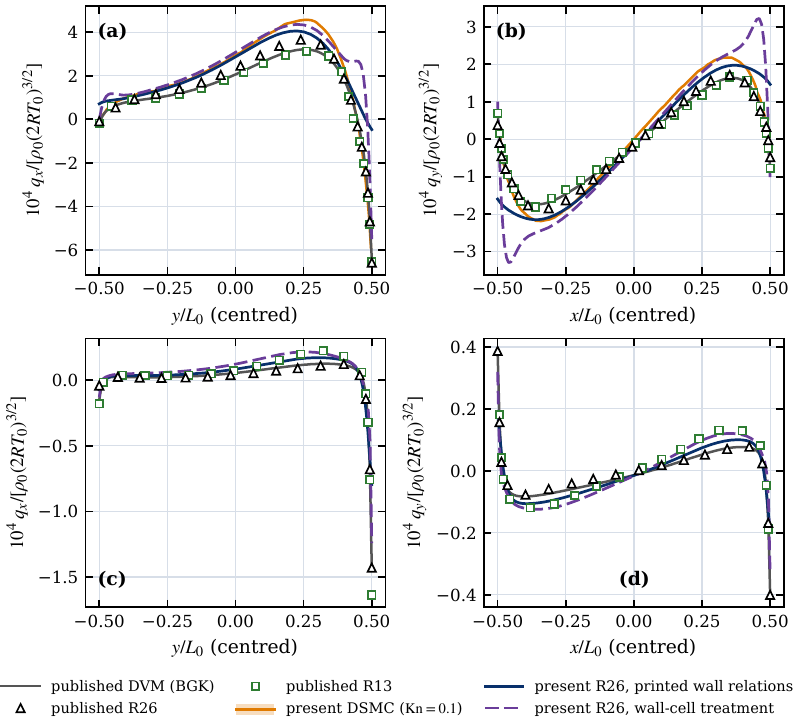}
  \captionsetup{format=revisioncaption}
  \caption{Heat flux in the cavity of \citet{YangTangYang2019} with the two wall treatments, in units of $\rho_0(2RT_0)^{3/2}$: (a,c) $q_x$ on the vertical centreline and (b,d) $q_y$ on the horizontal centreline, at (a,b) $\Kn=0.1$ and (c,d) $\Kn=0.01$. Grey line, triangles and squares: published DVM, R26 and R13. Solid blue: present R26 with the wall relations of \S\ref{sec:discussion} ($56^2$ and $90^2$). Dashed purple: present R26 with the wall-cell treatment ($56^2$ and $48^2$); the steep part next to each wall spans one cell (figure~\ref{fig:wallcell_near_wall}c). Orange line with band: present DSMC.}
  \label{fig:walltreat_cavity}
\end{figure}

\begin{center}
\begin{minipage}{\linewidth}
\centering
\begin{tabular}{llrrrrr}
\toprule
$\Kn$ & wall treatment & grid & lid $q_x$ & left $q_y$ & right $q_y$ & $L_2$ vs DSMC ($q_x$, $q_y$) \\
\midrule
0.1 & printed relations & $56^2$ & $-0.47$ & $-1.60$ & $1.46$ & 0.28, 0.28 \\
    & wall-cell & $25^2$ & $-5.30$ & $0.91$ & $-1.05$ & 0.22, 0.42 \\
    & wall-cell & $34^2$ & $-5.43$ & $0.97$ & $-1.10$ & 0.26, 0.44 \\
    & wall-cell & $48^2$ & $-5.49$ & $1.00$ & $-1.13$ & 0.29, 0.46 \\
    & wall-cell & $56^2$ & $-5.51$ & $1.00$ & $-1.13$ & 0.30, 0.46 \\
    & published R26 & & $-6.58$ & $0.38$ & $-0.48$ & \\
    & DSMC (first cell) & & $-5.88$ & $0.43$ & $-0.43$ & \\
\addlinespace
0.01 & printed relations & $90^2$ & $-0.80$ & $0.19$ & $-0.20$ & \\
     & wall-cell & $48^2$ & $-1.24$ & $0.32$ & $-0.33$ & \\
     & published R26 & & $-1.43$ & $0.39$ & $-0.40$ & \\
\bottomrule
\end{tabular}
\captionsetup{format=revisioncaption}
\captionof{table}{Heat flux at the walls of the cavity of \citet{YangTangYang2019} with the two wall treatments, in units of $10^{-4}\rho_0(2RT_0)^{3/2}$, as in table~\ref{tab:pub_cavity}, and relative $L_2$ differences from the present DSMC over the centrelines. The $L_2$ values use the first DSMC sample (18 of the 57 blocks); with the full sample of \S\ref{app:pub_cavity} the $64^2$ values with the printed relations are 0.27 and 0.28. All $\Kn=0.1$ rows except published and DSMC are roots of the discrete equations.}
\label{tab:walltreat}
\end{minipage}
\end{center}

\subsection{\texorpdfstring{\colorbox{white}{\parbox{0.75\linewidth}{Closures next to the wall and the hybrid method}}}{Closures next to the wall and the hybrid method}}\label{app:pub_hybrid}

\begin{shaded}
DSMC shows why the R26 equations miss part of the wall layer. Figure~\ref{fig:closure_wall_distance} compares $\psi_{xyy}$, $\Omega_x$ and $\Phi_{xyyy}$ measured in our DSMC of the $\Kn=0.1$ cavity with their gradient closures (appendix~\ref{app:actual_algorithm}) evaluated from the same DSMC fields, as functions of distance from the lid. Within a quarter of a mean free path of the lid the closures of $\psi_{xyy}$ and $\Omega_x$ have the sign opposite to the measured moments. Between 0.5 and 1 mean free path they exceed them by factors of 1.5 and 4.4. The closure of $\psi_{xyy}$ is within 15\% of the measured moment between 1 and 2.5 mean free paths from the lid, and that of $\Omega_x$ only beyond about 2.5. The wall relations of \S\ref{sec:discussion} use these closures at the wall itself. Table~\ref{tab:new_closure} gives the corresponding comparison over the interior of the cavity at $\Kn_{\rm Gu}=0.05$.
\end{shaded}

\begin{shaded}
A conference paper on the same comparison has the authors of the R26 equations among its authors \citep{YangGuEmersonZhangTang2019RGD}, and in the same year that group published a hybrid method. It solves the kinetic equation by the DVM in a layer next to each wall and the R26 equations in the rest of the flow, because the R26 equations lose accuracy close to walls as the Knudsen number increases \citep{YangGuEmersonZhangTang2019,YangGuWuEmersonZhangTang2020}. In a later application to flow past a square cylinder at $\Kn=0.1$ the kinetic layer is $0.113L_0$ thick, about one mean free path \citep{GuYangEmersonZhang2022}, and the method is the subject of a monograph \citep{YangMenLi2026}. A solution obtained in this way carries the kinetic heat flux at the wall. \citet{YangTangYang2019} do not say how the points next to the walls were computed. Their R26 values there coincide with the DVM at $\Kn=0.01$ and follow it at the lid at $\Kn=0.1$, which the R26 equations do not produce and their R13 values do not show. A hybrid solution would carry the kinetic heat flux at the wall, but the wall-cell treatment of \S\ref{app:pub_walltreat} also brings the lid value close to the published one, and gives the side walls the kinetic sign, without a kinetic solver. We cannot tell from the published material which, if either, was used.

In summary, we have found no indication of an error in our implementation. It reproduces the published R26 and R13 Poiseuille profiles at $\Kn=0.1$ and $0.5$ within 0.1\% and 0.2\%, the cavity velocities within 6\% at $\Kn=0.1$ and 2\% at $\Kn=0.01$, and the cavity heat flux at $\Kn=0.1$ more than $0.1L_0$ from the walls within 16\%. It also agrees with DSMC over the whole cavity at $\Kn_{\rm Gu}=0.05$, within 1.5\% in velocity and 13\% in heat flux (table~\ref{tab:current_fields}). The published R26 heat flux next to a wall cannot serve as a test in the same way. In Couette flow it is reproduced by reversing the sign of one term, and less closely by the wall-cell treatment; in the cavity it equals the kinetic value, and the wall-cell treatment comes close to it at the lid but not within the first mean free path. Neither is the boundary-value problem of the R26 equations with the wall relations derived from diffuse reflection.
\end{shaded}

For heat flux we therefore use the DSMC comparisons, subject to their sampling and resolution limits. The centreline velocity and heat-flux profiles alone cannot determine the full higher-moment state (\S\ref{sec:theory}).

\section{\texorpdfstring{\colorbox{white}{\parbox{0.75\linewidth}{Moment estimators and the effect of pooling DSMC samples}}}{Moment estimators and the effect of pooling DSMC samples}}
\label{app:same_state}

\begin{shaded}
We pooled equal exposures from seven DSMC realisations (random seeds 104729, 130363, 155921, 196613, 242081, 32452843 and 49979687). We checked the particle weights, cell volume and sampling counts in the raw estimator. The pool contains 140,000 time samples before centralisation. Successive samples can be correlated. For moments formed from the pooled raw samples, $P=\rho\Sigma$ and $C^{(4)}$ includes $\rho$. The contracted excess is $A_{ij}=C^{(4)}_{ijkk}-7\theta\sigma_{ij}-5\rho\theta^2\delta_{ij}$. Equation~\eqref{eq:finite_N_fourth} applies to the earlier estimator that takes central moments within each sample. Pooling the raw samples before centralisation uses a different estimator.
\end{shaded}

\begin{shaded}
For each closure, let $e_s=Z_s-C_s$ be the difference between the DSMC target and closure prediction in realisation $s$, $\bar e=7^{-1}\sum_s e_s$ the mean of these differences, and $e_p=Z_p-C_p$ the difference obtained after pooling. We use SI units and the same spatial weights and Cartesian component multiplicities in every inner product. The following quantities distinguish changes in the target and the prediction from variation between realisations.
\end{shaded}

\begin{shaded}
\begin{equation}
B_Z=\bar Z-Z_p,\quad B_C=\bar C-C_p,\quad B=B_Z-B_C=\bar e-e_p,
\quad M={1\over7}\sum_s\|e_s\|^2,\quad V={1\over7}\sum_s\|e_s-\bar e\|^2.
\label{eq:pool_def}
\end{equation}
\end{shaded}

\begin{shaded}
\begin{equation}
\begin{aligned}
M&=\|\bar e\|^2+V,\\
\|\bar e\|^2&=\|e_p\|^2+2\langle e_p,B\rangle+\|B\|^2,\\
\|B\|^2&=\|B_Z\|^2+\|B_C\|^2-2\langle B_Z,B_C\rangle.
\end{aligned}
\label{eq:pool_identity}
\end{equation}
\end{shaded}

\begin{shaded}
A separate calculation from the component arrays reproduced table~\ref{tab:new_closure} and the three identities, with a maximum relative defect of $1.21\times10^{-15}$. The cross term matters because changes in the target and closure prediction can cancel even when both are large. The ratio $V/M$ measures variation between realisations. Neither this ratio nor the reduction after pooling identifies the share of the error due to sampling noise. The tensor-RMS comparison reported separately uses local units $\rho\theta^{k/2}$ in place of fixed SI units.
\end{shaded}

\begin{shaded}
We used $C_1=2.097$ and $C_2=0.291$ in $\Phi$, $Y_1=1.698$, $Y_2=1.203$ and $Y_3=0.854$ in $\psi$, and the stated relation for $\Omega$. The R26 gas constant was $R_{\rm gas}=208$ and the particle conversion used $k_B/m=208.1328107334$. These coefficients were fixed throughout. Agreement of the higher-order collision rates remains uncertain.

Using centred differences over two cell spacings gives relative discrepancies of 2.070793, 1.756491 and 4.221023 for $\Phi$, $\psi$ and $\Omega$, respectively. Table~\ref{tab:new_closure} uses neighbouring cells. Both stencils act on the same grid, isolating the effect of the derivative approximation at fixed spatial resolution.
\end{shaded}

\begin{shaded}
For an earlier iterate on the uniform R26 grid, we found positive representations of all 28 target moments at 347 of the 352 non-corner wall nodes. At a given node one parameter interval applies to every component, although this interval varies along the wall. At the five nodes with zero-based $(y,x)$ indices $(87,89),(88,89),(89,86),(89,87),(89,88)$ we could not reconstruct the target moments exactly, so their realisability remains open. These pointwise results concern that earlier iterate. They leave open realisability for the field used in the comparisons, existence of an interior distribution consistent with the wall and corner data, and existence of a second R26 solution.
\end{shaded}

\input{paired_bulk_coupling_20261020.tex}
\input{finite_wall_coupling_addendum_20261020.tex}

\input{conditional_observation_azz_wall_addendum_20261021.tex}

\input{native_n60_value_agreement_addendum_20261021.tex}
\input{r26_wall_residuals.tex}

\clearpage
\section{\texorpdfstring{\colorbox{white}{\parbox{0.75\linewidth}{Quadratic entropy and kinetic-reference diagnostics}}}{Quadratic entropy and kinetic-reference diagnostics}}
\label{app:r26_entropy}

\begin{shaded}
We evaluated the quadratic entropy functional for linear R26 given by \citet{RanaGuptaSprittlesTorrilhon2021} (their equation~25 and appendix~A) on the nonlinear Gu--Maxwell cavity field at $\Kn_{\rm Gu}=0.20$. The functional uses a fixed Gaussian reference, with density $\rho$ and temperature $\theta$ normalized by their equilibrium values. It gives
\end{shaded}

\begin{shaded}
\begin{equation}
\begin{aligned}
\eta-a_0=-\bigg[&{(\rho-1)^2\over2}+{v_i v_i\over2}+{3(\theta-1)^2\over4}
+{\sigma_{ij}\sigma_{ij}\over4}+{q_iq_i\over5}\\
&+{m_{ijk}m_{ijk}\over12}+{R_{ij}R_{ij}\over56}+{\Delta^2\over240}\bigg].
\end{aligned}
\label{eq:r26_entropy}
\end{equation}
\end{shaded}

\begin{shaded}
The contractions include all three spatial directions. We recovered the $zz$ and $xzz,yzz$ components from the trace-free relations for the 17 planar variables. The quadratic form makes $\eta-a_0$ non-positive. The arbitrary offset $a_0=1$ does not affect entropy production or the interpretation of positivity. The nonlinear Boltzmann entropy requires more information than these moments provide.
\end{shaded}

\begin{shaded}
With the R26 reference speed $c_0=\sqrt{208\times300}\ \mathrm{m\,s^{-1}}$, $\eta-a_0$ ranges from $-0.155313766$ to $-0.000238618$ over 8,100 nodes. Its arithmetic mean is $-0.00552627349$. The mean is $-0.00508296782$ at the 7,744 interior nodes and $-0.01466804421$ at the 352 non-corner wall nodes. The four corners enter only the all-node mean (the other groups contain 8,096 nodes). Trapezoidal integration over the unit square gives $-0.00530278612$. Neither the spatial means nor this area integral has the particle weighting of a total thermodynamic entropy. The moments $m,R,\Delta$ account for a mean fraction of 0.0678850 of the local quadratic penalty, with a maximum of 0.3855213. This is the average of the local fractions. Taking the ratio of spatially integrated penalties would weight the points differently.
\end{shaded}

\begin{shaded}
To compare with the particle reference speed $c_{0,p}=\sqrt{k_B300/m}=249.8796575\ \mathrm{m\,s^{-1}}$, we use $\lambda=c_0/c_{0,p}=0.9996808962$. Conversion multiplies velocity by $\lambda$, temperature and stress by $\lambda^2$, heat flux and $m$ by $\lambda^3$, and $R,\Delta$ by $\lambda^4$. The all-node mean then becomes $-0.00551591341$. The small change measures the effect of the reference speed. It does not test the higher-order Maxwell/VSS collision rates. The accuracy of the quadratic approximation far from equilibrium remains unknown.
\end{shaded}

\begin{figure}[htbp]
\centering
\includegraphics[width=\linewidth]{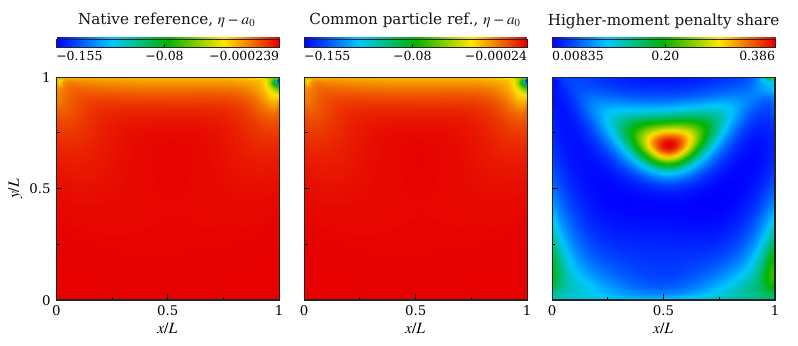}
\captionsetup{format=revisioncaption}
\caption{Quadratic entropy of the nonlinear $90\times90$ R26 cavity field at $\Kn_{\rm Gu}=0.20$, evaluated with the fixed Gaussian reference of \eqref{eq:r26_entropy}. The panels distinguish the reference-speed conventions and the fraction of the local quadratic penalty due to $m,R,\Delta$.}
\label{fig:r26_entropy}
\end{figure}

\input{entropy_methods_and_comparison.tex}

\begin{shaded}
An entropy balance, $\partial_t\eta+\nabla\cdot\Gamma=\Sigma$, also requires the entropy flux and the collision and wall production for the same model. We evaluated only the density functional at one iterate. Its value supplies neither the balance defect nor a test of monotonicity in time or during iteration. The expression $q\cdot\nabla(1/T)=-q\cdot\nabla T/T^2$ gives only the conventional heat-flux contribution to production. A nonlinear R26 entropy theorem would also have to include the Gu wall conditions. The DSMC histogram estimate depends on the velocity bins, sampling window, density and reference distribution. The 56 moments alone do not determine that estimate. \citet{VenugopalPraturiGirimaji2019} tested the local entropy statement of kinetic cavity solutions with the $H$-function and with the Sackur--Tetrode entropy. Our quantity differs in method and in kind: it is the quadratic approximation of the $H$-function about a fixed Maxwellian, evaluated on a moment solution, and it is an entropy density rather than a production. The two results therefore cannot be compared directly.
\end{shaded}

\section{\texorpdfstring{\colorbox{white}{\parbox{0.75\linewidth}{Steady iteration and local amplification}}}{Steady iteration and local amplification}}
\label{app:numerical_modes}

\begin{shaded}
Let $x=\mathcal E(U)$ denote the solver variables, with logarithmic coordinates for density and temperature only. Let $F_h(x)$ be the steady residual including its wall, corner and global equations. At an iterate $x_0$ the local expansion is
\end{shaded}

\begin{shaded}
\begin{equation}
F_h(x_0+\delta x)=r_0+J_0\delta x+O(\|\delta x\|^2),
\qquad r_0=F_h(x_0),\quad J_0=D_xF_h(x_0).
\end{equation}
\end{shaded}

\begin{shaded}
The constant term $r_0$ remains finite at the R26 iterate used in the comparisons. For diagonal variable and residual scales the matrix is $\widehat J=D_F^{-1}J_0D_x$. Its singular values and condition estimates depend on the choice of variables, units and scales. A pair satisfying $\widehat Jv\simeq\sigma u$, $\widehat J^Tu\simeq\sigma v$ describes the steady residual linearised at this iterate. A search for another nonlinear root requires solving the nonlinear equations. Temporal stability requires the semidiscrete evolution equations and linearised boundary conditions, for example $M_*\dot{\delta U}+K_*\delta U=0$ with $K_*v=-\lambda M_*v$. The Levenberg--Marquardt system and the steady Jacobian with replacement rows do not supply this time-evolution problem.
\end{shaded}

\begin{shaded}
For an earlier transient calculation on the $12^2$ grid, we froze the coefficients at one time step and evaluated 456 local interior models at 201,096 Fourier modes. In the normalized basis of 26 raw moments the amplification matrix is $M(k)=C[I+\delta t L(k)]$, where $\delta t=5.55555556\times10^{-5}$, $L$ is the interior transport and regularisation symbol with fixed coefficients, and $C$ is the Jacobian of the implicit collision step. The largest sampled spectral radius was 1.0000001603863042. In 47 models, 926 modes exceeded the chosen radius threshold of 1.00000001. The largest normalized one-step norm was 1.000009539862506. We also examined norms of selected powers to distinguish growth over several steps from one-step amplification. These local calculations imposed no contraction condition. Global stability of the $90^2$ cavity field, including its spectrum and pseudospectrum, requires the spatially varying coefficients and wall equations.
\end{shaded}

\begin{shaded}
On the $29^2$ and $44^2$ grids the estimated Jacobian 1-norm condition numbers were $4.2450\times10^8$ and $2.2031\times10^9$. Equilibration reduced them to $4.4940\times10^6$ and $1.0858\times10^7$, respectively. With consistent changes of variables the exact Newton correction is invariant under this scaling. These estimates concern the two frozen matrices and do not identify the cause of slow convergence on the $90^2$ grid. The accompanying calculation records give the source data and limits of the directional reconstruction.
\end{shaded}

\begin{shaded}
Two frozen matrices on $30^2$ grids gave approximate singular values of $3.5665064\times10^{-4}$ and $2.7820344\times10^{-4}$. At a root on the $29^2$ grid the value was $8.8200664\times10^{-4}$. The residuals of the left and right singular-vector relations were of order $10^{-9}$ or smaller. A small singular value identifies a direction in which the steady residual changes little to first order. It does not establish an admissible perturbation, a fold or a bifurcation. We have not calculated a time propagator or temporal eigenmodes for the principal $90^2$ field.
\end{shaded}

\subsection{\texorpdfstring{\colorbox{white}{\parbox{0.75\linewidth}{Iteration of the sharp-lid R26 equations}}}{Iteration of the sharp-lid R26 equations}}
\label{app:recorded_behaviour}

\begin{shaded}
The additional R26 moments couple transport, relaxation and the wall layers. \citet{GuEmerson2009} examined these effects in channel flows. The cavity solver of \citet{RanaTorrilhonStruchtrup2013} uses coupled boundary conditions with the R13 equations. The entropy theorem and boundary conditions of \citet{RanaGuptaSprittlesTorrilhon2021} apply to linear R26. These different models and geometries do not provide a convergence benchmark for the present nonlinear R26 cavity. For Newton--Krylov iteration, preconditioning can also determine the convergence rate \citep{KnollKeyes2004}. The present tests do not isolate a single cause of slow convergence.
\end{shaded}

\begin{shaded}
On the $90\times90$ grid there are 137,700 unknowns ($90\times90\times17$), with 352 fixed velocity entries and five constraints, namely the four global balances and one wall equation (appendix~\ref{app:complete_continuation}). The solver projects and regularises its linear corrections and tests each step against the nonlinear equations. The last nonlinear step reduced the scaled RMS from $5.693\times10^{-4}$ to $5.656\times10^{-4}$, a decrease of $0.64\%$ (figure~\ref{fig:convergence_earliest_new}).
\end{shaded}

\renewcommand\floatpagefraction{.6}
\subsection{\texorpdfstring{\colorbox{white}{\parbox{0.75\linewidth}{Spatial distribution of the equation residuals}}}{Spatial distribution of the equation residuals}}
\label{app:residual_fields}

\begin{table}
\input{equation_rms_table.tex}
\end{table}

\begin{shaded}
Figure~\ref{fig:r26_residual_components} shows the unscaled residual of the principal $90\times90$ R26 calculation at $\Kn_{\rm Gu}=0.20$. Each bulk equation contributes 7744 interior nodes except continuity. The global mass constraint replaces continuity at one central node, leaving 7743 interior continuity equations. At the walls the equations impose boundary relations and extrapolations, and at the corners they complete the state. Grey marks these rows and the central mass constraint. We evaluated continuity at the central node separately.
\end{shaded}

\begin{shaded}
These maps show the pointwise magnitude of the unscaled interior residual. The supplementary atlas in \S\ref{sec:model_diagnostic} uses its Gaussian local RMS. Each equation has its own logarithmic scale. Table~\ref{tab:equation_rms} gives the corresponding root-mean-square values for all 17 equations. The full-system values in table~\ref{tab:residual_measures} also include the boundary and mass equations.
\end{shaded}

\begin{shaded}
Figures~\ref{fig:resid_r13_kn005}--\ref{fig:resid_r26_kn005} show the other three fields on $60\times60$ grids. For R13 at $\Kn_{\rm Gu}=0.05$ and $0.20$, the plotted quantity is the unbordered residual $A(U)U-b(U)$ in the original units. For R26 at $\Kn_{\rm Gu}=0.05$, it is the solver-scaled residual on a stretched grid, using the same R26 operator as at $\Kn_{\rm Gu}=0.20$ (appendix~\ref{app:eq25_resolution}). Table~\ref{tab:equation_rms} gives the interior RMS values from these arrays and from the $90^2$ R26 field. The R13 values range from $10^{-16}$ to $3\times10^{-10}$ and the R26 values from $4\times10^{-7}$ to $2\times10^{-3}$. Their different scalings preclude a ranking by residual magnitude. At $\Kn_{\rm Gu}=0.05$ the R26 residual concentrates near the lid corners, especially the upper-left one. The portions of the cavity within $0.1L$ of the two corners contain between 54\% and 100\% of the squared interior residual of each equation. At $\Kn_{\rm Gu}=0.20$ the residual extends over the upper half of the cavity, with at most 18\% in those corner regions.
\end{shaded}

\begin{figure}[p]
\centering
\includegraphics[width=\textwidth,height=0.79\textheight,keepaspectratio]{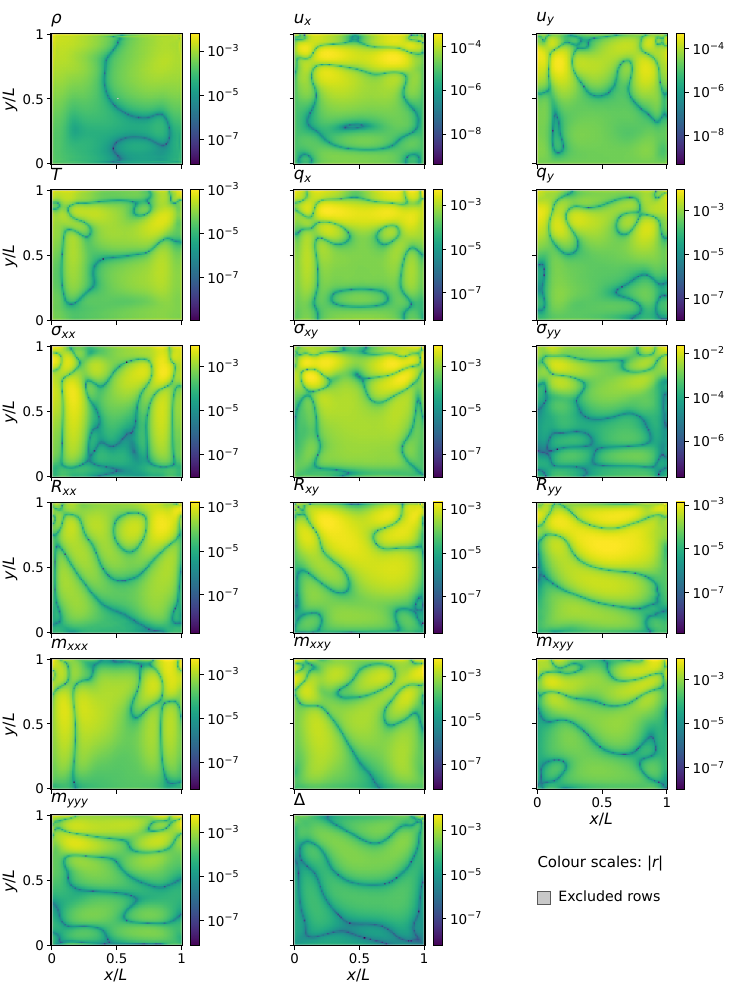}
\captionsetup{format=revisioncaption}
\caption{Magnitude $|r|$ of the unscaled residual for the 17 equations of the $90\times90$ R26 field at $\Kn_{\rm Gu}=0.20$. Each colour scale spans six decades below the maximum for that equation. The maps use the unsmoothed interior values. Grey marks the wall and corner equations and the row replaced by the mass constraint. The supplementary atlas in \S\ref{sec:model_diagnostic} gives local RMS values for all four cases.}
\label{fig:r26_residual_components}
\end{figure}
\begin{figure}[p]
\centering
\includegraphics[width=\textwidth,height=0.79\textheight,keepaspectratio]{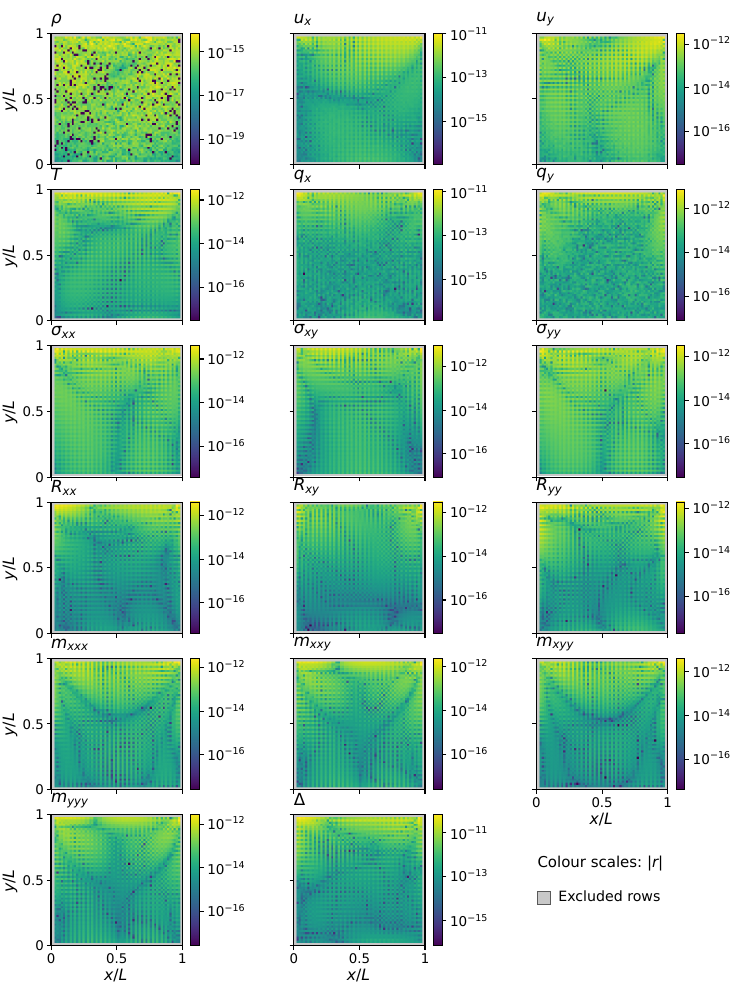}
\captionsetup{format=revisioncaption}
\caption{Interior residual magnitudes for the 17 R13 equations at $\Kn_{\rm Gu}=0.05$ on the $60\times60$ grid. The residual is $A(U)U-b(U)$ in the original units before adding the mass constraint. The last eight equations supply the regularising closures. As in figure~\ref{fig:r26_residual_components}, each scale spans six decades below its own maximum. Wall and corner rows are grey.}
\label{fig:resid_r13_kn005}
\end{figure}
\begin{figure}[p]
\centering
\includegraphics[width=\textwidth,height=0.79\textheight,keepaspectratio]{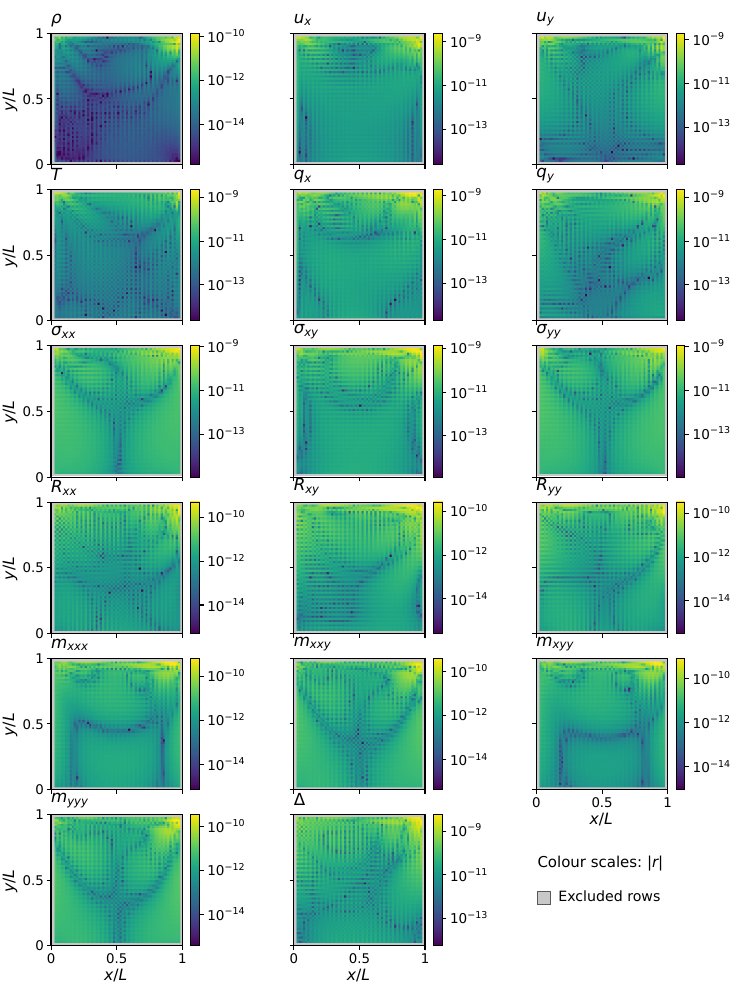}
\captionsetup{format=revisioncaption}
\caption{Interior residual magnitudes for the 17 R13 equations at $\Kn_{\rm Gu}=0.20$ on the $60\times60$ grid. The residual is $A(U)U-b(U)$ in the original units before adding the mass constraint. The last eight equations supply the regularising closures. As in figure~\ref{fig:r26_residual_components}, each scale spans six decades below its own maximum. Wall and corner rows are grey.}
\label{fig:resid_r13_kn020}
\end{figure}
\begin{figure}[p]
\centering
\includegraphics[width=\textwidth,height=0.79\textheight,keepaspectratio]{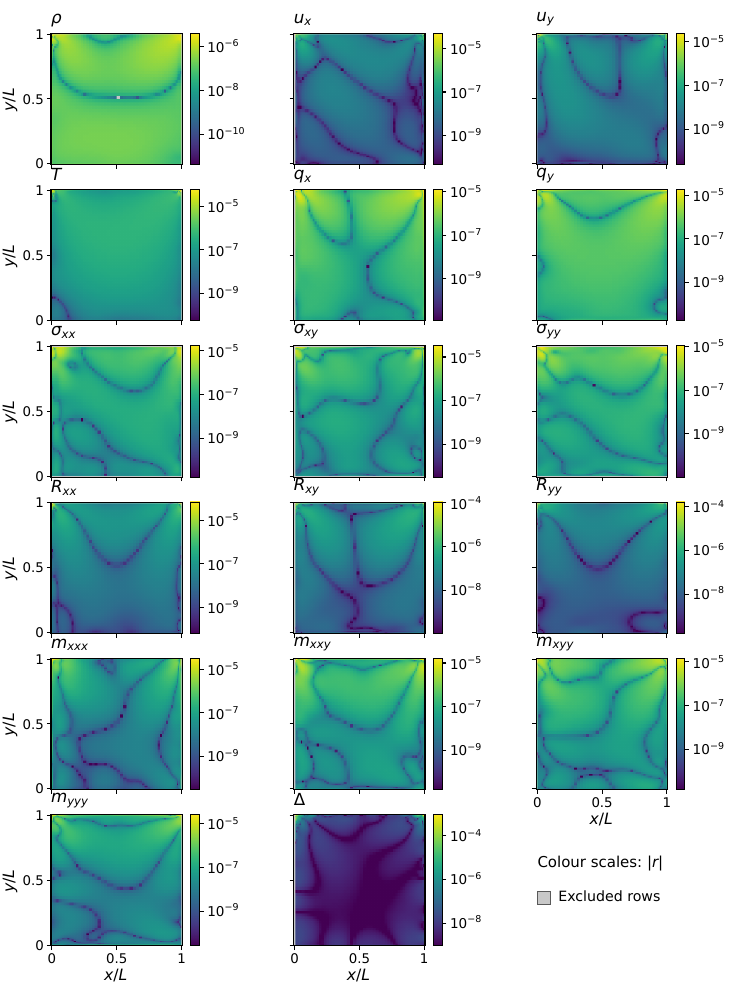}
\captionsetup{format=revisioncaption}
\caption{Interior residual magnitudes for the 17 R26 equations at $\Kn_{\rm Gu}=0.05$ on the stretched $60\times60$ grid. The field uses the same R26 operator as the $\Kn_{\rm Gu}=0.20$ calculation (appendix~\ref{app:eq25_resolution}). The residual uses the solver scaling. As in figure~\ref{fig:r26_residual_components}, each scale spans six decades below its own maximum. The wall and corner equations and the row replaced by the mass constraint are grey.}
\label{fig:resid_r26_kn005}
\end{figure}

\subsection{\texorpdfstring{\colorbox{white}{\parbox{0.75\linewidth}{Continued iteration of the complete system and grid sensitivity}}}{Continued iteration of the complete system and grid sensitivity}}
\label{app:complete_continuation}

\begin{shaded}
For the sharp-lid $90\times90$ R26 calculation at $\Kn_{\rm Gu}=0.20$ (appendix~\ref{app:recorded_behaviour}), the prescribed criteria were maxima of $3\times10^{-2}$ for the unscaled residual and the wall equations, mass and omitted-continuity defects below $10^{-2}$, energy and momentum defects below $10^{-7}$ and $2\times10^{-7}$, respectively, and positive density, temperature and pressures. The field used in the main comparisons meets the balance and positivity criteria. Its maximum unscaled residual is $3.95\times10^{-2}$ and its maximum wall residual is $3.46\times10^{-2}$. Independent re-evaluation gives an RMS of $5.59\times10^{-4}$ over all 137\,700 scaled equations (table~\ref{tab:residual_measures}). The iteration itself recorded $5.656\times10^{-4}$ (appendix~\ref{app:recorded_behaviour}). We use the independently evaluated residuals for the following comparisons.

From this field we continued the complete discrete boundary-value problem with constrained trust-region Gauss--Newton steps, a lagged Jacobian and inexact Krylov solves. The system retains all seventeen equations at each of the 8100 nodes, including the wall conditions and the six auxiliary \mbox{(free-)extrapolation} equations per wall node (appendix~\ref{app:wall_residual_rows}). The RMS fell to $4.90\times10^{-4}$. The four global balances and the wall-normal heat-flux extrapolation at the lid node adjacent to the top-right corner hold to better than $10^{-10}$. The preceding iteration had held this extrapolation residual at $-2.53\times10^{-2}$. In the continuation we constrained it to zero. Density, temperature and wall pressure remained positive. The largest wall-condition residual, $2.996\times10^{-2}$, occurs at a node next to that corner.
\end{shaded}

\begin{shaded}
Six auxiliary equations still exceed the original limit, with a maximum residual of $3.28\times10^{-2}$. They involve the wall-normal heat flux and the normal--tangential--tangential component of $m_{ijk}$. All six lie within three nodes of the corner where the moving lid meets the stationary right wall and the continuum fields are singular. After inspecting the residuals we separated the auxiliary equations within three nodes of either lid corner (72 of 137\,700 rows). Outside this group every row meets the $3\times10^{-2}$ limit. The maximum residual still exceeds the $10^{-5}$ required for a converged root of the complete system.
\end{shaded}

\begin{shaded}
Continuing the iteration changed $\Aten$ by $0.23\%$ and $\qv$ by $0.76\%$ over $0.05\le x/L,\,y/L\le0.95$. These are relative $L^2$ differences with trapezoidal weights. The Frobenius norm of $\Aten$ counts $A_{xy}$ twice. Over the whole cavity the changes are $0.94\%$ and $3.7\%$, mostly within five nodes of the lid corners. For comparison, the separately stored $90^2$ field used in table~\ref{tab:canonical_table6} has R26--DSMC relative errors of 0.62--0.98 in the in-plane components and 0.96 in the full tensor, two orders of magnitude larger than the interior changes.
\end{shaded}

\begin{shaded}
On the $120\times120$ and $180\times180$ grids the scaled RMS residuals were $1.07\times10^{-3}$ and $4.33\times10^{-3}$, respectively. The maximum scaled residuals were $0.16$ and $0.74$, and the unscaled maxima $0.68$ and $0.85$. Energy and momentum balance defects reached $4.1\times10^{-4}$ and $7.7\times10^{-5}$, more than two orders of magnitude above their tolerances. We initialised the $120\times120$ grid from the principal $90^2$ field and the $180\times180$ grid from the $120\times120$ iterate. Differences between the refined fields therefore include the effect of incomplete iteration from a common starting field.

On $120\times120$, the RMS fell about twentyfold between steps 1 and 36, from $2.1\times10^{-2}$ to $1.07\times10^{-3}$, while $\Aten$ changed by $0.29\%$ over the interior region defined above. On $180\times180$, steps 1 to 8 reduced the RMS about ninefold and changed $\Aten$ by $0.09\%$. The finer calculation stopped after step 8. An additional step on $120\times120$ required a step length of $2^{-8}$ and reduced the RMS by only $1.8\times10^{-6}$ relative. Steps 34 to 37 together reduced it by $0.2\%$. With the specified Krylov iteration limit the continuation had stalled. We use the field after 36 steps for comparison.

Between the two refined grids, $\Aten$ changes by $0.19\%$ in the interior and $1.6\%$ over the whole cavity, mainly at the lid corners. At the 24 near-lid DSMC positions on $x/L=0.5$ with $0.85\le y/L\le1$, the RMS change in $q_x$ between grids is $8.9\times10^{-5}$. The RMS difference between the $120\times120$ iterate and the equal-window seven-realisation DSMC ensemble of \S\ref{sec:updated_evidence} is $1.0\times10^{-2}$. Omitting the five cells nearest the lid gives $8.8\times10^{-5}$ and $4.1\times10^{-3}$, respectively. Over the cavity the corresponding RMS differences are $4.2\times10^{-4}$ and $7.1\times10^{-3}$. The two iterates differ by one to two orders of magnitude less than R26 differs from DSMC. Both refined grids retain appreciable residuals, so their difference cannot bound the discretisation error.
\end{shaded}

\begin{shaded}
Along $x/L=0.5$ the three R26 heat-flux profiles coincide to within the line width near the lid (figure~\ref{fig:r26_grid_sensitivity}a,b), although $q_x$ on $120\times120$ differs from DSMC by up to $3\times10^{-2}$ near the wall. Over $0.05\le y/L\le0.95$, the relative $L^2$ changes from $120\times120$ to $180\times180$ are $0.32\%$, $0.03\%$, $0.06\%$ and $0.16\%$ for $A_{xx}$, $A_{xy}$, $A_{yy}$ and $A_{zz}$. In the same order, the differences between $120\times120$ and the late-window DSMC pool of \S\ref{sec:updated_evidence} are $162\%$, $56\%$, $53\%$ and $135\%$ (panels c--f).

Using the earlier DSMC fields, finite-particle correction and derivative operator of \S\ref{sec:kn020_projection}, the hidden in-plane fraction is 0.9540 for the principal $90^2$ field, 0.9543 after further iteration and 0.9541 on both refined grids (0.9637--0.9639 with $A_{zz}$). The earlier $28\times28$ field gives 0.9508. For the four $90^2$, continued, $120^2$ and $180^2$ fields, the heat-flux discrepancies on the mask and DSMC pool of table~\ref{tab:current_fields} are 0.3644, 0.3666, 0.3663 and 0.3629, respectively.

Two separately stored $90^2$ fields enter different comparisons. We used the field with an independently evaluated scaled residual RMS of $5.59\times10^{-4}$ to initialise the refined grids, and for tables~\ref{tab:current_fields} and \ref{tab:residual_measures}, figure~\ref{fig:r26_residual_components} and appendix~\ref{app:complete_response}. The common-estimator tensor comparison in table~\ref{tab:canonical_table6} and figures~\ref{fig:canonical_A} and \ref{fig:canonical_sign} uses the other field. Their planar tensors $\Aten$ differ by $0.25\%$ after excluding eight cells at each boundary. On the mask and DSMC pool of table~\ref{tab:current_fields}, their density discrepancies are 0.060 and 0.093, respectively. Both heat-flux discrepancies round to 0.3644.
\end{shaded}

\begin{figure}[tbp]
\centering
\includegraphics[width=\linewidth,height=0.68\textheight,keepaspectratio]{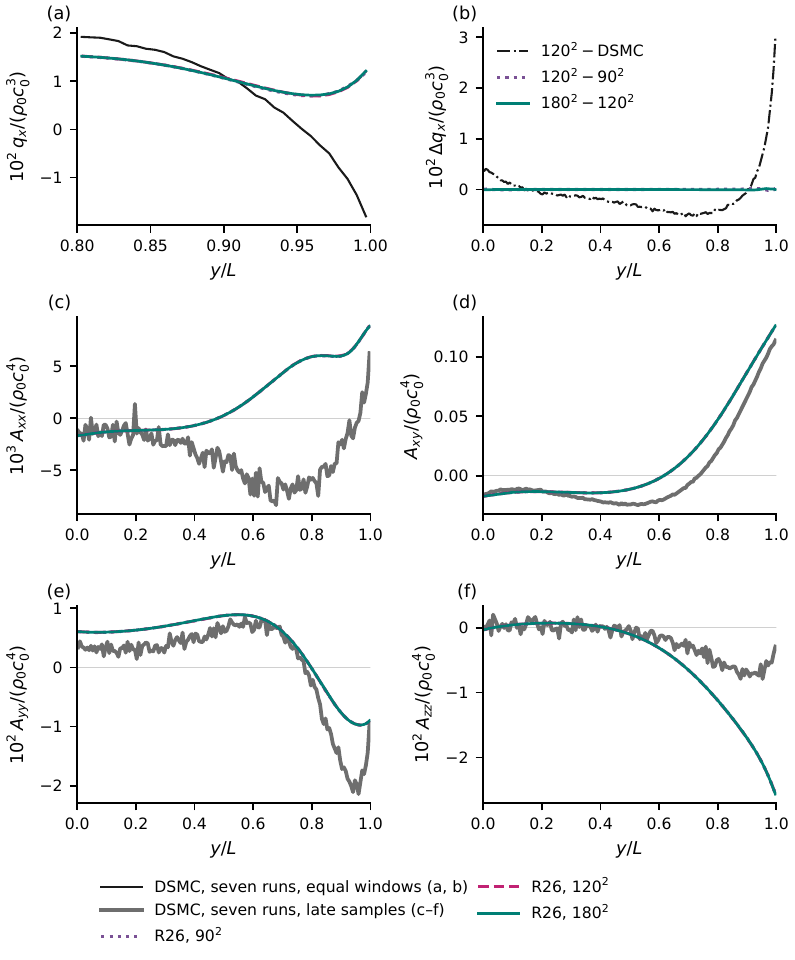}
\captionsetup{format=revisioncaption}
\caption{R26 fields at $\Kn_{\rm Gu}=0.20$ on the vertical centreline $x/L=0.5$. The $90\times90$ field is the principal iterate, the $120\times120$ field is from step 36 and the $180\times180$ field from step 8. All three retain finite residuals above the root tolerance. Panels (a,b) show $q_x$ near the lid and the profile differences using the equal-window seven-realisation DSMC ensemble of \S\ref{sec:updated_evidence}. Panels (c--f) compare $A_{xx}$, $A_{xy}$, $A_{yy}$ and $A_{zz}$ with the late-window seven-realisation pool of \S\ref{sec:updated_evidence}. Linear interpolation maps each R26 field to the native DSMC cell centres. The R26 curves nearly coincide.}
\label{fig:r26_grid_sensitivity}
\end{figure}

\begin{shaded}
We obtained another $90\times90$ field using the Levenberg--Marquardt grid sequence of appendix~\ref{app:solver_reconciliation}. We interpolated the principal $90^2$ field to $10^2$, then bilinearly interpolated the solver variables from each grid to the next. The appendix gives the linear solver, the rule for $\lambda_{\rm LM}$ and the step-acceptance test. The first run used forward differences for the Jacobian through $20^2$ and in its $28^2$ attempt. Subsequent runs used centred differences from $22^2$ onwards. The four runs allowed 60, 80, 40 and 40 steps per grid, respectively. The initial $\lambda_{\rm LM}$ was $10^{-3}$ in the first two and $10^{-6}$ in the last two. The last two also advanced to the next grid when the RMS fell by less than 5\% over six steps.

The sequence reached a largest scaled residual below $10^{-9}$ on $10^2$, $14^2$, $20^2$, $22^2$ and $24^2$. On the thirteen finer grids from $26^2$ to $90^2$, the stopping rule was based on the change in the residual. With the 5\% rule the final scaled RMS increased from $4.4\times10^{-6}$ on $26^2$ to $5.3\times10^{-5}$ on $52^2$ and ranged from $4.98\times10^{-5}$ to $6.06\times10^{-5}$ on $58^2$ through $90^2$. The first two runs had reached their step limits on $28^2$ and $26^2$ at $1.0\times10^{-5}$ and $4.0\times10^{-6}$, respectively.

On $90^2$ the solver took 18 accepted steps in 20 trials from an interpolated field with RMS $7.0\times10^{-2}$. Every accepted step used the maximum number of conjugate-gradient iterations. The last step still reduced the RMS by 0.5\%. The final scaled RMS was $6.06\times10^{-5}$, nine times smaller than for the principal field. The respective maxima were $8.5\times10^{-3}$ and $3.9\times10^{-2}$. The $4\times4$ blocks at the two lid corners contain 81\% of the squared residual for the grid-sequenced field and 40\% for the principal field. Outside these blocks the RMS values are $2.7\times10^{-5}$ and $4.3\times10^{-4}$, respectively. The grid-sequenced field remains above the root tolerance. Its iteration reduces the residual over all rows without constraining the mass equation separately. Its mass defect is $7.9\times10^{-4}$ and its omitted-continuity residual is $2.4\times10^{-4}$. For the principal field both are below $10^{-10}$. Appendix~\ref{app:eq25_resolution} gives the balances for both fields.
\end{shaded}

\begin{shaded}
On the mask of table~\ref{tab:current_fields}, the differences between the two $90^2$ fields are $3.8\times10^{-2}$ in $\rho-1$, $7.3\times10^{-2}$ in $T-1$, $1.8\times10^{-2}$ in velocity and $9.3\times10^{-2}$ in heat flux. These are $L^2$ differences relative to the principal field. Over the whole comparison grid they are $6.0\times10^{-2}$, $1.4\times10^{-1}$, $2.6\times10^{-2}$ and $1.8\times10^{-1}$, with the largest local changes next to the lid corners. After mapping to the DSMC cells and omitting eight boundary cells, the planar tensors $\Aten$ differ by 6.3\%. On the near-centre cuts in figure~\ref{fig:kn020_gridseq}, the largest difference between the two R26 profiles is at most 36\% of the smaller of their two maximum differences from DSMC in each panel. For the lid-parallel velocity the ratio is 4\%.

With the DSMC pool and mask of table~\ref{tab:current_fields}, the four R26 discrepancies are 0.0608, 0.2450, 0.0556 and 0.3442 for the grid-sequenced field and 0.0601, 0.2460, 0.0584 and 0.3644 for the principal field. Thus the changes are $+1.1\%$, $-0.4\%$, $-4.8\%$ and $-5.6\%$, respectively. With the DSMC fields, correction and operator of \S\ref{sec:kn020_projection}, the respective hidden in-plane fractions are 0.9482 and 0.9540 (0.9580 and 0.9637 including $A_{zz}$). Across the 48 projection variants the in-plane ranges are 0.935--0.966 and 0.941--0.971. Reducing the overall residual ninefold by passing through coarse-grid roots thus changes each DSMC discrepancy in table~\ref{tab:current_fields} by less than 6\% and the hidden fraction by 0.006.

The grid-sequenced field has the smaller overall and wall residuals. The principal field has smaller mass, continuity, energy and momentum defects (appendix~\ref{app:eq25_resolution}). The grid-sequenced internal-energy defect, $1.9\times10^{-6}$, exceeds the prescribed $10^{-7}$. Both fields remain above the root tolerance. We used the principal calculation for the continuation, the refined grids and the perturbation test of appendix~\ref{app:complete_response}. The separately stored field used for the common-estimator fourth-order comparison is identified above. The grid-sequenced calculation tests the sensitivity of the comparisons to the solver and field.
\end{shaded}

\begin{shaded}
We also compared DSMC with sharp-lid R26 fields on coarser grids where the full discrete equations reach a root (maximum scaled residual below $10^{-9}$, appendix~\ref{app:solver_reconciliation}). Table~\ref{tab:root_states} uses the same DSMC fields, mask and interpolation as table~\ref{tab:current_fields}. At $\Kn_{\rm Gu}=0.20$ it also gives the in-plane hidden fraction calculated as for the grid-sequenced field.

At $\Kn_{\rm Gu}=0.05$ the density, temperature, velocity and heat-flux discrepancies on the finest root ($38^2$) are 0.0201, 0.0436, 0.0165 and 0.1345. The corresponding values on the principal $60^2$ field are 0.0200, 0.0379, 0.0149 and 0.1316. Temperature, velocity and heat-flux discrepancies decrease monotonically from $20^2$ to $60^2$, across both roots and the two least-squares iterates. Every root has smaller discrepancies in all four quantities than R13 on $60^2$ (table~\ref{tab:current_fields}). The root heat-flux discrepancies of 0.13--0.19 are also below the R13 value of 0.3016 on $180^2$. For velocity, the $20^2$ R26 value of 0.0331 exceeds the $180^2$ R13 value of 0.026. The finer R26 roots fall below it.

At $\Kn_{\rm Gu}=0.20$ the density and velocity discrepancies decrease under refinement. The temperature discrepancy is 0.219--0.223 on $20^2$ through $28^2$ and 0.245--0.246 on the two $90^2$ fields. The heat-flux discrepancy increases from 0.2822 on the $20^2$ root to 0.3442 and 0.3644 on $90^2$. Every R26 row remains below the R13 row of table~\ref{tab:current_fields}. The two coarse roots give hidden in-plane fractions of 0.9409 and 0.9452, compared with 0.9482 and 0.9540 on $90^2$. Including $A_{zz}$ gives 0.9524 and 0.9555 for the roots. For the spatially varying part the fractions are 0.896 and 0.902. Across the 48 variants on $24^2$ the fraction ranges from 0.921 to 0.969.

Thus the four smaller R26 discrepancies relative to R13 on the principal comparison grids, and the 94--95\% of the in-plane fourth-order discrepancy invisible to the heat-flux balance, are also found at discrete roots. For the refined R13 velocity comparison at $\Kn_{\rm Gu}=0.05$, this ordering holds only for the finer R26 roots. The two roots at $\Kn_{\rm Gu}=0.20$ have spacings of $L/19$ and $L/23$. They show that an incomplete nonlinear solve is unnecessary for these results, while leaving the spatial discretisation error unresolved. The supplementary material contains the root fields and the scripts and numerical results for the table.
\end{shaded}

\begin{center}
\begin{minipage}{\linewidth}
\centering
\footnotesize
\setlength{\tabcolsep}{4pt}
\rowcolors{1}{yellow!40}{yellow!40}
\begin{tabular}{lllcccccc}
\toprule
$\Kn_{\rm Gu}$ & grid & state & RMS & $E_{\rho-1}$ & $E_{T-1}$ & $E_u$ & $E_q$ & hidden \\
\midrule
0.05 & $20^2$ & root & $1.3\times10^{-12}$ & 0.0405 & 0.1098 & 0.0331 & 0.1942 & -- \\
0.05 & $28^2$ & root & $2.3\times10^{-13}$ & 0.0219 & 0.0629 & 0.0202 & 0.1460 & -- \\
0.05 & $34^2$ & root & $1.2\times10^{-12}$ & 0.0204 & 0.0488 & 0.0172 & 0.1365 & -- \\
0.05 & $38^2$ & root & $2.0\times10^{-14}$ & 0.0201 & 0.0436 & 0.0165 & 0.1345 & -- \\
0.05 & $40^2$ & step-limited iterate & $5.2\times10^{-6}$ & 0.0204 & 0.0428 & 0.0162 & 0.1330 & -- \\
0.05 & $60^2$ & principal least-squares & $1.4\times10^{-5}$ & 0.0200 & 0.0379 & 0.0149 & 0.1316 & -- \\
\midrule
0.20 & $20^2$ & root & $4.0\times10^{-11}$ & 0.0905 & 0.2212 & 0.1106 & 0.2822 & 0.9409 \\
0.20 & $24^2$ & root & $3.9\times10^{-13}$ & 0.0770 & 0.2193 & 0.0894 & 0.2909 & 0.9452 \\
0.20 & $28^2$ & step-limited iterate & $1.0\times10^{-5}$ & 0.0698 & 0.2227 & 0.0775 & 0.2985 & 0.9469 \\
0.20 & $90^2$ & grid-sequenced & $6.1\times10^{-5}$ & 0.0608 & 0.2450 & 0.0556 & 0.3442 & 0.9482 \\
0.20 & $90^2$ & principal iterate & $5.6\times10^{-4}$ & 0.0601 & 0.2460 & 0.0584 & 0.3644 & 0.9540 \\
\bottomrule
\end{tabular}
\captionsetup{format=revisioncaption}
\captionof{table}{R26 discrepancies from DSMC for roots and approximate iterates of the sharp-lid equations. A root has a maximum scaled residual below $10^{-9}$. RMS is the root-mean-square scaled residual. The four field discrepancies use the DSMC fields, mask and interpolation of table~\ref{tab:current_fields}. At $\Kn_{\rm Gu}=0.20$, the last column gives the in-plane hidden fraction with the DSMC fields, correction and operator of \S\ref{sec:kn020_projection}. The principal $60^2$ and $90^2$ fields are those used in the main comparisons. The $90^2$ principal field comes from the earlier continuation. The $40^2$ and $28^2$ iterates reached the step limit in the first grid sequences (appendices~\ref{app:solver_reconciliation} and \ref{app:complete_continuation}).}
\label{tab:root_states}
\end{minipage}
\end{center}

\begin{shaded}
To test the influence of the lid-corner discontinuity, we multiplied the lid speed by $[1-\mathrm{e}^{-x^2/\delta^2}][1-\mathrm{e}^{-(L-x)^2/\delta^2}]$, where $x$ is distance along the lid from the left corner and $\delta=0.001^{1/2}L\approx0.032L$. We repeated DSMC with this profile using the sharp-lid SPARTA settings, $160^2$ cells and $10^5$ samples per cell after $2\times10^5$ warm-up steps. With the smooth lid, the R26 grid sequence of appendix~\ref{app:solver_reconciliation} reached a root on every grid tested from $30^2$ to $90^2$ at $\Kn_{\rm Gu}=0.05$, and on $24^2$, $28^2$, $34^2$, $48^2$ and $56^2$ at $0.20$. We selected these roots by their equation residuals. A small least-squares gradient alone can also occur at a finite residual.

Table~\ref{tab:smooth_lid} compares the smooth-lid R26 roots and R13 field with DSMC at the same lid velocity profile. At $\Kn_{\rm Gu}=0.05$ the density, temperature, velocity and heat-flux discrepancies on the $90^2$ R26 root are 0.0253, 0.0476, 0.0165 and 0.1220. They change by at most 3.0\% between the $70^2$, $80^2$ and $90^2$ roots and are comparable with 0.0200, 0.0379, 0.0149 and 0.1316 for the sharp-lid $60^2$ field. On $180^2$ the smooth-lid R13 discrepancies are similar to the sharp-lid R13 results in table~\ref{tab:current_fields}. R26 remains closer to DSMC in all four quantities.

At $\Kn_{\rm Gu}=0.20$ the smooth-lid R26 discrepancies are comparable to the sharp-lid values in table~\ref{tab:root_states}. Between $24^2$ and $56^2$ the velocity discrepancy decreases, while the temperature and heat-flux discrepancies increase. The density discrepancy decreases up to $48^2$ and changes by 2\% on $56^2$. The sharp-lid grids show similar trends from the coarse roots to the $90^2$ iterates. Table~\ref{tab:smooth_lid} includes the smooth-lid R13 result at $\Kn_{\rm Gu}=0.05$. At $0.20$ the $60^2$ and $80^2$ R13 calculations ended with finite residuals above the steady-solution tolerance. At $\Kn_{\rm Gu}=0.05$ the ordering of the four discrepancies thus persists when the lid velocity is continuous and R26 reaches a root. At $0.20$ the evidence for the smoothed boundary condition extends only to coarse roots through $56^2$.
\end{shaded}

\begin{center}
\begin{minipage}{\linewidth}
\centering
\begin{tabular}{lllrrrr}
\toprule
$\Kn_{\rm Gu}$ & model & grid & $E_{\rho-1}$ & $E_{T-1}$ & $E_u$ & $E_q$ \\
\midrule
0.05 & R26 (root) & $90^2$ & 0.0253 & 0.0476 & 0.0165 & 0.1220 \\
0.05 & R13 & $180^2$ & 0.0687 & 0.2342 & 0.0207 & 0.3097 \\
0.20 & R26 (root) & $24^2$ & 0.0759 & 0.1866 & 0.0886 & 0.3135 \\
0.20 & R26 (root) & $56^2$ & 0.0551 & 0.2429 & 0.0615 & 0.3440 \\
\bottomrule
\end{tabular}
\captionsetup{format=revisioncaption}
\captionof{table}{Density, temperature, velocity and heat-flux discrepancies for the Gaussian lid profile in appendix~\ref{app:complete_continuation}. Each R26 root and R13 field is compared with DSMC at the same lid profile using the mask and interpolation of table~\ref{tab:current_fields}. These comparisons concern the smooth-lid boundary-value problem.}
\label{tab:smooth_lid}
\end{minipage}
\end{center}

\begin{figure}[tbp]
\centering
\includegraphics[width=\linewidth]{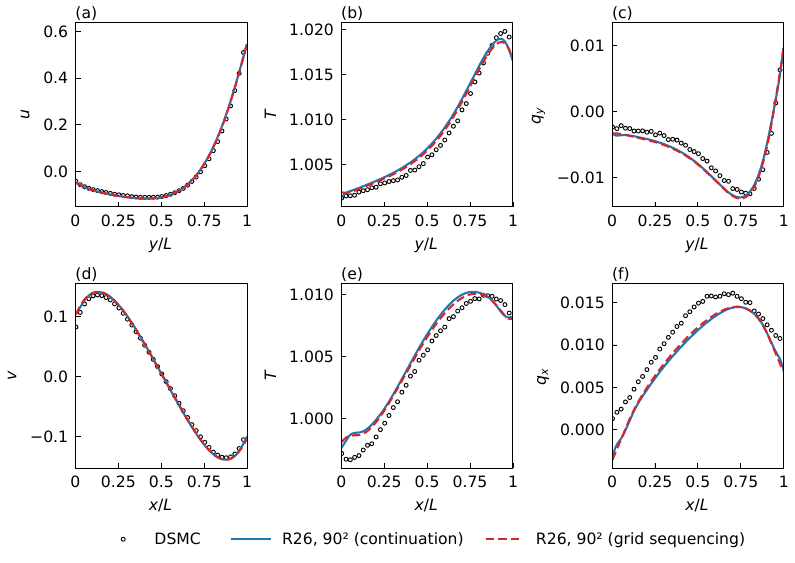}
\captionsetup{format=revisioncaption}
\caption{Near-centre profiles at $\Kn_{\rm Gu}=0.20$ for DSMC (symbols at every fourth cell) and two $90^2$ R26 iterates, from the principal calculation and from grid sequencing (appendix~\ref{app:complete_continuation}). The fields and DSMC pool are those of table~\ref{tab:current_fields}. The $160^2$ comparison grid gives cuts at $79.5/160=0.496875$. Panels (a--c) show $u$, $T$ and $q_y$ against $y/L$ on the vertical cut, and (d--f) show $v$, $T$ and $q_x$ against $x/L$ on the horizontal cut. Velocity, temperature and heat flux are in units of the lid speed, $T_w$ and $\rho_0c_0^3$, respectively.}
\label{fig:kn020_gridseq}
\end{figure}

\begin{shaded}
At $\Kn_{\rm Gu}=0.05$ we computed R13 fields on six grids. The residual RMS ranged from $10^{-14}$ to $1.1\times10^{-9}$ on $40^2$ through $80^2$ (figure~\ref{fig:r13_convergence}). The $100^2$ through $180^2$ grids met the practical tolerance of the refinement study, with RMS values of $1.2\times10^{-6}$ to $2.0\times10^{-5}$. Figure~\ref{fig:r13_grid_kn005} shows the cut $x/L=0.5$ in the format of figure~\ref{fig:r26_grid_sensitivity}, omitting $100^2$ for legibility. At DSMC positions near the lid ($0.85\le y/L\le0.97$), the RMS difference in $q_x$ between R13 and DSMC falls from $9.6\times10^{-3}$ on the principal $60^2$ grid to $4.5\times10^{-3}$ on $120^2$ and $2.7\times10^{-3}$ on $180^2$. The change between the two finest grids, $1.8\times10^{-3}$, is comparable with the remaining discrepancy. Over $0.1\le y/L\le0.9$ the respective discrepancies are $1.5\times10^{-3}$, $9.1\times10^{-4}$ and $7.2\times10^{-4}$, and the change between the finest grids is $2.1\times10^{-4}$.

On the mask of table~\ref{tab:current_fields}, the R13 heat-flux discrepancy falls from 0.6415 on $60^2$ to 0.5075, 0.4265, 0.3748 and 0.3016 on $80^2$, $100^2$, $120^2$ and $180^2$. The temperature discrepancy rises from 0.168 to 0.213 and the velocity discrepancy falls from 0.036 to 0.026. Near-wall resolution dominates the R13 heat-flux discrepancy at $\Kn_{\rm Gu}=0.05$. The heat-flux discrepancy is still decreasing at the finest spacing tested, where it exceeds the $60^2$ R26 discrepancy of 0.1316 (0.1500 for the alternative field of appendix~\ref{app:eq25_resolution}).

The components of the composite tensor depend differently on grid spacing. Between $120^2$ and $180^2$ the relative $L^2$ changes over $0.05\le y/L\le0.95$ are 3.9\%, 31\%, 0.7\% and 1.3\% for $A_{xx}$, $A_{xy}$, $A_{yy}$ and $A_{zz}$, respectively. The corresponding R13--DSMC differences are 56\%, 55\%, 77\% and 73\%. Near the lid $A_{xy}$ remains sensitive to refinement. In contrast $A_{yy}$ and $A_{zz}$ change little and remain far from DSMC (panels c--f). At this Knudsen number the only R26 calculation finer than $60^2$ used the earlier implementation on $90^2$ and retained a residual RMS of 0.40. It cannot supply a comparable fine-grid sequence. At $\Kn_{\rm Gu}=0.20$ no R13 iteration beyond $60^2$ converged (figure~\ref{fig:r13_convergence}b).
\end{shaded}

\begin{figure}[tbp]
\centering
\includegraphics[width=\linewidth,height=0.70\textheight,keepaspectratio]{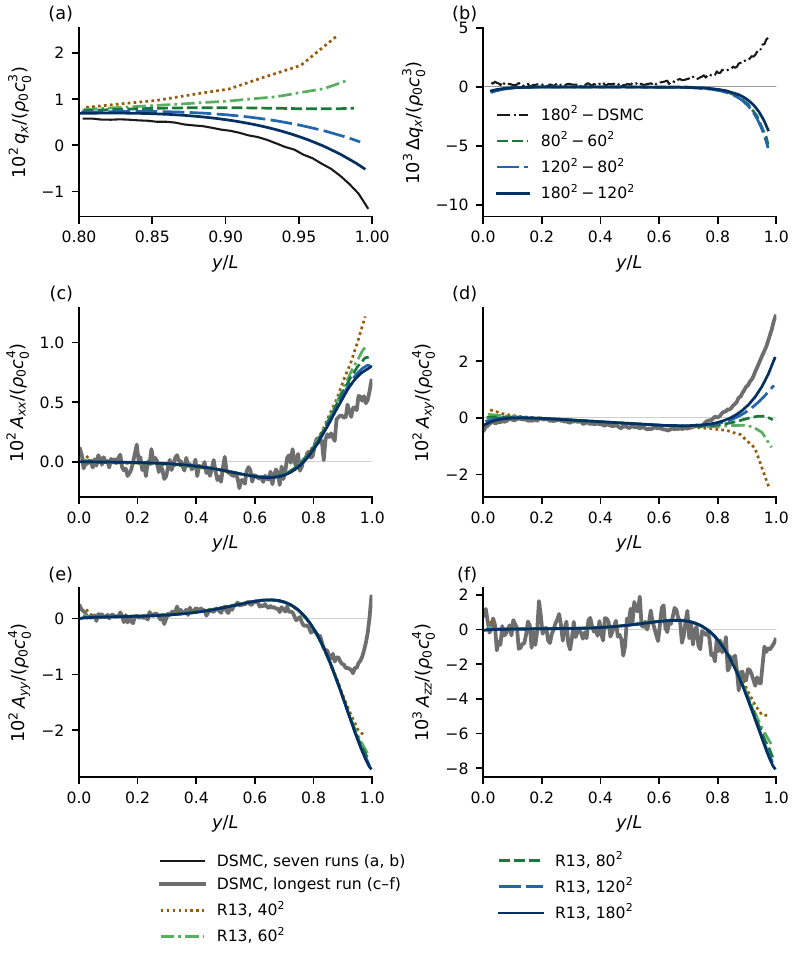}
\captionsetup{format=revisioncaption}
\caption{R13 profiles along $x/L=0.5$ at $\Kn_{\rm Gu}=0.05$. The $40^2$, $60^2$ and $80^2$ grids converged. The $120^2$ and $180^2$ grids met the practical tolerance of the refinement study. The principal grid is $60^2$, and $100^2$ is omitted for legibility. Panels (a,b) compare near-lid $q_x$ and the profile differences with the pool of seven DSMC runs. Panels (c--f) compare $A_{xx}$, $A_{xy}$, $A_{yy}$ and $A_{zz}$ with the single long DSMC run ($10^5$ to $5\times10^6$ steps), the only reconstructed DSMC fourth-moment profile at this Knudsen number. The R13 closure gives $A=R+\Delta I/3$. The fields use R13 units $\rho_0c_0^n$ with $c_0=\sqrt{208\times300}\,\mathrm{m\,s^{-1}}$. Differences between the computed profiles measure the observed grid changes.}
\label{fig:r13_grid_kn005}
\end{figure}

\subsection{\texorpdfstring{\colorbox{white}{\parbox{0.75\linewidth}{R26 response to fourth-order perturbations}}}{R26 response to fourth-order perturbations}}
\label{app:complete_response}

\begin{shaded}
We applied the two perturbations of \S\ref{sec:nullspace} to the principal $90^2$ R26 field and to the field obtained by continuing its nonlinear iteration (appendix~\ref{app:complete_continuation}). For each field we first evaluated all seventeen equations at the 8100 nodes, including the wall and auxiliary equations. We then perturbed only $R^{\cl}_{ij}$ and $\Delta$ by changing $\Aten$ according to \eqref{eq:shift_recover}. Define $w=(1-r^2/r_0^2)^6$ for $r<r_0=0.3L$, and $w=0$ elsewhere, with $r$ measured from $(0.5L,0.45L)$. In (P1) we set $\delta A_{zz}=w$ and $\delta A_{xx}=\delta A_{xy}=\delta A_{yy}=0$. In (P2) we used the Airy potential $\Phi=w$ in \eqref{eq:airy}, with $\delta A_{zz}=0$, and evaluated its second derivatives analytically at the nodes. In the control (P3) we changed $\delta A_{xx}=w$ alone, giving nonzero divergence. All three perturbations vanish within $0.15L$ of every wall. For each we used two amplitudes, $10^{-3}$ and $10^{-1}$ times the RMS of $\Aten$. Both RMS values include $A_{xx}$, $A_{xy}$, $A_{yy}$ and $A_{zz}$ at the 2244 nodes in the support.
\end{shaded}

\begin{center}
\begin{minipage}{\linewidth}
\centering
\rowcolors{1}{yellow!40}{yellow!40}
\begin{tabular}{lccccc}
\toprule
 & \multicolumn{5}{c}{Equations for}\\
Perturbation & $q_i$ & $R^{\cl}_{ij}$, $\Delta$ & $m_{ijk}$ & $\sigma_{ij}$ & $\rho$, $\uv$, $T$ \\
\midrule
(P1) $\delta A_{zz}=w$ & $<10^{-11}$ & 22.0 & 0.61 & 0 & 0 \\
(P2) Airy, $\Phi=w$ & 0.018 & 51.7 & 1.79 & 0 & 0 \\
(P3) $\delta A_{xx}=w$ (control) & 3.36 & 71.8 & 1.14 & 0 & 0 \\
\bottomrule
\end{tabular}
\captionsetup{format=revisioncaption}
\captionof{table}{Change in the unscaled R26 residual following fourth-order perturbations of the principal $90\times90$ field at $\Kn_{\rm Gu}=0.20$. For each group of interior equations the table gives the RMS change divided by the RMS of $\delta\Aten$ over its support. The ratios for the continued field agree within $10^{-5}$ relative. The wall and auxiliary equations are unaffected.}
\label{tab:complete_response}
\end{minipage}
\end{center}

\begin{shaded}
The ratios in table~\ref{tab:complete_response} use the nondimensional solver variables. The two amplitudes give ratios within $10^{-11}$ of each other, consistent with a linear response over this range. The ratios for the two fields agree within $10^{-5}$. The conservation and stress equations contain neither $R^{\cl}_{ij}$ nor $\Delta$. They remain unchanged, as do the wall and auxiliary equations outside the support.

The absolute RMS change in the heat-flux equations under (P1) is $4\times10^{-17}$, at round-off level. For the Airy perturbation (P2) the discrete divergence gives a small response, $0.55\%$ of that to the control (P3). The third- and fourth-order equations respond to all three. For (P1) and (P2), the fourth-order equations for $R^{\cl}_{ij}$ and $\Delta$ give ratios of 22 and 52, respectively. Both moments enter their collision and flux terms directly. On this uniform grid the two directions are therefore invisible to the heat-flux equations to round-off in (P1) and up to the divergence truncation error in (P2). A perturbation in either direction with an RMS of $1\%$ of $\Aten$ on the support would add a fourth-order residual about five to fifteen times the RMS already present in those equations.
\end{shaded}

\begin{shaded}
This test varies the fourth-order moments around two approximate fields. We imposed the directions derived in \S\ref{sec:nullspace}. The measured R26--DSMC difference was not imposed. Simultaneous perturbations of several moments could cancel in the residual. The support excludes the walls, where $R_{nn}$ and $\Delta$ enter the boundary conditions (appendix~\ref{app:finite_wall_coupling}). The calculation therefore leaves open both the existence of another admissible solution and uniqueness for the discrete and continuum boundary-value problems.
\end{shaded}

\subsection{\texorpdfstring{\colorbox{white}{\parbox{0.75\linewidth}{Start-up of the cavity flow}}}{Start-up of the cavity flow}}\label{app:transient}

\begin{shaded}
To compute start-up we restored the time derivatives in the balance laws of each model, retaining the steady spatial operators, wall conditions and closures. In the 17-variable R13 system these terms are $\partial_t\rho$, $\rho\,\partial_t\uv$, $\tfrac32\rho\,\partial_t T$, $\partial_t\sigma_{ij}$ and $\partial_t q_i$. The eight equations for $m_{ijk}$, $R_{ij}$ and $\Delta$ remain algebraic. In R26 the bulk equations advance $\rho$, $\rho\uv$, $\rho T$, $\sigma_{ij}$, $q_i$, $m_{ijk}$, $R_{ij}$ and $\Delta$ per unit volume. The wall, extrapolation and corner equations and the global mass constraint remain algebraic.

We used second-order backward differences with a time step of $0.0645\,L/c_0=0.258$~ns. The first step used backward Euler, divided into substeps in R26. R13 used the principal $60\times60$ grid. The smooth-lid R26 grids were $70\times70$ (stretched) at $\Kn_{\rm Gu}=0.05$ and $56\times56$ at $0.20$. For R13 we evaluated the coefficients at the extrapolated field $2\boldsymbol{U}^n-\boldsymbol{U}^{n-1}$. Each linearly implicit step required one direct sparse factorisation. For R26 we used modified Newton iteration and the finite-difference Jacobian of the steady solver, with a tolerance of $10^{-6}$ on the largest scaled residual at each step. Halving the time step changed R13 density, velocity, temperature and heat flux by at most 3\% of their range at $t=2.6$~ns and less than 1\% after 5~ns. During the impulsive start the stresses differed by up to 20\% at 1~ns.

The gas starts at rest at the wall temperature. The lid starts moving at $t=0$. For R26, the sharp-lid impulse prevented the first Newton step. We therefore tested a finite rise time, increasing the lid speed as $U_w\sin^2(\upi t/2t_r)$ for $t<t_r$, with $t_r=1$~ns. We also tested a 4~ns rise time.

For DSMC we used the SPARTA set-up of \S\ref{sec:dsmc}, starting from rest at $t=0$. At each Knudsen number we averaged moments from eight independent runs over blocks of 500 time steps. Each run lasted $2\times10^5$ steps (206~ns). Figure~\ref{fig:transient_history} shows the rate of change of the moment fields. By 80~ns it had fallen six to seven orders of magnitude, to rates of order $10^{-7}$~ns$^{-1}$ or less, except for smooth-lid R13 at $\Kn_{\rm Gu}=0.20$. The sharp-lid R13 rates reached about $10^{-12}$~ns$^{-1}$ at 206~ns. The R13 density averaged over the lower half of the cavity oscillated as it decayed. Its period was about 8.6~ns at $\Kn_{\rm Gu}=0.20$ and 7.4~ns at $0.05$, somewhat longer than the acoustic round-trip time $2L/a_0=6.2$~ns with $a_0=(5/3)^{1/2}c_0$. The rate in figure~\ref{fig:transient_history} has two maxima per period.

Within one to two nanoseconds, both DSMC and R13 show cooling as the gas expands at the upstream lid corner and heating as it is compressed downstream. R26 shows the same pattern at $\Kn_{\rm Gu}=0.05$. In R13 the cold-to-hot area fraction approaches its steady value within about 15~ns at $\Kn_{\rm Gu}=0.20$ and 20~ns at $0.05$ (figure~\ref{fig:transient_history}, right). At $0.05$ the fraction is already about 0.2 after 3~ns, compared with a steady DSMC fraction of 0.09 (table~\ref{tab:q_parts}). At $\Kn_{\rm Gu}=0.20$, within $0.1L$ of the upstream wall and over $0.3\le y/L\le0.9$, the mean $T/T_w-1$ after 6~ns is $-0.015$ in R13 and $-0.004$ in DSMC. The corresponding steady values are $-0.018$ and $-0.005$. Thus the excess R13 cooling is already present in the first few nanoseconds.

For smooth-lid R13 at $\Kn_{\rm Gu}=0.20$, the rate of change fell to $1.4\times10^{-5}\,\mathrm{ns}^{-1}$ at $t\approx43$~ns. A disturbance then grew on the nodes next to the upstream lid corner, by a factor of about 4.5 every 8~ns. It was strongest in $\Delta$. By 82~ns it changed the local density and velocity by about $10^{-3}$. Table~\ref{tab:smooth_lid} gives the steady smooth-lid comparisons for the available fields.

The R26 start-up histories extend to the last time step satisfying the Newton tolerance, at lid speeds between 33 and 100~m\,s$^{-1}$ on the sharp-lid grids ($90^2$ at $\Kn_{\rm Gu}=0.20$ and $60^2$ at $0.05$) and on the smooth-lid $56^2$ grid at $0.20$. The tests used one-sided and centred finite-difference Jacobians, up to sixteen backward-Euler substeps and rise times $t_r=1$ and 4~ns. For the sharp lid at $\Kn_{\rm Gu}=0.20$, the impulsive start first encounters difficulty at the same lid corner that dominates the steady residual (appendix~\ref{app:residual_fields}). Both the spatial equations and the time discretisation enter this test.

The movies in table~\ref{tab:movies} show the steady fields, the individual heat-flux contributions in table~\ref{tab:q_parts} and start-up with sharp and smooth lid profiles. For the DSMC start-up movies we averaged the eight runs over $2\times2$ cells and up to thirteen successive blocks of 500 steps. The steady movies use the pooled DSMC fields.
\end{shaded}

\begin{figure}
\centering
\includegraphics[width=\linewidth]{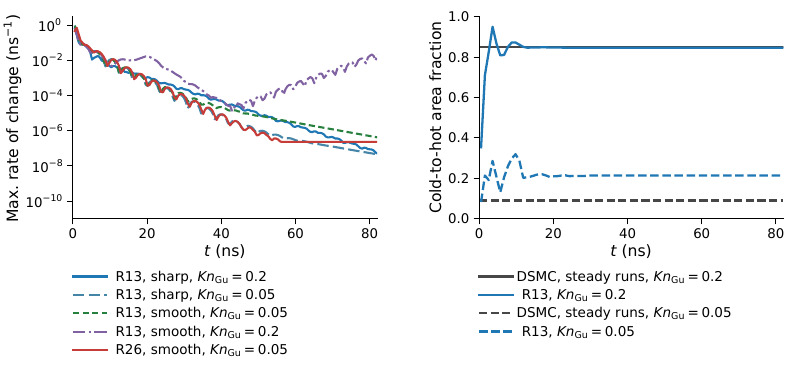}
\captionsetup{format=revisioncaption}
\caption{Cavity start-up from rest. The left panel shows the largest rate of change of $u/U_w$, $v/U_w$ and $T/T_w$ over interior nodes for sharp- and smooth-lid R13 and smooth-lid R26. The R26 rate levels off at the Newton tolerance of each time step. The right panel shows the sharp-lid R13 cold-to-hot area fraction as defined in table~\ref{tab:q_parts}. Horizontal lines give its steady DSMC values. At $\Kn_{\rm Gu}=0.20$ the R13 and DSMC lines coincide after 15~ns. Sampling fluctuations in the start-up DSMC runs are too large to resolve this sign-based quantity.}
\label{fig:transient_history}
\end{figure}

\begin{center}
\begin{minipage}{\linewidth}
\centering
\begin{tabular}{llll}
\toprule
movie & content & $\Kn_{\rm Gu}$ & YouTube \\
\midrule
1 & steady DSMC, R13 and R26 fields with tracers along $\qv$ & 0.20 & \href{https://youtu.be/j6WzMMRgqJA}{\texttt{youtu.be/j6WzMMRgqJA}} \\
2 & steady DSMC, R13 and R26 fields with tracers along $\qv$ & 0.05 & \href{https://youtu.be/uQVeE1rl8ZU}{\texttt{youtu.be/uQVeE1rl8ZU}} \\
3 & steady fields and the parts $\qv^{F}, \qv^{\sigma}, \qv^{R}, \qv^{\Delta}, \qv^{N}$ & 0.20 & \href{https://youtu.be/mod1PgTAkmQ}{\texttt{youtu.be/mod1PgTAkmQ}} \\
4 & steady fields and the parts $\qv^{F}, \qv^{\sigma}, \qv^{R}, \qv^{\Delta}, \qv^{N}$ & 0.05 & \href{https://youtu.be/UWSyRreRiIA}{\texttt{youtu.be/UWSyRreRiIA}} \\
5 & start-up, sharp lid: DSMC, time-dependent R13 & 0.20 & \href{https://youtu.be/RN72PGinCYc}{\texttt{youtu.be/RN72PGinCYc}} \\
6 & start-up, sharp lid: DSMC, time-dependent R13 & 0.05 & \href{https://youtu.be/1PRoDPomhF0}{\texttt{youtu.be/1PRoDPomhF0}} \\
7 & start-up, smooth lid: DSMC, time-dependent R13 and R26 & 0.05 & \href{https://youtu.be/gvaklvnKSdY}{\texttt{youtu.be/gvaklvnKSdY}} \\
8 & start-up, smooth lid: DSMC, time-dependent R13 (corner mode after 45~ns) & 0.20 & \href{https://youtu.be/DgI5p-Yowok}{\texttt{youtu.be/DgI5p-Yowok}} \\
\bottomrule
\end{tabular}
\captionsetup{format=revisioncaption}
\captionof{table}{Steady cavity fields and start-up movies in the supplementary material and on YouTube.}
\label{tab:movies}
\end{minipage}
\end{center}

\subsection{\texorpdfstring{\colorbox{white}{\parbox{0.75\linewidth}{Supplementary data and reproducibility}}}{Supplementary data and reproducibility}}\label{app:evidence_lineage}

\begin{shaded}
The two $90^2$ sharp-lid R26 fields at $\Kn_{\rm Gu}=0.20$ are identified by their use: one supplies the tensor in table~\ref{tab:canonical_table6} and the other supplies the residuals and starting field for refinement (appendix~\ref{app:complete_continuation}). The supplementary data include both fields.

For the $\Kn_{\rm Gu}=0.20$ closure comparison in table~\ref{tab:new_closure}, the supplementary code evaluates the R26 closures on velocity moments of degree zero through five for seven DSMC realisations and their pool. The supplied input and target moments recover the closure discrepancies of table~\ref{tab:new_closure} to a relative tolerance of $10^{-9}$ (script \texttt{reproduce\_table6.py} in \texttt{JFM-2026-2292\_table6\_data\_20261116.zip}). Table~\ref{tab:canonical_table6} is computed by \texttt{BOUNDED\_LOCAL\_CANONICAL\_CALCULATION.py} from the saved moments in \texttt{revision\_archive\_20261115/common\_estimator\_table7} of \texttt{JFM-2026-2292\_revision\_archive\_20261116.zip} (results in \texttt{CANONICAL\_RESULT.json}). The supplementary data also give the noise estimates and the corresponding results at $\Kn_{\rm Gu}=0.05$.

The available data are reduced moment and field arrays. For the active-cell population defined in \S\ref{sec:driver} they give $f_{AF|act}=0.17374$ and $P_\Delta/P_R=0.041921$. These values replace the conditional fraction 0.343 and ratio 0.063 of the reviewed version, which could not be regenerated from the archived arrays.
\end{shaded}

\section*{Acknowledgements}
The author thanks Dr Anirudh S. Rana for providing the original 17-field R13 cavity code, including the coefficient matrices and saddle-point formulation used here.

\section*{Funding}
This work received no specific grant from any funding agency, commercial or not-for-profit sector.

\section*{Declaration of interests}
The author reports no conflict of interest.

\section*{Data availability statement}
\begingroup\sloppy
The \href{https://github.com/Ehsan-Roohi/Fourier/tree/main/r13-r26-cavity}{public GitHub repository} contains the R13 source, two R26 implementations, numerical tests, scripts for generating and checking the SPARTA cases, common-grid analysis and plotting code, reduced DSMC sensitivity tables, and the mean and standard errors of the seven-realisation \(\Kn_{\rm Gu}=0.20\) ensemble. It also contains the fourth-order projection with the finite-particle correction and the digitised Yang velocity reference with its extraction and comparison scripts. The earlier R26 implementation of appendix~\ref{app:eq25_resolution} is in \texttt{r13-r26-cavity/\allowbreak code/\allowbreak r26}. The source used for the plotted fields is in \texttt{r13-r26-refinement-\allowbreak 20261104/\allowbreak solver/\allowbreak r26/\allowbreak source}.

The supplementary material provides the R26 bulk equations and tensor closures used for the plotted fields, the grid-sequencing parameters, and the sharp- and smooth-lid solutions compared in tables~\ref{tab:root_states} and \ref{tab:smooth_lid}. It also contains the DSMC closure inputs, noise estimates and the two $90^2$ R26 fields identified in appendix~\ref{app:evidence_lineage}.

The numerical states, processed arrays, sampling masks and scripts for the tensor, divergence and profile comparisons accompany the article. The supplementary material includes the movies listed in table~\ref{tab:movies}. The public repository also provides the time-dependent R13 and R26 solvers, the DSMC start-up inputs and the scripts for figure~\ref{fig:transient_history}.
\endgroup

\section*{Declaration of artificial-intelligence-assisted technologies}
\begingroup\sloppy
The author used ChatGPT (OpenAI) for language editing and code refactoring in May--August 2026, and for language editing in November 2026. The author reviewed the text and code and takes full responsibility for the manuscript.
\endgroup

\end{document}

%% file: updated_field_metrics.tex
\begin{center}
\begin{minipage}{\linewidth}
\centering
\begin{tabular}{lllrrrr}
\toprule
$\Kn$ & model & grid & $E_{\rho-1}$ & $E_{T-1}$ & $E_u$ & $E_q$\\
\midrule
0.05 & R13 & $60^2$ & 0.0576 & 0.1679 & 0.0364 & 0.6415 \\
\rowcolor{yellow!40}0.05 & R13 & $180^2$ & 0.0623 & 0.2134 & 0.0261 & 0.3016 \\
0.05 & R26 & $60^2$ & 0.0200 & 0.0379 & 0.0149 & 0.1316 \\
0.2 & R13 & $60^2$ & 0.1791 & 0.6970 & 0.1374 & 0.6053 \\
0.2 & R26 & $90^2$ & 0.0601 & 0.2460 & 0.0584 & 0.3644 \\
\bottomrule
\end{tabular}
\captionof{table}{Relative $L^2$ differences from DSMC before filtering. We interpolate within the native R13 domain onto common DSMC cell centres, omitting squares of side $0.05L$ at the lid corners. Density and temperature are measured relative to equilibrium. At $\Kn_{\rm Gu}=0.05$, refining R13 from $60^2$ to $180^2$ reduces the heat-flux difference from 0.6415 to 0.3016; the finer calculation has residual RMS $2.0\times10^{-5}$. For R26, two approximations on the same grid give differences of 3--26\% in these entries at $\Kn_{\rm Gu}=0.05$ and 0.4--5.6\% at $0.20$. The $60^2$ R26 grid is stretched towards the walls. Table~\ref{tab:root_states} gives the corresponding comparisons on the finest grids meeting the full residual criterion ($38^2$ and $24^2$, respectively).}
\label{tab:current_fields}
\end{minipage}
\end{center}

%% file: numerical_algorithms.tex
\section{Moment equations, discretization and nonlinear solution}
\label{app:actual_algorithm}
\subsection{State, tensor convention and conservation laws}
The independent planar variables are
\begin{equation}
U=(\rho,u_x,u_y,\theta,q_x,q_y,\sigma_{xx},\sigma_{xy},\sigma_{yy},R_{xx},R_{xy},R_{yy},m_{xxx},m_{xxy},m_{xyy},m_{yyy},\Delta).
\end{equation}
Here $\theta=k_BT/m$, $p=\rho\theta$, and molecular velocities remain three-dimensional. Repeated indices run over $x,y,z$, even though $\partial_z=0$ and reflection symmetry eliminates components odd in $z$. In particular,
\begin{equation}
\sigma_{zz}=-\sigma_{xx}-\sigma_{yy},\quad R_{zz}=-R_{xx}-R_{yy},\quad
m_{xzz}=-m_{xxx}-m_{xyy},\quad m_{yzz}=-m_{xxy}-m_{yyy}.
\end{equation}
Angular brackets denote the symmetric trace-free (STF) projection in three dimensions, not a two-dimensional trace subtraction. For example,
$A_{\langle ij\rangle}=(A_{ij}+A_{ji})/2-\delta_{ij}A_{kk}/3$.
With $C_i=c_i-u_i$ and a mass distribution $f$, the conventions are
\begin{align}
\sigma_{ij}&=\int C_{\langle i}C_{j\rangle}f\,\dd^3c,&
q_i&=\tfrac12\int C^2 C_i f\,\dd^3c,\\
m_{ijk}&=\int C_{\langle i}C_jC_{k\rangle}f\,\dd^3c,&
R_{ij}&=\int C^2 C_{\langle i}C_{j\rangle}f\,\dd^3c-7\theta\sigma_{ij},\\
\Delta&=\int C^4 f\,\dd^3c-15\rho\theta^2.
\end{align}
Both descriptions retain mass, momentum and energy conservation. In the steady, unforced cavity their conservative form is
\begin{align}
\partial_j(\rho u_j)&=0,\\
\partial_j(\rho u_i u_j+p\delta_{ij}+\sigma_{ij})&=0,\\
\partial_j(\rho\theta u_j)+\tfrac23\partial_jq_j+
\tfrac23(p\partial_ju_j+\sigma_{ij}\partial_ju_i)&=0.
\end{align}
We use this conservative temperature equation in the R26 residual. Its material form gives a different discrete residual until the discrete mass equation is satisfied.

\subsection{Closures for the higher moments}
R13 transports density, velocity, temperature, stress and heat flux. In the planar R13 solver the additional $m,R,\Delta$ entries enforce constitutive relations. The resulting 17 entries include these algebraic unknowns. The regularization follows \citet{RanaTorrilhonStruchtrup2013} and the coefficient implementation accompanying that formulation. The leading Maxwell terms are
\begin{equation}
m_{ijk}=-2\frac\mu\rho\partial_{\langle k}\sigma_{ij\rangle},\quad
R_{ij}=-\frac{24}{5}\frac\mu\rho\partial_{\langle i}q_{j\rangle},\quad
\Delta=-12\frac\mu\rho\partial_kq_k.
\end{equation}
The computations use the full nonlinear coefficients, including the terms given below.

The R26 bulk operator follows \citet{GuEmerson2009}. Writing $\mathcal D a=\partial_k(u_k a)$ componentwise, its additional balance equations are
\begin{align}
\mathcal D\sigma_{ij}+\partial_km_{ijk}+\frac p\mu\sigma_{ij}
 +2p\partial_{\langle i}u_{j\rangle}&=\Sigma_{ij},\\
\mathcal Dq_i+\tfrac12\partial_jR_{ij}+\frac{2p}{3\mu}q_i
 +\tfrac52p\partial_i\theta&=Q_i,\\
\mathcal Dm_{ijk}+\partial_l\Phi_{ijkl}+\frac{3p}{2\mu}m_{ijk}
 +3p\partial_{\langle k}(\sigma_{ij\rangle}/\rho)&=M_{ijk},\\
\mathcal DR_{ij}+\partial_k\psi_{ijk}+\frac{7p}{6\mu}R_{ij}
 +\frac{28p}{5}\partial_{\langle j}(q_{i\rangle}/\rho)&=S_{ij},\\
\mathcal D\Delta+\partial_i\Omega_i+\frac{2p}{3\mu}\Delta
 +8p\partial_i(q_i/\rho)&=N.
\end{align}
The source terms are the nonlinear terms of equations (8), (10), and (19)--(21) of that reference, and the closures $\Phi,\psi,\Omega$ are its equations (23), (25) and (26), with the implementation's explicitly recorded convention. In particular,
\begin{align}
\Sigma_{ij}&=-\tfrac45\partial_{\langle i}q_{j\rangle}
 -2\sigma_{k\langle i}\partial_k u_{j\rangle},\\
Q_i&=-\tfrac72\sigma_{ik}\partial_k\theta-\theta\partial_j\sigma_{ij}
 +\frac{\sigma_{ij}}\rho(\partial_jp+\partial_k\sigma_{jk})\nonumber\\
&\quad-\tfrac25\bigl(\tfrac72q_k\partial_ku_i+q_k\partial_iu_k+q_i\partial_ku_k\bigr)
 -\tfrac16\partial_i\Delta-m_{ijk}\partial_ku_j.
\end{align}
The last two terms in $Q_i$ lie outside the $-2/5$ factor. We use the quotient rule for gradients of density-normalised quantities and set $C_1=2.097$, $C_2=0.291$, $Y_1=1.698$, $Y_2=1.203$ and $Y_3=0.854$. The supplementary implementation includes all tensor components of $M,S,N$.

\input{r26_expanded_terms.tex}
\input{r13_explicit_production.tex}
\subsection{Discrete residual and solution procedure}
For R13, the interior residual is $F(U)=A(U)U-b(U)$. A global constraint fixes the mass, and adjacent cells share the same face flux in the continuity equation. The wall treatment uses two-point extrapolation and includes tangential stress in the effective pressure. Newton--Krylov iteration uses an Armijo constant of $10^{-4}$, halves the step during backtracking and permits a minimum step of $2^{-10}$. We reassemble the unbordered $F(U)$ at the final field to measure the equation residual separately from the mass constraint.

The R26 residual contains the bulk equations, wall conditions, extrapolation and corner equations and the mass constraint. We solve for logarithmic density and temperature and for the signed values of the remaining moments. The full tensors are reconstructed before evaluating the closures. For the $\Kn=0.05$ grid sequence, the Levenberg--Marquardt step is chosen by the decrease in the scaled least-squares objective (appendix~\ref{app:solver_reconciliation}).

\begin{shaded}
\noindent\textbf{R13 iteration.} We start from the planar cavity formulation and MATLAB code supplied by A.~S.~Rana \citep{RanaTorrilhonStruchtrup2013}. The coefficient matrices and the saddle-point structure are common to both solvers (table~\ref{tab:r13_solver_comparison}). In the supplied Picard iteration, each full step solves a frozen-coefficient system with a Lagrange multiplier for the total mass. The iteration ends when the largest change falls below $10^{-4}$. With nodal continuity equations the multiplier can remain finite. In our formulation the face fluxes make the continuity and mass constraints compatible. We solve the bordered residual equations by Newton's method. GMRES uses finite-difference Jacobian--vector products, adaptive forcing and a sparse LU right preconditioner built from the frozen coefficients.

At $\Kn_{\rm Gu}=0.05$ five Newton steps reduce the RMS residual from about $0.2$ to between $10^{-14}$ and $1.1\times10^{-9}$ on grids from $40^2$ to $80^2$ (figure~\ref{fig:r13_convergence}). At $\Kn_{\rm Gu}=0.20$ the $60^2$ grid reaches $6.7\times10^{-13}$ in eight steps. For comparison, the four steps shown for $80^2$ end at $7.6\times10^{-2}$, with a later value of $1.45\times10^{-3}$; the field comparison at this Knudsen number uses $60^2$. The refinement study and tables~\ref{tab:residual_measures} and \ref{tab:equation_rms} use different residual scalings. At $\Kn_{\rm Gu}=0.05$ the two $60^2$ R13 fields agree to a relative $7.5\times10^{-9}$.

\noindent\textbf{R26 iteration.} In the main iteration we assemble the sparse Jacobian of the full residual with an exact row for the mass constraint. On $120^2$ we solve the linear system by column-scaled LSMR with a limit on the iteration count. A step must reduce both the scaled residual and the raw residual with the omitted continuity equation restored, and meet the physical admissibility conditions. For continuation we use trust-region Gauss--Newton steps with a lagged Jacobian and inexact Krylov solves. When refining the grid we interpolate logarithmic density and temperature, then normalise the density to the total mass. A second calculation uses Levenberg--Marquardt iteration and interpolates the solver variables without rescaling (appendix~\ref{app:solver_reconciliation}).

For a sharp lid, every grid tested from $10^2$ to $38^2$ at $\Kn_{\rm Gu}=0.05$, and from $10^2$ to $24^2$ at $0.20$, meets the maximum scaled-residual criterion of $10^{-9}$ (table~\ref{tab:root_states}). On finer grids the Levenberg--Marquardt sequences yield RMS residuals of $3\times10^{-7}$ to $6\times10^{-5}$ ($40^2$--$60^2$ and $26^2$--$90^2$, respectively). We use the $60^2$ field at $\Kn_{\rm Gu}=0.05$ in the main comparison and the $90^2$ result at $0.20$ as a comparison with continuation. In these two fields the $4\times4$ blocks at the lid corners account for 93\% and 81\% of the squared residual. The largest residual in each of the three fine-grid fields examined in appendix~\ref{app:eq25_resolution} belongs to an extrapolation equation within two nodes of a lid corner. These are the points at which the imposed wall velocity jumps. For the published-case comparisons of appendix~\ref{app:yang_validation} we also report the change of the centreline profiles between the last two grids; changes below 2\% are taken as grid convergence.

To test the role of this jump, we multiplied the lid speed by $\sin^2(\pi d/2\delta)$ within a distance $d<\delta$ of each corner and allowed more iterations per grid. At $\Kn_{\rm Gu}=0.05$ and $\delta=0.05L$, seven of nine grids from $30^2$ to $90^2$ meet the root criterion, including $60^2$ and $80^2$. On $60^2$ the criterion is also met at $\delta=0.035L$; the narrower widths $0.025L$ and $0.0125L$ lie outside that range. At $\Kn_{\rm Gu}=0.20$ the criterion is met on $24^2$. On six finer grids from $28^2$ to $66^2$ the final RMS ranges from $1.6\times10^{-7}$ to $1.7\times10^{-5}$, below the sharp-lid values on the nearest grids. In several cases a small relative least-squares gradient stopped the iteration. Smoothing the corners therefore helps substantially at the lower Knudsen number. At the higher Knudsen number the effects of the boundary and the stopping rule remain coupled. These calculations change the boundary condition; table~\ref{tab:smooth_lid} gives a separate comparison with DSMC for a matched Gaussian lid.

The mixed structure of steady linear R13 requires compatible finite elements or stabilisation \citep{WesterkampTorrilhon2019,TheisenTorrilhon2021}, and a closed cavity requires a constraint on the mean density \citep{RanaTorrilhonStruchtrup2013}. Excess wall conditions can produce boundary layers \citep{TorrilhonStruchtrup2008} and oscillations \citep{RanaTorrilhonStruchtrup2013}. For linearised systems the stability also depends on the wall formulation \citep{SarnaTorrilhon2018,TorrilhonSarna2017,RanaStruchtrup2016}; entropy-compatible boundary conditions are available for linear R26 \citep{RanaGuptaSprittlesTorrilhon2021}. The finite-element verification of \citet{TheisenTorrilhon2021} used a domain without sharp corners. Here we use a collocated discretisation without additional stabilisation. R26 combines the Maxwell-type wall relations of \citet{GuEmerson2009} with six extrapolation equations at each smooth-wall node (\S\ref{sec:boundary_scope}). The R13 mass multiplier vanishes at the solution. In R26 the mass residual is below $10^{-10}$ for the $90^2$ continuation field with an exact mass row; appendix~\ref{app:eq25_resolution} gives the corresponding values for the least-squares fields. Figure~\ref{fig:convergence_earliest_new} and appendix~\ref{app:complete_continuation} give the iteration histories.
\end{shaded}

\begin{table}
\centering\small
\rowcolors{1}{yellow!40}{yellow!40}
\begin{tabular}{>{\raggedright\arraybackslash}p{0.2\linewidth}>{\raggedright\arraybackslash}p{0.33\linewidth}>{\raggedright\arraybackslash}p{0.39\linewidth}}
\toprule
 & Supplied code & Present R13\\
\midrule
Continuity & nodal rows of the coefficient matrices & conservative shared-face finite-volume balance\\
Total mass & Lagrange-multiplier border with nodal quadrature weights; multiplier non-zero at the end & compatible global constraint; multiplier zero at the accepted state\\
Nonlinear iteration & frozen-coefficient (Picard) iteration with the full update & Newton iteration on the exact bordered residual (defect-Newton linearisation)\\
Linear solves & QMR (relative tolerance $10^{-4}$) or direct & GMRES with finite-difference Jacobian--vector products and a sparse LU right preconditioner\\
Globalisation & none & Armijo backtracking\\
Stopping test & largest update below $10^{-4}$ & residual, mass, multiplier and global-balance tolerances; independent reassembly\\
Wall extrapolation & three-point one-sided rows & two-point linear extrapolation, equations (15)--(20) of \citet{RanaTorrilhonStruchtrup2013}\\
Effective wall pressure & normal stress & tangential stress, their equation (7)\\
Production & coded matrix scaled by $\rho\sqrt\theta/\Kn_{\rm Rana}$, i.e.\ $\mu\propto\sqrt\theta$ & reduced matrix of their appendix A scaled by $\rho/\Kn_{\rm Rana}$, i.e.\ Maxwell molecules with $\mu\propto\theta$\\
\bottomrule
\end{tabular}
\captionsetup{format=revisioncaption}
\caption{Nonlinear solvers for R13 in the cavity code supplied with \citet{RanaTorrilhonStruchtrup2013} and in the present calculation. Both use the same coefficient matrices and saddle-point structure.}
\label{tab:r13_solver_comparison}
\end{table}

\begin{figure}[tbp]
\centering
\includegraphics[width=\linewidth]{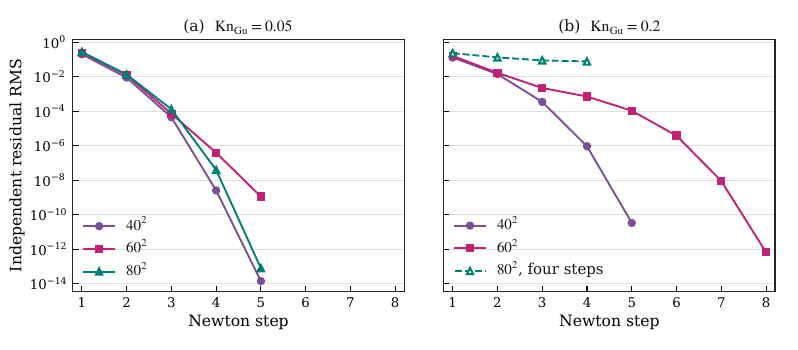}
\captionsetup{format=revisioncaption}
\caption{R13 equation-residual RMS after each Newton step on $40^2$, $60^2$ and $80^2$ grids at (a) $\Kn_{\rm Gu}=0.05$ and (b) $\Kn_{\rm Gu}=0.20$. Open symbols in (b) show the four available $80^2$ steps. Field comparisons at this Knudsen number use the $60^2$ solution. All curves use the scaling of the refinement study.}
\label{fig:r13_convergence}
\end{figure}

For continuation we re-evaluate the full residual, mass balance and wall conditions after each trial step.

For equation $k$ the reported interior value is
\begin{equation}
\mathrm{RMS}_k=\left(\frac{1}{|I_k|}\sum_{a\in I_k}r_{ka}^2\right)^{1/2}.
\end{equation}
The set $I_k$ contains the interior nodes for equation $k$. It excludes the boundary and the continuity equation replaced by the total-mass constraint. Whole-vector norms include the boundary and mass equations and use the row scaling specified for each model.

\input{solver_residuals.tex}

\begin{figure}[tbp]
\centering
\includegraphics[width=\linewidth]{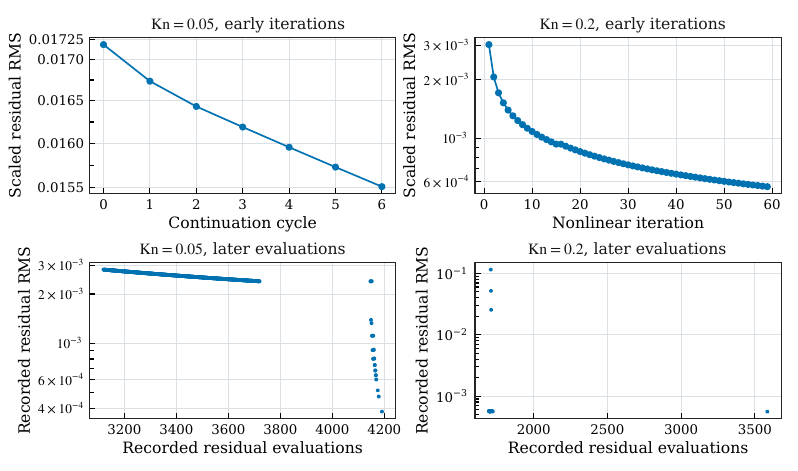}
\caption{R26 residual RMS during nonlinear iteration. The $\Kn=0.05$ curves show the earlier closure implementation described in appendix~\ref{app:eq25_resolution}; the main field uses the grid sequence of appendix~\ref{app:solver_reconciliation}. The upper panels start at the first measured RMS. The lower panels show the recorded residual evaluations, including unsuccessful trials. Gaps indicate intervals without measurements. The horizontal axes count nonlinear steps or residual evaluations, as indicated; each calculation uses its own residual scaling.}
\label{fig:convergence_earliest_new}
\end{figure}

%% file: r26_expanded_terms.tex
\subsection{Higher-order terms in the R26 equations}
Both R26 fields use the three- and four-index projections given below. Appendix~\ref{app:eq25_resolution} compares two of these contractions with an earlier implementation at $\Kn=0.05$. Write $d=\partial_ku_k$, $b_i=\partial_k\sigma_{ik}$, $a_i=\partial_i p+b_i$ and $h=\partial_kq_k+\sigma_{ij}\partial_ju_i$. The operator $\mathcal P_n$ denotes the three-dimensional symmetric trace-free projection over the $n$ free indices. Contracted indices lie outside that projection:
\begin{align}
M_{ijk}={}&\mathcal P_3\left[\frac3\rho\sigma_{ij}b_k-\frac{12}{5}q_i\partial_ku_j-3m_{lij}\partial_lu_k-\frac37\partial_kR_{ij}\right],\\
S_{ij}={}&-\frac{2p}{3\mu\rho}\mathcal P_2[\sigma_{ki}\sigma_{jk}]
 -\frac{28}{5}\mathcal P_2[q_i\partial_j\theta]+\frac{28}{5\rho}\mathcal P_2[q_i b_j]
 +\frac{14}{3\rho}\sigma_{ij}h\nonumber\\
&-4\theta\mathcal P_2[\sigma_{ki}(\partial_j u_k+\partial_k u_j)]
 +\frac83\theta\sigma_{ij}d-2\theta\partial_km_{ijk}-9m_{ijk}\partial_k\theta
 -2\Phi_{ijkl}\partial_lu_k\nonumber\\
&+\frac2\rho m_{ijk}a_k
 -\mathcal P_2[\tfrac67\mathcal R_{ij}+\tfrac45R_{ki}\partial_j u_k+2R_{ki}\partial_ku_j]
 -\tfrac{14}{15}\Delta\mathcal P_2[\partial_i u_j]
 -\tfrac25\mathcal P_2[\partial_i\Omega_j],\\
N={}&-\frac{2p}{3\mu\rho}\sigma_{ij}\sigma_{ij}
 -4(2\theta\sigma_{ij}+R_{ij})\partial_ju_i
 +\frac8\rho q_i b_i-20q_i\partial_i\theta-\tfrac43\Delta d.
\end{align}
The regularizing fluxes are
\begin{align}
\Phi_{ijkl}={}&-\frac{4\mu}{C_1}\mathcal P_4[\partial_l(m_{ijk}/\rho)]
 -\frac{12\mu}{C_1\rho}\mathcal P_4[\sigma_{ij}\partial_lu_k]
 +\frac{4\mu}{C_1p\rho}\mathcal P_4[m_{ijk}b_l]\nonumber\\
&-\frac{12\mu}{7C_1p}\mathcal P_4[R_{ij}\partial_lu_k]
 -\frac{C_2}{C_1\rho}\mathcal P_4[\sigma_{ij}\sigma_{kl}],\\
\psi_{ijk}={}&-\frac{27\mu}{7Y_1}\mathcal P_3[\partial_k(R_{ij}/\rho)]
 -\frac{27\mu}{7Y_1p}\mathcal P_3[(R_{ij}+7\theta\sigma_{ij})\partial_k\theta]
 -\frac{108\mu}{5Y_1\rho}\mathcal P_3[q_i\partial_ku_j]\nonumber\\
&+\frac{27\mu}{7Y_1p\rho}\mathcal P_3[R_{ij}b_k]+\frac{6\mu}{Y_1p\rho}m_{ijk}h
 -\frac\mu{Y_1\rho}\mathcal P_3[\tfrac{54}{7}m_{lij}\partial_ku_l+8\mathcal T_{ijk}-6m_{ijk}d]\nonumber\\
&-\frac{Y_2}{Y_1\rho}\mathcal P_3[\sigma_{li}m_{jkl}]
 -\frac{Y_3}{Y_1\rho}\mathcal P_3[q_i\sigma_{jk}],\\
\Omega_i={}&-\tfrac73\mu\partial_i(\Delta/\rho)-4\mu\partial_j(R_{ij}/\rho)
 -\frac{56\mu}{5\rho}q_j(\partial_ju_i+\partial_i u_j)
 -\frac{8\mu}{\rho}m_{ijk}\partial_ku_j\nonumber\\
&-\frac{14\mu}{p}(2\theta\sigma_{ij}+R_{ij})\partial_j\theta
 +\frac{56\mu}{3p\rho}q_i h+\frac{4\mu}{p\rho}R_{ij}b_j
 +\frac{7\mu\Delta}{3p}(b_i/\rho-2\partial_i\theta)\nonumber\\
&-\frac{2}{15\rho}(14q_j\sigma_{ij}+5m_{ijk}\sigma_{jk}).
\end{align}
Here $\mathcal T_{ijk}=\sum_s(\mathcal P_4 A^{(s)})_{ijks}$ with $A^{(s)}_{ijkl}=m_{ijk}\partial_s u_l$. The derivative index $s$ is held fixed during the four-index projection and contracted afterwards. In general $\mathcal T_{ijk}\ne m_{ijk}\partial_lu_l$. This bracket appears in equation (25) and is repeated in (A2) and (A9b) of \citet{GuEmerson2009}. Likewise $\mathcal R_{ij}=\sum_s(\mathcal P_3 B^{(s)})_{ijs}$ with $B^{(s)}_{ijk}=R_{ij}\partial_s u_k$, the three-index bracket of their equation (20); in general $\mathcal R_{ij}\ne R_{ij}\partial_ku_k$.

%% file: r13_explicit_production.tex
\subsection{Explicit R13 production and normalization}
\label{app:r13_production_explicit}
The quasilinear R13 bulk equations take the form
\begin{equation}
 A_x(U)\partial_xU+A_y(U)\partial_yU+
 \frac{\rho}{\Kn_R}\,\mathcal C(U)=0,
 \qquad \Kn_R=\sqrt{2/\pi}\,\Kn.
\end{equation}
Here $\mathcal C(U)=P(U)U$ is the nonlinear production vector. The discretisation of continuity and the wall conditions are given below. With Maxwell viscosity, $\mu/\mu_0=\theta/\theta_0$, the dimensionless collision frequency is $\rho/\Kn_R$. Equation (11) and appendix A of \citet{RanaTorrilhonStruchtrup2013} omit the density prefactor in the printed matrix form. The supplementary code includes both conventions; the present calculations use the Maxwell form above.

Set $a=q_x$, $b=q_y$, $c=\sigma_{xx}$, $d=\sigma_{xy}$ and $e=\sigma_{yy}$. The four conservation components of $\mathcal C$ vanish, and
\begin{equation}
 (\mathcal C_{q_x},\mathcal C_{q_y},\mathcal C_{\sigma_{xx}},
 \mathcal C_{\sigma_{xy}},\mathcal C_{\sigma_{yy}})
 =(2a/3,2b/3,c,d,e).
\end{equation}
The remaining components of the implemented reduced vector are
\begin{align}
\mathcal C_{R_{xx}}&=\frac5{24}R_{xx}
 +\frac{-16a^2+8b^2}{45\rho\theta}
 +\frac{25(-c^2+2ce+2e^2-d^2)}{126\rho},\\
\mathcal C_{R_{xy}}&=\frac5{24}R_{xy}
 -\frac{8ab}{15\rho\theta}-\frac{25d(c+e)}{42\rho},\\
\mathcal C_{R_{yy}}&=\frac5{24}R_{yy}
 +\frac{8a^2-16b^2}{45\rho\theta}
 +\frac{25(2c^2+2ce-e^2-d^2)}{126\rho},\\
\mathcal C_{m_{xxx}}&=\frac12m_{xxx}
 +\frac{-2ca/5+4db/15}{\rho\theta},\\
\mathcal C_{m_{xxy}}&=\frac12m_{xxy}
 +\frac{-16da/45-2cb/9+4eb/45}{\rho\theta},\\
\mathcal C_{m_{xyy}}&=\frac12m_{xyy}
 +\frac{4ca/45-2ea/9-16db/45}{\rho\theta},\\
\mathcal C_{m_{yyy}}&=\frac12m_{yyy}
 +\frac{-2eb/5+4da/15}{\rho\theta},\\
\mathcal C_{\Delta}&=\frac23\Delta
 -\frac{112(a^2+b^2)}{15\rho\theta}
 -\frac{20(c^2+ce+e^2+d^2)}{3\rho}.
\end{align}
The third-order components appear in the order $(m_{xxx},m_{xxy},m_{xyy},m_{yyy})$ in both the flux and wall matrices. The supplementary material gives the full $17\times17$ transport matrices.

We evaluated the eight expressions independently at all 3600 nodes of each grid. Their relative differences from the matrix evaluation are $7.32\times10^{-17}$ and $9.88\times10^{-17}$ at $\Kn=0.05$ and $0.2$, respectively. Interchanging the coordinate axes gives the corresponding interchange of the transport and production terms at every node.

%% file: solver_residuals.tex
\subsection{Nonlinear iteration and stopping criteria}
\label{app:solver_reconciliation}
For R13 we solve for the field $U$ and the mass multiplier $\lambda$ together. The relative residual of the bordered system is $\|A(U)(U,\lambda)-b(U)\|_2/\max(\|b(U)\|_2,\mathrm{tiny})$, where $\mathrm{tiny}$ prevents division by zero. We also evaluate the residual of the physical equations after removing the source proportional to $\lambda$. GMRES uses a direct factorisation of the frozen-coefficient matrix as a right preconditioner.

Let $z$ denote the bordered unknown, $S$ the variable scaling and $W$ the residual weights. For a direction $v$, we approximate $WJSv$ by the central difference
\begin{equation}
\frac{W[F(z+hSv)-F(z-hSv)]}{2h},\qquad
h=6.0554544523933395\times10^{-6}\frac{1+\|S^{-1}z\|_2}{\|v\|_2}.
\end{equation}
We halve $h$ up to ten times if either perturbation leaves the physical domain. GMRES restarts after 40 vectors and allows up to 12 cycles. The forcing parameter starts at $10^{-2}$ and lies between $10^{-6}$ and $10^{-1}$. Armijo backtracking uses a constant of $10^{-4}$ and halves the step down to $2^{-10}$. We allow up to 50 outer iterations at $\Kn=0.05$ and 30 at $\Kn=0.20$.

The relative tolerances are $10^{-10}$ for the bordered system and $10^{-8}$ for the physical equations. The absolute mass error and local continuity error must each be at most $10^{-8}$. The multiplier, global balance errors and source from the mass constraint must each be at most $10^{-10}$. A separate fixed-point test has tolerance $10^{-6}$. We also require positive density and temperature and an effective wall pressure above $10^{-12}$. Both R13 fields satisfy these criteria. We checked the residuals and face fluxes by reassembling the equations at each final field.

For R26 at $\Kn=0.05$, we compare grid sequencing with direct initialisation on the final $60^2$ grid. Both calculations start from a field obtained with the earlier closure described in appendix~\ref{app:eq25_resolution}. Its last correction used a lagged Jacobian in a Krylov subspace tangent to the mass constraint. The trust region used the Euclidean norm of the solver variables and fixed residual weights. The Krylov basis had at most 63 vectors. The construction radius was $0.999$ times the acceptance radius of $0.016430492489892205$, and the step had norm $2.03021\times10^{-4}$. A uniform shift in log density set the weighted mean density. A local change in velocity corrected the continuity equation at the node used for the mass constraint.

We then evaluated the linear model and nonlinear residual at the corrected trial field. The step had to be finite, satisfy the construction radius and the radius in the Krylov representation, obey the lower and upper bounds, differ from the starting field and decrease the linearised least-squares objective at full step. For the unit step taken here, both nonlinear decrease conditions reduce to $\phi_{\rm trial}<(1-10^{-4})\phi_0$, where $\phi=\|F_{\rm scaled}\|_2^2/2$. This correction required two residual evaluations and reused the Jacobian and its LU factorisation. Only some of the convergence tolerances are known for this starting field.

For the grid sequence we interpolated the starting field to $10^2$ and used the same R26 equations as at $\Kn=0.20$. On refinement we bilinearly interpolated the solver variables from the preceding grid. At each Levenberg--Marquardt step we solved
\[
 (J^{\rm T}J+\lambda_{\rm LM}\,{\rm diag}(J^{\rm T}J))\,\delta
     =-J^{\rm T}F_{\rm scaled}.
\]
We recomputed $J$ at every step by finite differences in 425 colour groups, using forward differences up to $28^2$ and central differences from $30^2$ onwards. Conjugate gradients used a sparse LU preconditioner based on a recent Jacobian, renewed at least every third step, and a limit of 200 iterations. We accepted a trial if it reduced $\phi$. For an accepted step, we divided $\lambda_{\rm LM}$ by ten if the ratio of actual to predicted reduction exceeded 0.75 and doubled it if the ratio fell below 0.25. After a rejected trial we multiplied it by ten.

The root criterion was a maximum absolute scaled residual below $10^{-9}$. In three successive calculations we allowed 60, 80 and 40 steps per grid and started $\lambda_{\rm LM}$ at $10^{-3}$, $10^{-3}$ and $10^{-6}$, respectively. On $40^2$ the first two calculations ended at scaled RMS residuals of $5.2\times10^{-6}$ and $1.1\times10^{-6}$. The second calculation also allowed termination for an RMS decrease below 0.1\% over ten steps. The third allowed termination for an RMS decrease below 5\% over six steps or a relative gradient below $10^{-6}$. Only the six-step condition affected the results among these additional tests.

The grids $10^2$, $14^2$, $20^2$, $28^2$, $30^2$, $32^2$, $34^2$, $36^2$ and $38^2$ met the root criterion. The final scaled RMS residuals on $40^2$, $44^2$, $48^2$, $52^2$, $56^2$ and $60^2$ were $2.6\times10^{-7}$, $9.9\times10^{-6}$, $1.5\times10^{-5}$, $1.4\times10^{-5}$, $1.2\times10^{-5}$ and $1.38\times10^{-5}$, respectively. None of these finer grids met the criterion. The calculation on $40^2$ reached the step limit. On $60^2$ the interpolated field had an RMS residual of $3.3\times10^{-2}$. Of 27 trials, 15 were accepted, six at the conjugate-gradient limit. The last accepted step reduced the RMS by 0.35\%. The six-step stopping condition then applied, with a largest scaled residual of $1.73\times10^{-3}$. The acceptance test for this sequence used the scaled objective alone. Appendix~\ref{app:eq25_resolution} gives the omitted continuity residual, which enters the additional raw-residual test used in continuation.

For the direct initialisation we kept the earlier field on its $60^2$ grid. In the first five Levenberg--Marquardt steps we formed the Jacobian by forward differences at steps one and five and reused it between them. The damping parameter started at $10^{-3}$ and decreased by a factor of three when the reduction ratio exceeded 0.75. These steps reduced the scaled RMS from $1.591\times10^{-2}$ to $1.454\times10^{-4}$. We then resumed with the iteration used for grid sequencing and initial damping $10^{-4}$. Of 44 trials, 24 were accepted, each at the conjugate-gradient limit. The time limit ended the iteration at a scaled RMS of $6.24\times10^{-5}$ and a maximum scaled residual of $7.68\times10^{-3}$. The last step reduced the RMS by about 0.3\%.

\subsection{Residuals of the cavity fields}
\label{app:residual_measures}
Table~\ref{tab:residual_measures} gives the RMS over the full residual vector for each field shown in the figures. For R13 we evaluated the unbordered equations. For R26 we applied the same row weights as in the nonlinear iteration. At $\Kn=0.20$ the main R26 comparison uses continuation on $90^2$. Figure~\ref{fig:kn020_gridseq} compares that field with a grid-sequenced field on the same $90^2$ grid. Continuation requires nonlinear decrease, conservation and positivity. Both $90^2$ fields retain finite residuals, with different contributions from the wall equations and the mass constraint (appendix~\ref{app:eq25_resolution}).
\begin{center}
\begin{minipage}{\linewidth}
\centering\small
\begin{tabular}{lllr}
\toprule
Model & $\Kn$ & Field & Full-vector RMS\\
\midrule
R13 & 0.05 & $60^2$, Newton--Krylov & $3.29508\times10^{-12}$\\
R13 & 0.20 & $60^2$, Newton--Krylov & $3.51546\times10^{-10}$\\
R26 & 0.05 & $60^2$, grid sequencing & $1.38121\times10^{-5}$\\
R26 & 0.20 & $90^2$, continuation & $5.59055\times10^{-4}$\\
R26 & 0.20 & $90^2$, grid sequencing & $6.06040\times10^{-5}$\\
\bottomrule
\end{tabular}
\captionof{table}{RMS residual $[\sum_{a=1}^{N_r}\widetilde r_a^2/N_r]^{1/2}$ for $N_r$ equations. Here $\widetilde r_a$ is the canonical unbordered residual for R13 and the solver-scaled residual for R26. The row normalisations differ between models.}
\label{tab:residual_measures}
\end{minipage}
\end{center}

\subsection{Interior residuals by equation}
Figures~\ref{fig:r26_residual_components}, \ref{fig:resid_r13_kn005}, \ref{fig:resid_r13_kn020} and \ref{fig:resid_r26_kn005} in appendix~\ref{app:residual_fields} show where the interior residuals are largest. For equation $k$ we use the set $I_k$ of interior rows on the original grid, excluding the continuity equation replaced by the mass constraint. We compute the norm from the signed residuals before plotting,
\begin{equation}
 R_k=\sqrt{|I_k|^{-1}\sum_{a\in I_k}\widetilde r_{ka}^{\,2}}.
\end{equation}
Table~\ref{tab:equation_rms} gives the 68 values for 17 equations, two models and two Knudsen numbers. The machine-readable supplement gives the row count and normalisation for each equation.

The R13 values use the canonical unbordered residual on $60^2$ at both Knudsen numbers. For R26 the interior residuals are solver-scaled on $60^2$ at $\Kn=0.05$ and unscaled for the $90^2$ continuation field at $\Kn=0.20$. At $\Kn=0.20$ the R26 interior norms thus have a different normalisation from the full-vector RMS in table~\ref{tab:residual_measures}.

\subsection{Sensitivity to the R26 closure and initial field}
\label{app:eq25_resolution}
The earlier $\Kn=0.05$ calculation used $8m_{ijk}\partial_lu_l$ in the $\psi$ closure and $\tfrac67R_{ij}\partial_ku_k$ in the $R_{ij}$ equation. The present calculations at both Knudsen numbers use the four-index projection $\mathcal T_{ijk}$ defined above in the first term and the three-index bracket $\tfrac67\mathcal R_{ij}$ of equation (20) of \citet{GuEmerson2009} in the second. The convergence history in figure~\ref{fig:convergence_earliest_new} at $\Kn=0.05$ pertains to the earlier contractions.

To isolate the effect of the contractions, we evaluated each form on the same earlier field. The original operator reproduced the earlier residual vector with an RMS difference of $2.234\times10^{-14}$. Replacing only the contraction in the $\psi$ closure raised the scaled RMS from $3.59402\times10^{-4}$ to $1.53984\times10^{-2}$. Replacing both raised it to $1.59139\times10^{-2}$, with 99\% of the squared residual in the two $4\times4$ blocks adjoining the lid corners. After grid sequencing with the revised closures, the scaled RMS fell to $1.38121\times10^{-5}$. The same blocks account for 93\% of the squared residual, and the RMS outside them is $3.6\times10^{-6}$. The interior rows in table~\ref{tab:equation_rms}, which include some corner-block nodes, account for 11\% of the total. The largest entry belongs to an extrapolation equation at the left-wall node just below the lid.

On the comparison mask of table~\ref{tab:current_fields}, the relative $L^2$ differences between the grid-sequenced and earlier fields are $7.9\times10^{-3}$ for $\rho-1$, $3.4\times10^{-2}$ for $T-1$, $6.3\times10^{-3}$ for velocity and $6.8\times10^{-2}$ for heat flux. Each difference is relative to the corresponding norm of the earlier field. Over the full cavity these differences rise to $6.7\times10^{-2}$, $1.2\times10^{-1}$, $1.5\times10^{-2}$ and $3.5\times10^{-1}$, respectively. The largest absolute density change, $2.5\times10^{-1}$, occurs next to the upper-left lid corner.

On the same mask, the following relative $L^2$ errors against DSMC show the effect of the closure and initialisation at $\Kn=0.05$ (table~\ref{tab:current_fields}).
\begin{center}\small
\begin{tabular}{lrrrr}\toprule
R26 field on $60^2$ & $E_{\rho-1}$ & $E_{T-1}$ & $E_u$ & $E_q$\\\midrule
Earlier closure & 0.0203 & 0.0432 & 0.0188 & 0.1498\\
Present closure, grid sequencing & 0.0200 & 0.0379 & 0.0149 & 0.1316\\
Present closure, direct initialisation & 0.0207 & 0.0436 & 0.0187 & 0.1500\\\bottomrule
\end{tabular}
\end{center}
For the two fields with the present closure, the relative $L^2$ differences on the mask are $8\times10^{-3}$ in $\rho-1$, $3\times10^{-2}$ in $T-1$, $6\times10^{-3}$ in velocity and $7\times10^{-2}$ in heat flux, using the grid-sequenced field as the reference. Their errors against DSMC differ by 3--26\% relative to those of the grid-sequenced field. We use the grid-sequenced field for the main $\Kn=0.05$ comparison because it has the smaller residual and mass defect. Direct initialisation gives a mass defect of $3.4\times10^{-5}$.

\subsection{Wall conditions and integral balances}
We evaluated the eleven smooth-wall relations at every non-corner wall node using impermeability and an independent transcription of equations (32), (33) and (C1)--(C8) of \citet{GuEmerson2009}. Both evaluations used the same closure inputs. The RMS differences between the two evaluations were $2.142\times10^{-16}$ for the $\Kn=0.05$ field and $2.353\times10^{-16}$ and $2.073\times10^{-16}$ for the continuation and grid-sequenced fields at $\Kn=0.20$, respectively.

We also integrated the wall traction and heat flux using the dual-cell weights at the walls. Contributions from shared interior faces cancel in the sum. In the nondimensional equations the internal-energy defect is $3/2$ times the integrated temperature-equation residual. Before this conversion the integrated balances agree with the weighted equation residuals to within $5\times10^{-17}$.

\begin{center}\footnotesize\setlength{\tabcolsep}{4pt}
\begin{tabular}{lrrr}\toprule
 & $\Kn=0.05$ & \multicolumn{2}{c}{$\Kn=0.20$}\\
 & Grid sequencing & Continuation & Grid sequencing\\\midrule
Absolute mass defect & $1.08998\times10^{-5}$ & $4.49911\times10^{-11}$ & $7.91013\times10^{-4}$\\
Absolute omitted continuity residual & $2.90791\times10^{-5}$ & $1.64745\times10^{-12}$ & $2.37272\times10^{-4}$\\
Wall normal-velocity RMS & $3.02960\times10^{-5}$ & $0$ & $2.50276\times10^{-4}$\\
Wall-equation RMS & $2.25397\times10^{-5}$ & $2.26530\times10^{-3}$ & $1.49477\times10^{-4}$\\
Absolute internal-energy defect & $3.53273\times10^{-7}$ & $2.41512\times10^{-10}$ & $1.88744\times10^{-6}$\\
Signed $x$-momentum defect & $-4.19447\times10^{-9}$ & $1.00528\times10^{-9}$ & $3.83560\times10^{-9}$\\
Signed $y$-momentum defect & $1.91867\times10^{-8}$ & $1.44431\times10^{-9}$ & $-4.84194\times10^{-8}$\\\bottomrule
\end{tabular}\end{center}
The omitted continuity residual is evaluated at the node whose equation is replaced by the total-mass constraint. At $\Kn=0.05$ the Levenberg--Marquardt iteration minimises the total scaled residual and enforces the mass equation only to the accuracy shown. For comparison, the earlier field had zero mass defect, an omitted continuity residual of $9.4\times10^{-17}$, an internal-energy defect of $5.6\times10^{-7}$ and a wall normal-velocity RMS of $1.5\times10^{-4}$. At $\Kn=0.20$ the grid-sequenced field has a wall-equation RMS fifteen times smaller than that of the continuation field, but larger mass, continuity and energy defects and a finite wall normal velocity. Appendix~\ref{app:complete_continuation} compares the two iterations.

%% file: paired_bulk_coupling_20261020.tex
\subsection{\texorpdfstring{\colorbox{yellow!45}{\parbox{0.75\linewidth}{Fixed-state bulk coupling and paired derivative sensitivity}}}{Fixed-state bulk coupling and paired derivative sensitivity}}
\label{app:paired_bulk}

\begin{shaded}
A divergence-free fourth-order perturbation can still couple to the velocity gradient. To measure the resulting forcing, we hold the retained moments fixed and replace the closure moments by their DSMC values. Define $e_\Phi=\Phi_{\rm DSMC}-\Phi_{\rm model}$, $e_\psi=\psi_{\rm DSMC}-\psi_{\rm model}$ and $e_\Omega=\Omega_{\rm DSMC}-\Omega_{\rm model}$. Taking the left-hand side minus the right-hand side of each bulk equation gives
\begin{align}
 (C_m[e])_{ijk}&=\partial_l(e_\Phi)_{ijkl},\label{eq:paired_bulk_m}\\
 (C_R[e])_{ij}&=\partial_k(e_\psi)_{ijk}
       +2(e_\Phi)_{ijkl}\partial_l u_k
       +\frac25\operatorname{STF}_2\!\left(\partial_i(e_\Omega)_j\right),\label{eq:paired_bulk_R}\\
 C_\Delta[e]&=\partial_i(e_\Omega)_i.\label{eq:paired_bulk_Delta}
\end{align}
Here $\operatorname{STF}_2(B)=(B+B^{\mathsf T})/2-\operatorname{tr}(B)I/3$, and the repeated indices span three velocity components. The last two terms in \eqref{eq:paired_bulk_R} have positive signs because the corresponding source terms are subtracted. These expressions isolate the change due to the closure inputs at a prescribed field. The complete discrete R26 residual also depends on the discretisation and the boundary equations.
\end{shaded}

\begin{shaded}
We pool 20,000 samples from each of seven independent DSMC realisations between two and four million steps. From the 140,000 samples we form raw velocity moments before converting them to central moments. H1 and H2 evaluate the inner derivative with one and two native grid spacings, respectively. Both use the same outer three-point derivative on the $160\times160$ cell centres, with $L=10^{-6}\,\mathrm m$. We differentiate the full unfiltered fields before omitting two cell layers on each side, leaving $156^2=24,336$ cells. Spatial derivatives in $z$ vanish by homogeneity. All three velocity components and all measured Cartesian moments, including terms odd in $z$, enter the calculation.
\end{shaded}

\begin{shaded}
We express the RMS forcing in each moment equation $a$ in terms of its local density and thermal speed,
\begin{equation}
 N_a[e]=\left[\frac1{|\mathcal I|}\sum_{x\in\mathcal I}
 \left\|\frac{C_a[e](x)}{s_a(x)}\right\|_F^2\right]^{1/2},
 \qquad s_m=\frac{\rho\theta^2}{L},\quad
 s_R=s_\Delta=\frac{\rho\theta^{5/2}}{L}.
 \label{eq:paired_bulk_natural}
\end{equation}
The density $\rho$ and central velocity variance $\theta$ (in $\mathrm{m^2\,s^{-2}}$) are measured in the same DSMC sample. The Frobenius norms include all 27 components for $m$, nine for $R$ and one for $\Delta$. Each RMS is dimensionless; the dimensional forcings in the three equations have different units.
\end{shaded}

\begin{center}
\begin{minipage}{\linewidth}
\centering
\rowcolors{1}{yellow!40}{yellow!40}
\begin{tabular}{lrrr}
\toprule
Sector & $N_a[e_{H1}]$ & $N_a[e_{H2}]$ & $N_a[e_{H2}-e_{H1}]$ \\
\midrule
$m$ & 6.84686681 & 4.90721202 & 4.79771703 \\
$R$ & 18.50582761 & 8.95575370 & 18.26450988 \\
$\Delta$ & 24.45513951 & 9.15299703 & 25.60526367 \\
\bottomrule
\end{tabular}
\captionsetup{format=revisioncaption}
\captionof{table}{Change in the bulk moment equations when DSMC closure moments replace the model values at the same field. H1 and H2 differ only in the inner derivative spacing. The last column gives the RMS of the pointwise difference between their forcing fields.}
\label{tab:paired_bulk_forcing}
\end{minipage}
\end{center}

\begin{shaded}
An independent expansion of the Cartesian contractions, using explicit three-point difference coefficients and compensated summation, agreed with the tensor calculation to a maximum relative difference of $1.42\times10^{-14}$ under the stated comparison rule.
\end{shaded}

\begin{shaded}
Doubling the spacing of the inner derivative reduces the RMS forcing in all three equations (table~\ref{tab:paired_bulk_forcing}). For $R$ and $\Delta$, however, the RMS change in the forcing field is nearly as large as, or larger than, the H1 forcing itself. The choice of derivative spacing therefore matters even when the outer derivative and the DSMC sample are held fixed. Both estimates still include sampling fluctuations and differences in the collision and closure models. Separating those effects requires matched realisations and time blocks.
\end{shaded}

%% file: finite_wall_coupling_addendum_20261020.tex
\subsection{\texorpdfstring{\colorbox{yellow!45}{\parbox{0.75\linewidth}{Effect of finite closure changes on slip and temperature}}}{Effect of finite closure changes on slip and temperature}}
\label{app:finite_wall_coupling}

\begin{shaded}
The closure inputs enter two smooth-wall relations in the implementation of \citet{GuEmerson2009}. Fix all
retained wall quantities, the geometric frame and wall parameters, and write
$\vartheta=R_{\rm gas}T>0$, $H=\sqrt{\pi\vartheta/2}$ and
$A=(2-\chi)/\chi$, with accommodation $\chi$ fixed. The effective pressure is
\begin{equation}
 P=p+\frac{\sigma_{nn}}2-\frac{30R_{nn}+7\Delta}{840\vartheta}
       -\frac{\Phi_{nnnn}}{24\vartheta}.
\end{equation}
The slip and temperature right-hand sides are $N_s/P$ and
$N_T/P+u_t^2/4$, respectively, where
\begin{align}
 N_s&=-AH\sigma_{nt}-\frac{5m_{nnt}+2q_t}{10}
                +\frac{9\Omega_t+70\psi_{nnt}}{2520\vartheta},\\
 N_T&=-\frac{AHq_n}{2}-\frac{\vartheta\sigma_{nn}}4
                -\frac{75R_{nn}+28\Delta}{840}
                +\frac{\Phi_{nnnn}}{24}.
\end{align}
We vary the closure inputs in these expressions while holding the retained moments fixed.
\end{shaded}

\begin{shaded}
Changing the numerator by $a$ and the pressure by $b$ gives the exact identity
\begin{equation}
 \frac{N+a}{P+b}-\frac NP=\frac{Pa-Nb}{P(P+b)}.
 \label{eq:finite_wall_ratio}
\end{equation}
Here $b=-\delta\Phi_{nnnn}/(24\vartheta)$,
$a_s=(9\delta\Omega_t+70\delta\psi_{nnt})/(2520\vartheta)$ and
$a_T=\delta\Phi_{nnnn}/24$. The temperature $\vartheta$ and both pressures $P$ and $P+b$ must satisfy the positivity bounds. In each wall equation the residual is the left-hand side minus the right-hand side. At fixed retained moments, its change is therefore the negative of \eqref{eq:finite_wall_ratio} divided by the row scale. The slip and temperature equations follow the impermeability condition in the set of eleven wall equations. The finite changes can differ from their linearisations. For example,
\begin{equation}
 \vartheta P+N_T=
 \vartheta p+\frac{\vartheta\sigma_{nn}}4
 -\frac{R_{nn}}8-\frac{\Delta}{24}-\frac{AHq_n}{2},
\end{equation}
so the explicit $\Phi_{nnnn}$ terms cancel. If this combination vanishes, the temperature ratio can remain constant as the closure changes. Thus the dependence on the closure alone does not establish injectivity.
\end{shaded}

\begin{shaded}
At an impermeable wall the conservative finite-volume equations use the Cartesian fourth-moment fluxes $J_\Delta^x=\Omega_x$ and $J_\Delta^y=\Omega_y$, evaluated from the moments and closures. Together with the numerical fluxes at interior faces, they determine the discrete divergence in cells adjacent to the wall. The eleven algebraic wall conditions, six extrapolation conditions, corner equations and global mass constraint enter separately. The fixed-state calculation in appendix~\ref{app:paired_bulk} instead uses the same three-point derivative for both closures.
\end{shaded}

%% file: conditional_observation_azz_wall_addendum_20261021.tex
\subsection{\texorpdfstring{\colorbox{yellow!45}{\parbox{0.75\linewidth}{Tensors with the same planar divergence and the wall response to \(A_{zz}\)}}}{Tensors with the same planar divergence and the wall response to Azz}}
\label{app:conditional_observation_wall}
\begin{shaded}
Let $\mathcal O(A)=(\partial_xA_{xx}+\partial_yA_{xy},
\partial_xA_{xy}+\partial_yA_{yy})$ denote only the selected planar
fourth-order divergence. For $\chi,b\in C_c^\infty(\mathcal D)$,
\begin{equation}
 K_{\chi,b}=\begin{pmatrix}
 \partial_{yy}\chi&-\partial_{xy}\chi&0\\
 -\partial_{xy}\chi&\partial_{xx}\chi&0\\
 0&0&b
 \end{pmatrix},\qquad \mathcal O(K_{\chi,b})=0.
 \label{eq:conditional_observation_kernel}
\end{equation}
Mixed derivatives commute, so $A$ and $A+K_{\chi,b}$ give the same planar divergence. If $K_{\chi,b}\ne0$ and both tensors belong to the class under consideration, the divergence cannot distinguish them. Identifying either tensor with a solution of the cavity problem also requires kinetic realisability and satisfaction of the other balances and the wall and corner conditions.
\end{shaded}

\begin{shaded}
To determine how $A_{zz}$ enters the wall equations, set $\chi=0$ and allow the smooth perturbation $b$ to have a boundary trace. When $b$ has compact support in the interior, it and all its derivatives vanish at the wall. At finite grid spacing the discrete wall response is zero if $b$ vanishes throughout the wall stencil. Hold $\rho$, $\vartheta=R_{\rm gas}T$, velocity, stress, heat flux and third moments fixed. From $A=R^{\cl}+\Delta I/3$ we obtain
\begin{equation}
 \delta\Delta=b,\quad
 \delta R^{\cl}_{xx}=\delta R^{\cl}_{yy}=-b/3,\quad
 \delta R^{\cl}_{zz}=2b/3,\quad \delta R^{\cl}_{xy}=0.
 \label{eq:azz_retained_path}
\end{equation}
The implementation obtains the last diagonal component from the trace-free condition.
For every in-plane unit wall normal, $\delta R_{nn}=-b/3$.
Consequently the effective pressure in appendix~\ref{app:finite_wall_coupling}
changes by the exact identity
\begin{equation}
 \delta P=\frac{b}{280\vartheta}
           -\frac{\delta\Phi_{nnnn}}{24\vartheta}.
 \label{eq:azz_wall_pressure}
\end{equation}
Here $\delta\Phi$, $\delta\psi$ and $\delta\Omega$ include the changes in the closures due to both the moments and their gradients. With the numerators $N_s,N_T$ defined in appendix~\ref{app:finite_wall_coupling},
$\delta N_s=(9\delta\Omega_t+70\delta\psi_{nnt})/(2520\vartheta)$ and
$\delta N_T=-b/280+\delta\Phi_{nnnn}/24=-\vartheta\delta P$. Hence
\begin{align}
 \delta S_{\rm rhs}&=
 \frac{P\,\delta N_s-N_s\delta P}{P(P+\delta P)},\\
 \delta T_{\rm rhs}&=
 -\frac{(\vartheta P+N_T)\delta P}{P(P+\delta P)}.
 \label{eq:azz_finite_wall_response}
\end{align}
These expressions apply where the temperature and both effective pressures satisfy the positivity bounds. The wall residual is the measured left-hand side minus the prescribed right-hand side, so its change has the opposite sign, with the same row scaling. Either numerator can vanish. In the finite-volume equations the perturbed $\Phi,\psi,\Omega$ also enter the face fluxes. Thus a perturbation invisible to the planar divergence can still affect the boundary equations; membership in the null space of the full system requires those equations as well.
\end{shaded}

\begin{shaded}
As a separate test of the boundary implementation, we evaluated a perturbation of the 17-component state on a $7\times7$ grid. Directional comparisons account for 813 of the 833 equations. The other 20 enforce no penetration and give zero by construction. Independent arithmetic reproduced these results, including the 16 intermediate zeros introduced by the mask. This test uses a different perturbation from \eqref{eq:conditional_observation_kernel}.
\end{shaded}

%% file: native_n60_value_agreement_addendum_20261021.tex
\begin{shaded}
\noindent\textbf{Verification of the discrete equations at $\Kn_{\rm Gu}=0.05$.}
We assembled the residual independently at one intermediate R26 field on the stretched $60^2$ grid ($\beta=1.25$). Both evaluations used the earlier contractions described in appendix~\ref{app:eq25_resolution}. The comparison covers 61,200 equations: 57,187 interior equations, 2,552 physical wall conditions, 1,392 extrapolation conditions, 68 corner conditions and one mass constraint. The largest absolute difference between the two scaled residual vectors was $9.020563\times10^{-15}$. Before replacing the wall and mass equations, we also verified cancellation of the internal face fluxes for eight transport terms. The separate assembly tests the finite-volume and boundary discretisation at this field; both evaluations use the same numerical primitives for the underlying arithmetic.
\end{shaded}

%% file: r26_wall_residuals.tex
\subsection{\texorpdfstring{\colorbox{yellow!45}{\parbox{0.75\linewidth}{Wall and interior contributions to the R26 residual}}}{Wall and interior contributions to the R26 residual}}
\label{app:wall_residual_rows}
\begin{shaded}
Where is the residual largest after a further step of the R26 continuation? We separate the 137,700 equations on the $90^2$ grid at $\Kn_{\rm Gu}=0.20$ into the groups below. Each wall has 88 non-corner nodes with eleven physical and six extrapolation conditions per node. The global mass constraint replaces one interior continuity equation. The four corners account for 68 more equations.
\end{shaded}
\par\ifdim\dimexpr\pagegoal-\pagetotal\relax<14\baselineskip\newpage\fi
\begin{shaded}
\begin{center}
\begin{tabular}{lrr}
\hline
Partition  &  Rows  &  Scaled RMS\\
\hline
Interior  &  131,647  &  $2.42988268\times10^{-4}$\\
Physical smooth-wall  &  3,872  &  $0.00227517903$\\
Free extrapolation  &  2,112  &  $0.00269622618$\\
Corner completion  &  68  &  $0.00215687424$\\
Global mass  &  1  &  $0.00350410691$\\
\hline
Total  &  137,700  &  $0.000562041484$\\
\hline
\end{tabular}
\end{center}
The smooth-wall conditions and the extrapolation equations together account for $81.375\%$ of the sum of squared, scaled residuals. The corner equations are counted separately. This partition describes the numerical residual under the stated scaling.
\end{shaded}
\begin{shaded}
The residual vector used to linearise the equations agrees entry by entry with the independent evaluation at this state. The signed mass residual is $-0.00350410691$, or $-0.350410691\%$ of the reference mean density, within the $1\%$ mass tolerance. Evaluated separately, the continuity equation at the node used for the mass constraint gives $-2.49886917\times10^{-15}$. The field comparisons use the preceding continuation state; the additional step is used here to locate the residual.
\end{shaded}

%% file: entropy_methods_and_comparison.tex
\subsection{R13, R26 and DSMC entropy estimators}
We evaluate the quadratic entropy of \citet{RanaGuptaSprittlesTorrilhon2021} (equation (25) and appendix A) about a fixed Gaussian reference. Its first five moment families give the R13 expression
\begin{equation}
\eta_{13}-1=-\tfrac12(\rho-1)^2-\tfrac12u_i u_i-\tfrac34(\theta-1)^2
-\tfrac14\sigma_{ij}\sigma_{ij}-\tfrac15q_iq_i.
\end{equation}
R26 adds $-m_{ijk}m_{ijk}/12-R_{ij}R_{ij}/56-\Delta^2/240$. In R13 these regularising moments are algebraic functions of the retained fields; its entropy uses only the 13 retained moments. We use the same Gaussian reference for both models. For a planar trace-free tensor,
\begin{equation}
\sigma:\sigma=\sigma_{xx}^2+\sigma_{yy}^2+(\sigma_{xx}+\sigma_{yy})^2+2\sigma_{xy}^2,
\end{equation}
and similarly for $R$; for rank three,
\begin{equation}
m:m=m_{xxx}^2+3m_{xxy}^2+3m_{xyy}^2+m_{yyy}^2
+3(m_{xxx}+m_{xyy})^2+3(m_{xxy}+m_{yyy})^2.
\end{equation}
Thus dropping the implied out-of-plane components changes the diagnostic even in a planar flow.

For DSMC we estimate kinetic entropy from a velocity histogram. Let $p_b=N_b/\sum_bN_b$ be the fraction of sampled particles in bin $b$, and let
\begin{equation}
g_b=\int_{B_b}(2\pi\theta_0)^{-3/2}\exp[-|\boldsymbol c|^2/(2\theta_0)]\,\dd^3c
\end{equation}
be the integrated probability of the same bin under the fixed reference Maxwellian, with an outside-domain bin included. With $r=n/n_0$,
\begin{equation}
D_B=\sum_b p_b\log(p_b/g_b),\qquad
\eta_B=r[1-\log r-D_B],\qquad 0\log0=0.
\label{eq:discrete_entropy_new}
\end{equation}
Equation~\eqref{eq:discrete_entropy_new} is the discrete relative entropy with respect to the reference Maxwellian. Integrating the Maxwellian over each bin preserves its probability mass. Coarser bins discard information and can reduce $D_B$, even with arbitrarily long sampling.

The entropy sample comprises seven independent realisations with 4,000 frames each, giving 28,000 frames and $183{,}500{,}800{,}000$ particle observations. The cube extends to $\pm8c_0$ in each velocity component and contains all sampled particles. We subdivide it into 16, 8 and 4 bins per axis. The primary-field figures use a longer sampling interval for which only the 56 velocity moments, rather than the histogram needed for \eqref{eq:discrete_entropy_new}, are available.

For comparison, an absolute per-particle finite-bin estimator would be
\begin{equation}
\frac{s_B}{k_B}=1-\sum_b p_b\log p_b+\log(\Delta v^3)+3\log(m/h)-\log n.
\end{equation}
This absolute entropy has a different phase-space convention and different units from $\eta_{13}$ and $\eta_{26}$. We therefore compare the three estimators relative to the same Gaussian. Their differences include moment truncation, velocity binning, spatial discretisation and the choice of model fields.

\begin{figure}[p]
\centering
\includegraphics[width=\linewidth]{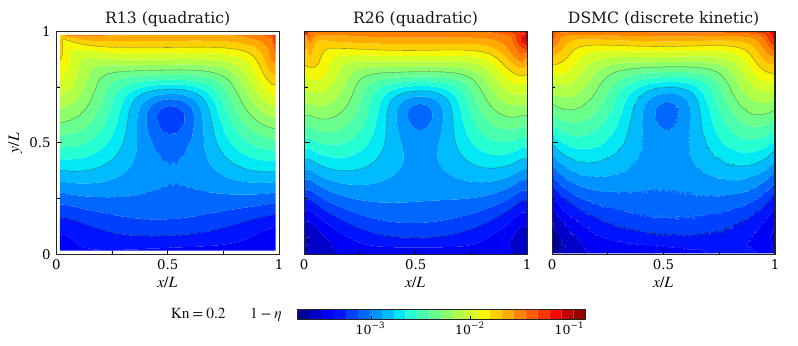}
\caption{Entropy at $\Kn=0.2$ relative to the same Gaussian: the quadratic R13 and R26 estimates and the DSMC estimate from the velocity histogram. The colour scale gives $1-\eta$ logarithmically so that both the weak interior signal and the lid region are visible. Each field is shown on its own grid. The DSMC histogram uses 4,000 frames per realisation.}
\label{fig:entropy_three_new}
\end{figure}
\begin{figure}[p]
\centering
\includegraphics[width=\linewidth]{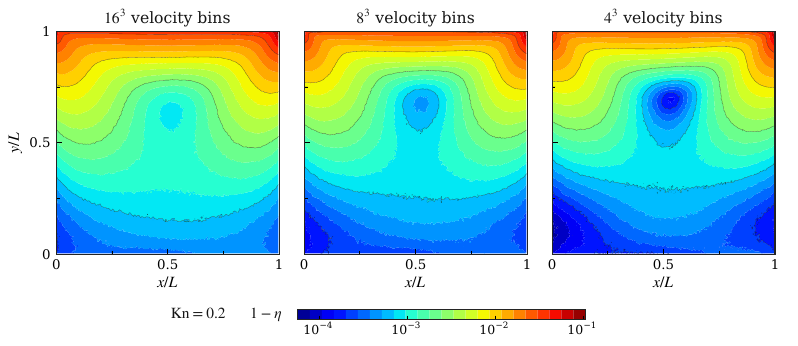}
\caption{DSMC entropy with 16, 8 and 4 velocity bins per axis. All three panels use the same particle sample and logarithmic scale for $1-\eta$. The difference between panels measures the effect of velocity-bin width.}
\label{fig:entropy_bins_new}
\end{figure}

%% file: equation_rms_table.tex
\begin{center}
\begin{minipage}{\linewidth}
\centering\small
\rowcolors{1}{yellow!40}{yellow!40}
\begin{tabular}{lrrrr}
\toprule
Equation & R13, $\Kn_{\rm Gu}=0.05$ & R13, $\Kn_{\rm Gu}=0.20$ & R26, $\Kn_{\rm Gu}=0.05$ & R26, $\Kn_{\rm Gu}=0.20$\\
\midrule
$\rho$ (continuity) & $9.9\times10^{-16}$ & $4.9\times10^{-12}$ & $3.5\times10^{-7}$ & $6.7\times10^{-4}$\\
$u_x$ (momentum) & $7.2\times10^{-13}$ & $8.1\times10^{-11}$ & $8.4\times10^{-7}$ & $9.1\times10^{-5}$\\
$u_y$ (momentum) & $2.8\times10^{-13}$ & $5.5\times10^{-11}$ & $6.9\times10^{-7}$ & $7.8\times10^{-5}$\\
$T$ (energy) & $2.7\times10^{-13}$ & $5.4\times10^{-11}$ & $1.4\times10^{-6}$ & $1.3\times10^{-4}$\\
$q_x$ & $6.4\times10^{-13}$ & $9.8\times10^{-11}$ & $1.0\times10^{-6}$ & $1.2\times10^{-3}$\\
$q_y$ & $4.8\times10^{-13}$ & $7.5\times10^{-11}$ & $1.0\times10^{-6}$ & $1.2\times10^{-3}$\\
$\sigma_{xx}$ & $3.7\times10^{-13}$ & $4.7\times10^{-11}$ & $8.3\times10^{-7}$ & $1.8\times10^{-3}$\\
$\sigma_{xy}$ & $5.2\times10^{-13}$ & $7.4\times10^{-11}$ & $1.2\times10^{-6}$ & $1.9\times10^{-3}$\\
$\sigma_{yy}$ & $2.5\times10^{-13}$ & $3.9\times10^{-11}$ & $7.9\times10^{-7}$ & $1.8\times10^{-3}$\\
$R_{xx}$ & $2.6\times10^{-13}$ & $2.3\times10^{-11}$ & $1.5\times10^{-6}$ & $2.7\times10^{-4}$\\
$R_{xy}$ & $1.3\times10^{-13}$ & $1.6\times10^{-11}$ & $2.0\times10^{-6}$ & $6.1\times10^{-4}$\\
$R_{yy}$ & $1.5\times10^{-13}$ & $1.5\times10^{-11}$ & $3.2\times10^{-6}$ & $4.4\times10^{-4}$\\
$m_{xxx}$ & $1.8\times10^{-13}$ & $3.7\times10^{-11}$ & $8.3\times10^{-7}$ & $9.3\times10^{-4}$\\
$m_{xxy}$ & $1.6\times10^{-13}$ & $1.8\times10^{-11}$ & $7.3\times10^{-7}$ & $9.1\times10^{-4}$\\
$m_{xyy}$ & $2.1\times10^{-13}$ & $4.6\times10^{-11}$ & $7.2\times10^{-7}$ & $9.9\times10^{-4}$\\
$m_{yyy}$ & $1.4\times10^{-13}$ & $1.6\times10^{-11}$ & $6.8\times10^{-7}$ & $7.0\times10^{-4}$\\
$\Delta$ & $4.6\times10^{-12}$ & $3.0\times10^{-10}$ & $1.9\times10^{-5}$ & $1.6\times10^{-4}$\\
\midrule
Interior rows & 3364 & 3364 & 3364 & 7744\\
Residual basis & canonical & canonical & solver-scaled & raw\\
\bottomrule
\end{tabular}
\captionsetup{format=revisioncaption}
\captionof{table}{Interior RMS residual by equation. The R13 columns use the unbordered residual $A(U)U-b(U)$ in native units on $60^2$ grids; their final eight rows are algebraic closure relations. The R26 columns use scaled residuals on the stretched $60^2$ grid at $\Kn_{\rm Gu}=0.05$ and unscaled residuals on the uniform $90^2$ grid at $\Kn_{\rm Gu}=0.20$ (appendix~\ref{app:eq25_resolution}). Each R26 continuity entry excludes the node used to impose total mass. The normalisation in each column applies to that model and grid.}
\label{tab:equation_rms}
\end{minipage}
\end{center}